\documentclass[11pt]{article}
\usepackage{fullpage,graphicx,amsmath,amssymb,verbatim,bm,xcolor}
\allowdisplaybreaks
\newtheorem{definition}{Definition}
\newtheorem{lemma}{Lemma}
\newtheorem{proposition}{Proposition}
\newtheorem{fact}{Fact}
\newtheorem{theorem}{Theorem}
\newtheorem{corollary}{Corollary}

\newenvironment{proof}{\textbf{Proof:}}{\hfill$\square$}

\newcommand{\cE}{\mathcal{E}}

\newcommand{\cN}{\mathcal{N}}
\newcommand{\cO}{\mathcal{O}}

\newcommand{\cT}{\mathcal{T}}
\newcommand{\cX}{\mathcal{X}}
\newcommand{\cY}{\mathcal{Y}}

\newcommand{\CMG}{\mathrm{CMG}}

\newcommand{\hA}{\hat{A}}
\newcommand{\hB}{\hat{B}}

\newcommand{\hR}{\hat{R}}

\newcommand{\htrho}{\hat{\tilde{\rho}}}
\newcommand{\httau}{\hat{\tilde{\tau}}}

\newcommand{\hbtrho}{\hat{\bar{\tilde{\rho}}}}

\newcommand{\tcT}{\tilde{\mathcal{T}}}

\newcommand{\htcT}{\hat{\tilde{\mathcal{T}}}}

\newcommand{\trho}{\tilde{\rho}}
\newcommand{\tsigma}{\tilde{\sigma}}
\newcommand{\ttau}{\tilde{\tau}}

\newcommand{\btrho}{\bar{\tilde{\rho}}}

\newcommand{\bsigma}{\bar{\sigma}}

\newcommand{\supp}{\mathrm{supp}}

\DeclareMathOperator*{\E}{{\rm {\mathbb E}}\,}
\DeclareMathOperator*{\Tr}{{\rm Tr}\;}

\newcommand{\zero}{\leavevmode\hbox{\small l\kern-3.5pt\normalsize0}}
\newcommand{\one}{\leavevmode\hbox{\small1\kern-3.8pt\normalsize1}}
\newcommand{\I}{\mathbb{I}}

\newcommand{\ket}[1]{| #1 \rangle}
\newcommand{\bra}[1]{\langle #1 |}
\newcommand{\ketbra}[1]{\ket{#1}\bra{#1}}

\begin{document}

\title{{\bf Fully smooth one shot multipartite covering and decoupling 
of quantum states via telescoping
}}

\author{Pranab Sen\footnote{
School of Technology and Computer Science, Tata Institute of Fundamental
Research, Mumbai, India. Most of this work was done 
while the author was on sabbatical leave
at the Centre for Quantum Technologies, National University of Singapore.
Email: {\sf pranab.sen.73@gmail.com}.
}
}

\date{}

\maketitle

\begin{abstract}
We prove fully smooth one shot 
multipartite covering, aka convex split, results as well
as fully smooth multipartite decoupling results for quantum states.
Fully smooth one shot results for 
these problems were not known earlier, though the works of
Cheng, Gao and Berta \cite{Cheng:convexsplit} for convex split,
and Colomer and Winter \cite{Colomer:decoupling} for decoupling,
had made substantial progress
by introducing a technique called telescoping cum mean zero
decomposition of quantum states.
We show that the telescoping cum mean zero decomposition 
technique can in fact be simplified and further extended in order
to prove fully smooth decoupling and convex split results.

Our techniques allow us to prove the first fully smooth one shot inner
bounds for various fundamental network quantum information theory problems
like e.g. the generalised Slepian Wolf problem of \cite{anshu:slepianwolf}.
We can also prove for the first time the natural polyhedral inner bound 
for sending quantum information
over a quantum multiple access channel with limited entanglement 
assistance, first conjectured in \cite{Chakraborty:simultaneous}.
\end{abstract}






\section{Introduction}
The simultaneous smoothing bottleneck is a famous open problem in 
network quantum
information theory \cite{drescher:simultaneous}. A positive result for 
this problem would imply 
that many union and intersection type arguments carried out, sometimes
implicitly, in network classical information theory extend similarly
to the quantum setting. Most tasks in information theory can be decomposed
into simpler tasks of one of two types: {\em packing} and 
{\em covering}. The earlier work of \cite{sen:oneshot} gave a machinery
for implementing union and intersection for packing tasks in one shot 
network quantum information theory, bypassing simulaneous smoothing.
However it left open the question of implementing union and intersection
for covering tasks, a lacuna which was explicitly pointed out in
\cite{ding:relay}.

Recent exciting work of Cheng, Gao and 
Berta~\cite{Cheng:convexsplit}, and Colomer and 
Winter~\cite{Colomer:decoupling} has introduced a telescoping 
cum mean zero decomposition technique, that bypasses the simultaneous
smoothing bottleneck for intersection arguments for two fundamental 
problems in quantum information theory viz. multipartite covering aka
multipartite convex split, and multipartite decoupling respectively.
However Cheng et al. do not state their
multipartite convex split
results in terms of smooth one shot quantities, leaving several
basic one shot and finite blocklength achievability questions 
in network quantum information theory, e.g. inner bounds for the 
generalised one shot quantum Slepian Wolf \cite{anshu:slepianwolf} problem,
unanswered. The first convex split lemma was actually a
`unipartite' statement \cite{Jain:convexsplit}. Various smooth versions
of the unipartite convex split lemma were proved, and used to obtain
inner bounds for quantum communication tasks like
state merging, state redistribution \cite{Jain:convexsplit},
sending private classical information over a quantum wiretap channel
\cite{Wilde:wiretap} etc. The first general multipartite convex
split lemma was proved by Anshu, Jain and Warsi \cite{anshu:slepianwolf},
but it was a non-smooth result. Anshu, Jain and Warsi used their
non-smooth multipartite lemma to obtain non-smooth one shot inner 
bounds for a generalised quantum Slepian-Wolf problem. They also managed
to obtain the desired asymptotic iid limit for a restricted version of
the generalised Slepian-Wolf problem by ad hoc techniques, but left 
the question of existence
of a fully smooth one shot multipartite convex split lemma as well
as the achievability of the desired asymptotic iid limit for the
generalised Slepian-Wolf problem open.

Colomer and Winter \cite{Colomer:decoupling} state their 
multipartite decoupling
theorem as a `half smooth' one shot generalisation of the non smooth
multipartite decoupling result of \cite{Chakraborty:simultaneous}. 
By half smooth, we mean that their decoupling theorem uses smooth
one shot conditional entropy for the `control state' that needs to be
decoupled but uses the non-smooth conditional entropy for the Choi state
of the superoperator involved in the decoupling theorem. Their half smooth
decoupling theorem suffices to obtain the `natural' polyhedral rate
region for sending quantum information over an entanglement unassisted
quantum multiple access
channel (QMAC) \cite{Colomer:decoupling}. However, it cannot be used to
obtain the `natural' polyhedral rate region for sending 
quantum information
over a QMAC with limited entanglement assistance. This is because
obtaining the polyhedral rate region in this setting requires a different
choice of control state and and a different choice of superoperator in 
the decoupling argument
\cite{Chakraborty:simultaneous}, which ultimately goes back to Dupuis' 
achievability result for sending quantum information over a point to point
channel under limited entanglement assistance \cite{decoupling}. For
these different choices, one would need to prove a different half smooth
decoupling theorem that uses smooth conditional entropy for the Choi
state! Earlier works like
\cite{Chakraborty:ratesplitting} have studied this problem and have
obtained subpolyhedral inner bounds in the one shot setting which 
nevertheless lead on to the full polyhedral region in the asymptotic
limit of many independent and identically distributed (iid) uses of
the channel. However, obtaining the full polyhedral inner bound 
in the one shot setting for the QMAC with limited entanglement 
assistance remained open.

In this paper, we show that the telescoping proof technique
of Colomer and Winter \cite{Colomer:decoupling}, and Cheng, Gao and
Berta \cite{Cheng:convexsplit}, can actually be simplified and further
extended to prove fully smooth multipartite decoupling 
and convex split results. 
Thus, this paper generalises both Cheng, Gao and
Berta's \cite{Cheng:convexsplit}, as well as Colomer and 
Winter's \cite{Colomer:decoupling} works. Our fully smooth multipartite
convex split theorem allows us to obtain a
fully smooth one shot version of the generalised Slepian Wolf problem
of \cite{anshu:slepianwolf}, as well as a fully smooth one shot inner
bound for sending private classical information over a wiretap QMAC.
Both these consequences were not known earlier. 
Our fully smooth
multipartite decoupling theorem allows us to obtain the natural polyhedral
one shot rate region for sending quantum information over a QMAC with
limited entanglement assistance, which also was not known earlier.

Smooth multipartite convex split, when applied to classical quantum
states, can be viewed as a smooth multipartite classical quantum 
soft covering lemma. In multipartite soft covering, there are $k$
classical parties, $k \geq 1$. Each party independently samples a set 
of classical 
symbols from a probability distribution.  A $k$-tuple of
symbols from the samples of the $k$ parties is fed to a box whose 
output is a quantum state
depending on the input $k$-tuple. The uniform average of the 
resulting input dependent quantum state is taken
over all $k$ tuples that can be constructed from the $k$ sets sampled
by the $k$ parties. The aim is to ensure that the resulting
sample averaged input dependent quantum state is close in 
trace distance to an `ideal'
fixed quantum state, under a `typical' collection of sets of samples 
by the $k$ parties.
All known soft covering lemmas for classical quantum states, including the
smooth multipartite one arising from convex split, fall in the 
mould of covering {\em in expectation} i.e. the expectation, 
over the random choices of the $k$ classical parties,
of the trace distance between 
the sample averaged input dependent quantum state and the fixed 
ideal state is small.
In \cite{Sen:MatrixChernoff}, for the first time a smooth
multipartite soft covering lemma for classical quantum states was proved
{\em in concentration}, extending the reach of the `in expectation'
result of this paper. More precisely, that paper shows that 
with exponentially high probability over the 
random choices of the $k$ parties,
the sample averaged input dependent quantum state is close to the 
ideal state. 

\subsection{Organisation of the paper}
Section~\ref{sec:CMGcovering} serves as a warmup, where
we will explain the 
bottleneck of intersection and simultaneous smoothing for covering tasks 
in more 
detail by taking up an example that we shall call the 
{\em CMG covering problem}. This example shall be our main motivating
running example for the convex split part of this work. We will show
as a warmup how the telescoping idea of 
Cheng, Gao and Berta~\cite{Cheng:convexsplit} can be simplified and
extended to prove a fully smooth inner bound for the CMG covering
problem. 
After that, we prove our fully smooth multipartite
convex split lemma in Section~\ref{sec:convexsplit} via telescoping.
In Section~\ref{sec:decoupling}, we will
see how the telescoping technique of Colomer and Saus 
\cite{Colomer:decoupling} can be simplied and further extended in order
to prove a fully smooth multipartite decoupling lemma. 
It will be followed by an
application to sending quantum information over a QMAC with limited
entanglement assistance.
We finally make some concluding remarks and list directions for
further research
in Section~\ref{sec:conclusion}.

\section{Warmup: Fully smooth CMG covering by telescoping}
\label{sec:CMGcovering}
The CMG covering problem that we are about to define lies at the heart 
of obtaining
inner bounds for private classical communication over a quantum 
wiretap interference channel. The name CMG is an
acronym for Chong, Motani and Garg, the discoverers of one of the most
well known inner bounds for the classical non-wiretap interference channel
in the asympotic independent and identically distributed (asymptotic iid)
setting \cite{CMGElGamal}. The paper \cite{sen:interference} first 
showed the achievability of
the same inner bound for sending classical information over a quantum
non-wiretap interference channel in the asymptotic iid setting. 
Subsequently, the achievability in the non-wiretap quantum case in 
the general one shot setting was proved in \cite{sen:simultaneous}.
That paper used the machinery of \cite{sen:oneshot} to bypass 
simultaneous smoothing for intersection problems arising in packing
tasks. In a companion paper \cite{Sen:flatten}, we will be able 
to extend the result of
\cite{sen:simultaneous} to the one shot wiretap quantum case by 
combining the packing arguments with the fully smooth CMG covering
lemma (Lemma~\ref{lem:smoothCMGcovering}) that we will prove below.

Let $\cN$ be a quantum channel with two inputs called Alice and Bob,
and one output called Eve. The intuition here is that Eve is an 
eavesdropper trying to get
information on the classical messages that Alice and Bob are sending 
to some other outputs of 
$\cN$, not described here explicitly as they are  not
relevant for CMG covering. Though $\cN$ has quantum inputs and quantum
outputs, we will assume without loss of generality that the inputs are
classical, because in the inner bound below, there is a standard 
optimisation
step over a choice of all ensembles of input states to Alice and Bob.
Thus, we will henceforth think of $\cN$ as a classical-quantum (cq)
channel with its two classical input alphabets being denoted by
$\cX$, $\cY$, and the quantum output alphabet by $\cE$. 

Let $0 \leq \epsilon \leq 1$. 
On input $(x,y) \in \cX \times \cY$, the channel outputs
a quantum state 
\[
\sigma^E_{xy} := \cN^{A B \rightarrow E}(\rho^A_x \otimes \rho^B_y)
\]
on $E$, where $\{\rho^A_x\}_x$, $\{\rho^B_y\}_y$ are the ensembles of the
so-called encoding quantum states at the quantum inputs $A$, $B$ 
of the channel which is
modelled by a completely positive trace non-increasing superoperator
$\cN: A B \rightarrow E$. 
To define the CMG covering problem, we need to define new alphabets
$\cX'$, $\cY'$, following the scheme of
the original paper \cite{CMGElGamal}. 
We then pick a `control' probability distribution 
\[
p(x',x,y',y) := p(x', x) p(y', y)
\]
on the classical alphabet 
$\cX' \times \cX \times \cY' \times \cY$.
The `control' cq state for the CMG covering problem is now defined as
\[
\sigma^{X' X Y' Y E} := 
\sum_{x',x,y',y} p(x',x,y',y) \ketbra{x',x,y',y}^{X'XY'Y}
     \otimes \sigma^E_{xy}.
\]
Though the state $\sigma^E_{xy}$ depends only on $x$ and $y$, for
later convenience we will write it as
$\sigma^E_{x'xy'y} := \sigma^E_{xy}$. Though this notation seems 
heavier now, it will be extremely useful when we take various marginals
of $\sigma^{X' X Y' Y E}$ in the entropic quantities and the proofs
below. Thus,
\[
\sigma^{X' X Y' Y E} := 
\sum_{x',x,y',y} p(x',x)p(y',y) \ketbra{x',x,y',y}^{X'XY'Y}
     \otimes \sigma^E_{x'xy'y}.
\]

Let us say Alice and Bob are trying to send a pair of classical messages
$a$ and $b$ respectively to Charlie. To obfuscate them from Eve, Alice
and Bob independently do the following strategy. 
Alice chooses iid samples $x'_1, \ldots, x'_{L'}$ from the
marginal distribution $p^{X'}$.
Conditioned on each sample $x'_{l'} \in \cX'$, Alice chooses iid
samples $x_{l', 1}, \ldots, x_{l', L}$ from the conditioned 
marginal distribution $p^{X} | (X' = x'_{l'}$. Thus, a total of
$L' L$ samples are chosen by Alice.
Similarly, Bob chooses iid samples $y'_1, \ldots, y'_{M'}$ from the
marginal distribution $p^{Y'}$.
Conditioned on each sample $y'_{m'} \in \cY'$, Bob chooses iid
samples $y_{m', 1}, \ldots, y_{m', M}$ from the conditioned 
marginal distribution $p^{Y} | (Y' = y'_{m'}$. Thus, a total of
$M' M$ samples are chosen by Bob. 
Alice inputs a uniformly random sample from her chosen set into the
channel $\cN$. 
Similarly, Bob inputs a uniformly random sample from his chosen set 
into the channel $\cN$. The hope is that this strategy obfuscates
Eve's received state so much that she is unable to figure out 
the actual input message pair $(a,b)$.

Let $\sigma^E$ denote the marginal on $E$ of the control state
$\sigma^{X' X Y' Y E}$. Define
\begin{equation}
\label{eq:CMGdef1}
\sigma^E_{\vec{x'},  \vec{x}, \vec{y'}, \vec{y}} :=
(L' (L+1) M' (M+1))^{-1}
\sum_{a'=1}^{L'} \sum_{a=1}^L
\sum_{b'=1}^{M'} \sum_{b=1}^M
\sigma^E_{x'_{a'} x_{a'a} y'_{b'} y_{b'b}}.
\end{equation}
We would like to show that the  so-called CMG convex split quantity 
defined below is small, when
$L'$, $L$, $M'$, $M$ are suitably large.
\begin{equation}
\label{eq:CMGdef2}
\CMG :=
\E_{\vec{x'}, \vec{x}, \vec{y'}, \vec{y}}[
\|\sigma^E_{\vec{x'},  \vec{x}, \vec{y'}, \vec{y}} - \sigma^E\|_1
],
\end{equation}
the expectation being over the choice of the samples
$\vec{x'} := (x'_1, \ldots, x'_{L'})$,
$\vec{x} := (x_{1,1}, \ldots, x'_{L',L})$,
$\vec{y'} := (y'_1, \ldots, y'_{M'})$,
$\vec{y} := (y_{1,1}, \ldots, y'_{M',M})$ according to the procedure
described above. Recall that for $p > 0$, the 
{\em Schatten-$\ell_p$ norm} of a matrix $M$ is defined as 
$
\|M\|_p := \Tr[(A^\dag A)^{p/2}].
$
The {\em Schatten-$\ell_\infty$ norm}, aka operator norm induced from the
Hilbert space norm, is defined by taking $p \rightarrow +\infty$ in the
above expression, resulting in
$\|M\|_\infty$ being the largest singular value of $M$. For normal
matrices $M$, it equals the largest absolute value of an 
eigenvalue of $M$.

We will try to prove the desired upper bound by using the matrix 
weighted Cauchy-Schwarz inequality given below. 
\begin{fact}
\label{fact:matrixCauchySchwarz}
Let $M$ be Hermitian matrix on a Hilbert space. Let $\sigma$ be a
normalised density matrix on the same space such that
$\supp(M) \subseteq \supp(\sigma)$. Then,
\[
\|M\|_1 \leq \|\sigma^{-1/4} M \sigma^{-1/4}\|_2,
\]
where $\sigma^{-1}$ is the so-called Moore-Penrose pseudoinverse of 
$\sigma$ i.e. the inverse on the support of $\sigma$ and the zero
operator orthogonal to the support.
If $\Tr \sigma \neq 1$, still an appropriate inequality
can be easily obtained by rescaling by $\Tr \sigma$.
\end{fact}

In the course of the attempted
proof, the role of simultaneous smoothing will become clear. Note that
the {\em support} of $\sigma^E_{\vec{x'},  \vec{x}, \vec{y'}, \vec{y}}$ 
lies inside the support of $\sigma^E$ for any choice of 
$\vec{x'}$, $\vec{x}$, $\vec{y'}$, $\vec{y}$. 
Recall that the {\em support}
of a Hermitian operator $M$ is defined as the span of its eigenvectors 
corresponding to non-zero eigenvalues, and is denoted in this paper by
$\supp(M)$. 
Let $(\sigma^E)^{-1}$ denote the {\em Moore-Penrose pseudoinverse}
of the matrix $\sigma^E$ i.e. the inverse of $\sigma^E$ on its 
support $\supp(\sigma^E)$, and the zero operator orthogonal to 
$\supp(\sigma^E)$.
Positive and negative
powers of $(\sigma^E)^{-1/4}$ can now be defined naturally by going
to the eigendecompostion of $\sigma^E$. Define
\[
\begin{array}{c}
\tsigma^E_{\vec{x'},  \vec{x}, \vec{y'}, \vec{y}} :=
(\sigma^E)^{-1/4} \circ \sigma^E_{\vec{x'},  \vec{x}, \vec{y'}, \vec{y}} :=
(\sigma^E)^{-1/4} (\sigma^E_{\vec{x'},  \vec{x}, \vec{y'}, \vec{y}})
(\sigma^E)^{-1/4}, \\
\tsigma^E_{x'_{a'} x_{a'a} y'_{b'} y_{b'b}} :=
(\sigma^E)^{-1/4}
(\sigma^E_{x'_{a'} x_{a'a} y'_{b'} y_{b'b}})
(\sigma^E)^{-1/4}.
\end{array}
\]
By Fact~\ref{fact:matrixCauchySchwarz} and convex of the square 
function,
it suffices to prove an upper bound on the following quantity.
\begin{equation}
\label{eq:CMGsquare1}
\begin{array}{rcl}
\CMG^2
& = &
\left(
\E_{\vec{x'}, \vec{x}, \vec{y'}, \vec{y}}[
\|\sigma^E_{\vec{x'},  \vec{x}, \vec{y'}, \vec{y}} - \sigma^E\|_1
]\right)^2 
\;\leq\;
\left(
\E_{\vec{x'}, \vec{x}, \vec{y'}, \vec{y}}[
\|\tsigma^E_{\vec{x'},  \vec{x}, \vec{y'}, \vec{y}} - (\sigma^E)^{1/2}\|_2
]\right)^2 \\
& \leq &
\E_{\vec{x'}, \vec{x}, \vec{y'}, \vec{y}}\left[
\|\tsigma^E_{\vec{x'},  \vec{x}, \vec{y'}, \vec{y}} - 
  (\sigma^E)^{1/2}\|_2^2
\right] \\
&   =  &
\E_{\vec{x'}, \vec{x}, \vec{y'}, \vec{y}}\left[
\|\tsigma^E_{\vec{x'},  \vec{x}, \vec{y'}, \vec{y}}\|_2^2 + 
\|(\sigma^E)^{1/2}\|_2^2 -
2 \Tr[\tsigma^E_{\vec{x'},  \vec{x}, \vec{y'}, \vec{y}} 
	(\sigma^E)^{1/2}]
\right]\\
& \leq &
\E_{\vec{x'}, \vec{x}, \vec{y'}, \vec{y}}\left[
\|\tsigma^E_{\vec{x'},  \vec{x}, \vec{y'}, \vec{y}}\|_2^2 +  
\Tr[\sigma^E]
	- 2 \Tr[\tsigma^E_{\vec{x'},  \vec{x}, \vec{y'}}] 
\right] \\
&   =  &
\E_{\vec{x'}, \vec{x}, \vec{y'}, \vec{y}}[
\|\tsigma^E_{\vec{x'},  \vec{x}, \vec{y'}, \vec{y}}\|_2^2
] + 1 - 2   
\;  = \;
\E_{\vec{x'}, \vec{x}, \vec{y'}, \vec{y}}[
\|\tsigma^E_{\vec{x'},  \vec{x}, \vec{y'}, \vec{y}}\|_2^2
]- 1.   
\end{array}
\end{equation}

It thus suffices to upper bound 
\begin{equation}
\label{eq:CMGsquare2}
\begin{array}{rcl}
\lefteqn{
\E_{\vec{x'}, \vec{x}, \vec{y'}, \vec{y}}[
\|\tsigma^E_{\vec{x'},  \vec{x}, \vec{y'}, \vec{y}}\|_2^2
] 
} \\
& = &
(L' (L+1) M' (M+1))^{-2} \\
&  &
~~~~
\E_{\vec{x'}, \vec{x}, \vec{y'}, \vec{y}}\left[
\Tr\left[
\left(
\sum_{a'=1}^{L'} \sum_{a=1}^L
\sum_{b'=1}^{M'} \sum_{b=1}^M
\tsigma^E_{x'_{a'} x_{a'a} y'_{b'} y_{b'b}}
\right) 
\right. 
\right. \\
&  &
~~~~~~~~~~~~~~~~~~~~~~~~
\left.
\left.
\left(
\sum_{a'=1}^{L'} \sum_{a=1}^L
\sum_{b'=1}^{M'} \sum_{b=1}^M
\tsigma^E_{x'_{\hat{a'}} x_{\hat{a'}\hat{a}} y'_{\hat{b'}} 
	   y_{\hat{b'}\hat{b}}}
\right)
\right]
\right] \\
& = &
(L' (L+1) M' (M+1))^{-2} \\
&   &
~~~~
{} \cdot
\sum_{a',\hat{a'}=1}^{L'} \sum_{a,\hat{a}=1}^L
\sum_{b',\hat{b'}=1}^{M'} \sum_{b,\hat{b}=1}^M
\E_{\vec{x'}, \vec{x}, \vec{y'}, \vec{y}}\left[
\Tr[
\tsigma^E_{x'_{a'} x_{a'a} y'_{b'} y_{b'b}}
\tsigma^E_{x'_{\hat{a'}} x_{\hat{a'}\hat{a}} y'_{\hat{b'}} 
	   y_{\hat{b'}\hat{b}}}
]
\right].
\end{array}
\end{equation}

For any fixed choice of, say, $a'$, $b'$ define 
\[
\sigma^E_{x'_{a'}, y'_{b'}} :=
\E_{x y | x'_{a'} y'_{b'}}[
\sigma^E_{x'_{a'} x y'_{b'} y}
]
\]
in the natural fashion, where the expectation is taken over the independent
choice
of $x$ according to the distribution $p(x | x'_{a'})$ and 
$y$ according to the distribution $p(y | y'_{b'})$. Other terms like
$\sigma^E_{x'_{a'}, x_a, y'_{b'}}$,
$\sigma^E_{y'_{b'}}$,
$\tsigma^E_{x'_{a'}, x_a}$ etc. are defined in the natural fashion.

Consider a term like the one below for a fixed choice of 
$a' \neq \hat{a'}$, $b' \neq \hat{b'}$, and any fixed choice of
$a$, $\hat{a}$, $b$, $\hat{b}$. We have,
\begin{equation}
\label{eq:CMGterm1}
\begin{array}{rcl}
\lefteqn{
\E_{\vec{x'}, \vec{x}, \vec{y'}, \vec{y}}\left[
\Tr[
\tsigma^E_{x'_{a'} x_{a'a} y'_{b'} y_{b'b}}
\tsigma^E_{x'_{\hat{a'}} x_{\hat{a'}\hat{a}} y'_{\hat{b'}} 
	   y_{\hat{b'}\hat{b}}}
]
\right]
} \\
& = &
\Tr\left[
\E_{x', \hat{x'}, y', \hat{y'}, x, \hat{x}, y, \hat{y}}\left[
\tsigma^E_{x' x y' y}
\tsigma^E_{\hat{x'} \hat{x} \hat{y'} \hat{y}}
\right]
\right]
\;=\;
\Tr[\tsigma^E \tsigma^E] = \Tr[\sigma^E] = 1,
\end{array}
\end{equation}
where the second expectation is taken in the natural fashion over
independent choices of 
$(x', x)$ according to $p(x', x)$,
$(y', y)$ according to $p(y', y)$,
$(\hat{x'}, \hat{x})$ according to $p(\hat{x'}, \hat{x})$,
$(\hat{y'}, \hat{y})$ according to $p(\hat{y'}, \hat{y})$.
The number of such terms is $(L'-1)L' (M'-1)M' L^2 M^2$.

Consider a term like the one below for a fixed choice of 
$a' = \hat{a'}$, $b' \neq \hat{b'}$, and any fixed choice of
$a \neq \hat{a}$, $b$, $\hat{b}$. We have,
\begin{equation}
\label{eq:CMGterm2}
\begin{array}{rcl}
\lefteqn{
\E_{\vec{x'}, \vec{x}, \vec{y'}, \vec{y}}\left[
\Tr[
\tsigma^E_{x'_{a'} x_{a'a} y'_{b'} y_{b'b}}
\tsigma^E_{x'_{\hat{a'}} x_{\hat{a'}\hat{a}} y'_{\hat{b'}} 
	   y_{\hat{b'}\hat{b}}}
]
\right]
} \\
& = &
\Tr\left[
\E_{x', y', \hat{y'}, x, \hat{x}, y, \hat{y}}\left[
\tsigma^E_{x' x y' y}
\tsigma^E_{x' \hat{x} \hat{y'} \hat{y}}
\right]
\right]
\;=\;
\Tr\left[
\E_{x'}\left[
\E_{x|x', y', y}[\tsigma^E_{x' x y' y}]
\E_{\hat{x}|x',\hat{y'}, \hat{y}}[\tsigma^E_{x' \hat{x} \hat{y'} \hat{y}}]
\right]
\right] \\
& = &
\E_{x'}\left[
\Tr\left[
\tsigma^E_{x'}
\tsigma^E_{x'}
\right]
\right] 
\;=\;
2^{I_2(X':E)_\sigma}
\;\leq\;
2^{I_\infty(X':E)_\sigma}
\end{array}
\end{equation}
where the second expectation is taken in the natural fashion over
choices of $x'$ according to $p(x')$, independent choices of
$x$ and $\hat{x}$ according to $p(\cdot | x')$, and independent choices of
$(y', y)$ according to $p(y', y)$,
$(\hat{y'}, \hat{y})$ according to $p(\hat{y'}, \hat{y})$.
The number of such terms is $L' (M'-1)M' L(L-1) M^2$.
Above, we have used the {\em non-smooth R\'{e}nyi-2 mutual information}
$I_2(X':E)_\sigma$, and 
the {\em non-smooth R\'{e}nyi-$\infty$ mutual information}, aka
non-smooth max mutual information,
$I_\infty(X':E)_\sigma$ under the joint state $\sigma^{X'E}$. These two
mutual informations
as well as the {\em non-smooth R\'{e}nyi-2 divergence} and 
{\em non-smooth R\'{e}nyi-$\infty$ divergence}, aka non-smooth max
divergence, behind the mutual informations are now defined below.
\begin{definition}
\[
\begin{array}{c}
I_2(X':E)_\sigma := 
D_2(\sigma^{X'E} \| \sigma^{X'} \otimes \sigma^E), \\
I_\infty(X':E)_\sigma := 
D_\infty(\sigma^{X'E} \| \sigma^{X'} \otimes \sigma^E), \\
D_2(\alpha \| \beta) := 
2 \log\|\beta^{-1/4} \alpha \beta^{-1/4}\|_2, ~~
\mbox{if $\supp(\alpha) \leq \supp(\beta)$, $+\infty$ otherwise}, \\
D_\infty(\alpha \| \beta) := 
\log\|\beta^{-1/2} \alpha \beta^{-1/2}\|_\infty, ~~
\mbox{if $\supp(\alpha) \leq \supp(\beta)$, $+\infty$ otherwise}.
\end{array}
\]
\end{definition}
where $\alpha$ is a {\em subnormalised density matrix} and 
$\beta$ is a positive
semidefinite matrix acting on the same Hilbert space. In this paper,
a subnormalised density matrix means a positive semidefinite matrix with
trace at most one;  a {\em normalised density matrix} means the trace is
exactly one. Note that 
$D_2(\alpha \| \beta) \leq D_\infty(\alpha \| \beta)$, and that
for a cq state $\sigma^{X'E}$, 
\[
I_2(X':E)_\sigma =
\log \E_{x'} \Tr[((\sigma^E)^{-1/4} \sigma^E_{x'} (\sigma^E)^{-1/4})^2] =
\log \E_{x'} \Tr[(\tsigma^E_{x'})^2],
\]
the expecation over $x'$ being taken according to the classical probability
distribution $\sigma^{X'}$.
For later use, we remark that for a cq state $\sigma^{X'E}$,
\[
I_\infty(X':E)_\sigma =
\max_{x'} D_\infty(\sigma^E_{x'} \| \sigma^E).
\]

Consider a term like the one below for a fixed choice of 
$b' = \hat{b'}$, $a' \neq \hat{a'}$, and any fixed choice of
$b \neq \hat{b}$, $a$, $\hat{a}$. We have similarly
\begin{equation}
\label{eq:CMGterm3}
\begin{array}{rcl}
\E_{\vec{x'}, \vec{x}, \vec{y'}, \vec{y}}\left[
\Tr[
\tsigma^E_{x'_{a'} x_{a'a} y'_{b'} y_{b'b}}
\tsigma^E_{x'_{\hat{a'}} x_{\hat{a'}\hat{a}} y'_{\hat{b'}} 
	   y_{\hat{b'}\hat{b}}}
]
\right]
& = &
2^{I_2(Y':E)_\sigma}
\;\leq\;
2^{I_\infty(Y':E)_\sigma}.
\end{array}
\end{equation}
The number of such terms is $M' (L'-1)L' L^2 M(M-1)$.

Consider a term like the one below for a fixed choice of 
$b' = \hat{b'}$, $b = \hat{b}$, 
$a' \neq \hat{a'}$, and any fixed choice of
$a$, $\hat{a}$. We have similarly
\begin{equation}
\label{eq:CMGterm4}
\begin{array}{rcl}
\lefteqn{
\E_{\vec{x'}, \vec{x}, \vec{y'}, \vec{y}}\left[
\Tr[
\tsigma^E_{x'_{a'} x_{a'a} y'_{b'} y_{b'b}}
\tsigma^E_{x'_{\hat{a'}} x_{\hat{a'}\hat{a}} y'_{\hat{b'}} 
	   y_{\hat{b'}\hat{b}}}
]
\right]
} \\
& = &
\Tr\left[
\E_{x', \hat{x'}, x, \hat{x}, y', y}\left[
\tsigma^E_{x' x y' y}
\tsigma^E_{\hat{x'} \hat{x} y' y}
\right]
\right]
\;=\;
\Tr\left[
\E_{y', y}\left[
\E_{x', x}[\tsigma^E_{x' x y' y}]
\E_{\hat{x'}, \hat{x}}[\tsigma^E_{\hat{x'} \hat{x} y' y}]
\right]
\right] \\
& = &
\E_{y',y}\left[
\Tr\left[
\tsigma^E_{y' y}
\tsigma^E_{y' y}
\right]
\right] 
\;=\;
2^{I_2(Y'Y:E)_\sigma}
\;\leq\;
2^{I_\infty(Y'Y:E)_\sigma}.
\end{array}
\end{equation}
The number of such terms is $ (L'-1)L' L^2 M' M$.

Consider a term like the one below for a fixed choice of 
$a' = \hat{a'}$, $a = \hat{a}$, 
$b' \neq \hat{b'}$, and any fixed choice of
$b$, $\hat{b}$. We have similarly
\begin{equation}
\label{eq:CMGterm5}
\begin{array}{rcl}
\E_{\vec{x'}, \vec{x}, \vec{y'}, \vec{y}}\left[
\Tr[
\tsigma^E_{x'_{a'} x_{a'a} y'_{b'} y_{b'b}}
\tsigma^E_{x'_{\hat{a'}} x_{\hat{a'}\hat{a}} y'_{\hat{b'}} 
	   y_{\hat{b'}\hat{b}}}
]
\right]
& = &
2^{I_2(X'X:E)_\sigma}
\;\leq\;
2^{I_\infty(X'X:E)_\sigma}.
\end{array}
\end{equation}
The number of such terms is $ L'L (M'-1) M' M^2$.

Consider a term like the one below for a fixed choice of 
$b' = \hat{b'}$, $a' = \hat{a'}$, and any fixed choice of
$b \neq \hat{b}$, $a \neq \hat{a}$. We have similarly,
\begin{equation}
\label{eq:CMGterm6}
\begin{array}{rcl}
\lefteqn{
\E_{\vec{x'}, \vec{x}, \vec{y'}, \vec{y}}\left[
\Tr[
\tsigma^E_{x'_{a'} x_{a'a} y'_{b'} y_{b'b}}
\tsigma^E_{x'_{\hat{a'}} x_{\hat{a'}\hat{a}} y'_{\hat{b'}} 
	   y_{\hat{b'}\hat{b}}}
]
\right]
} \\
& = &
\Tr\left[
\E_{x', y', x, \hat{x}, y, \hat{y}}\left[
\tsigma^E_{x' x y' y}
\tsigma^E_{x' \hat{x} y' \hat{y}}
\right]
\right]
\;=\;
\Tr\left[
\E_{x', y'}\left[
\E_{x|x', y|y'}[\tsigma^E_{x' x y' y}]
\E_{\hat{x}|x', \hat{y}|y'}[\tsigma^E_{x' \hat{x} y' \hat{y}}]
\right]
\right] \\
& = &
\E_{x',y'}\left[
\Tr\left[
\tsigma^E_{x' y'}
\tsigma^E_{x' y'}
\right]
\right] 
\;=\;
2^{I_2(X'Y':E)_\sigma}
\;\leq\;
2^{I_\infty(X'Y':E)_\sigma}.
\end{array}
\end{equation}
The number of such terms is $M' L' L(L-1) M(M-1)$.

Consider a term like the one below for a fixed choice of 
$b' = \hat{b'}$, $a' = \hat{a'}$, and any fixed choice of
$b = \hat{b}$, $a \neq \hat{a}$. We have similarly,
\begin{equation}
\label{eq:CMGterm7}
\begin{array}{rcl}
\lefteqn{
\E_{\vec{x'}, \vec{x}, \vec{y'}, \vec{y}}\left[
\Tr[
\tsigma^E_{x'_{a'} x_{a'a} y'_{b'} y_{b'b}}
\tsigma^E_{x'_{\hat{a'}} x_{\hat{a'}\hat{a}} y'_{\hat{b'}} 
	   y_{\hat{b'}\hat{b}}}
]
\right]
} \\
& = &
\Tr\left[
\E_{x', y', y}\left[
\E_{x|x'}[\tsigma^E_{x' x y' y}]
\E_{\hat{x}|x'}[\tsigma^E_{x' \hat{x} y' y}]
\right]
\right] \\
& = &
\E_{x',y', y}\left[
\Tr\left[
\tsigma^E_{x' y' y}
\tsigma^E_{x' y' y}
\right]
\right] 
\;=\;
2^{I_2(X'Y'Y:E)_\sigma}
\;\leq\;
2^{I_\infty(X'Y'Y:E)_\sigma}.
\end{array}
\end{equation}
The number of such terms is $M' L' L(L-1) M$.

Consider a term like the one below for a fixed choice of 
$b' = \hat{b'}$, $a' = \hat{a'}$, and any fixed choice of
$b \neq \hat{b}$, $a = \hat{a}$. We have similarly,
\begin{equation}
\label{eq:CMGterm8}
\begin{array}{rcl}
\E_{\vec{x'}, \vec{x}, \vec{y'}, \vec{y}}\left[
\Tr[
\tsigma^E_{x'_{a'} x_{a'a} y'_{b'} y_{b'b}}
\tsigma^E_{x'_{\hat{a'}} x_{\hat{a'}\hat{a}} y'_{\hat{b'}} 
	   y_{\hat{b'}\hat{b}}}
]
\right]
& = &
2^{I_2(X'XY':E)_\sigma}
\;\leq\;
2^{I_\infty(X'X Y':E)_\sigma}.
\end{array}
\end{equation}
The number of such terms is $M' L' L M(M-1)$.

Finally, consider a term like the one below for a fixed choice of 
$b' = \hat{b'}$, $a' = \hat{a'}$, and any fixed choice of
$b = \hat{b}$, $a = \hat{a}$. We have similarly,
\begin{equation}
\label{eq:CMGterm9}
\begin{array}{rcl}
\lefteqn{
\E_{\vec{x'}, \vec{x}, \vec{y'}, \vec{y}}\left[
\Tr[
\tsigma^E_{x'_{a'} x_{a'a} y'_{b'} y_{b'b}}
\tsigma^E_{x'_{\hat{a'}} x_{\hat{a'}\hat{a}} y'_{\hat{b'}} 
	   y_{\hat{b'}\hat{b}}}
]
\right]
} \\
& = &
\Tr\left[
\E_{x', y', x, y}[
\tsigma^E_{x' x y' y}
\tsigma^E_{x' x y' y}
]
\right] 
\;=\;
2^{I_2(X'X Y'Y:E)_\sigma}
\;\leq\;
2^{I_\infty(X'X Y'Y:E)_\sigma}.
\end{array}
\end{equation}
The number of such terms is $M' L' L M$.

Putting Equations~\ref{eq:CMGterm1}, \ref{eq:CMGterm2}, \ref{eq:CMGterm3},
\ref{eq:CMGterm4}, \ref{eq:CMGterm5}, \ref{eq:CMGterm6}, \ref{eq:CMGterm7},
\ref{eq:CMGterm8}, \ref{eq:CMGterm9} back into
Equation~\ref{eq:CMGsquare2}, we get
\begin{equation}
\label{eq:CMGsquare3}
\begin{array}{rcl}
\lefteqn{
\E_{\vec{x'}, \vec{x}, \vec{y'}, \vec{y}}[
\|\tsigma^E_{\vec{x'},  \vec{x}, \vec{y'}, \vec{y}}\|_2^2
] 
} \\
& = &
(L' (L+1) M' (M+1))^{-2} \\
&  &
~~~
{} \cdot
\sum_{a',\hat{a'}=1}^{L'} \sum_{a,\hat{a}=1}^L
\sum_{b',\hat{b'}=1}^{M'} \sum_{b,\hat{b}=1}^M
\E_{\vec{x'}, \vec{x}, \vec{y'}, \vec{y}}\left[
\Tr[
\tsigma^E_{x'_{a'} x_{a'a} y'_{b'} y_{b'b}}
\tsigma^E_{x'_{\hat{a'}} x_{\hat{a'}\hat{a}} y'_{\hat{b'}} 
	   y_{\hat{b'}\hat{b}}}
]
\right] \\
& \leq &
(L' (L+1) M' (M+1))^{-2} \\
&  &
~~~~~~
\left(
(L'-1)L' (M'-1)M' L^2 M^2 \cdot 1 +
L' (M'-1)M' L(L-1) M^2 \cdot 2^{I_2(X':E)_\sigma} 
\right. \\
&  &
~~~~~~~
{} +
(L'-1) L' M' L^2  M(M-1) \cdot 2^{I_2(Y':E)_\sigma} \\
&  &
~~~~~~~
{} +
(L'-1) L' L^2  M' M \cdot 2^{I_2(Y'Y:E)_\sigma} \\
&  &
~~~~~~~
{} +
L' L (M'-1) M' M^2 \cdot 2^{I_2(X'X:E)_\sigma} \\
&  &
~~~~~~~
{} +
L' M' L(L-1)  M(M-1) \cdot 2^{I_2(X'Y':E)_\sigma} \\
&  &
~~~~~~~
{} +
L' M' L(L-1)  M \cdot 2^{I_2(X'Y'Y:E)_\sigma} \\
&  &
~~~~~~~
{} +
L' M' L  M(M-1) \cdot 2^{I_2(X'X Y':E)_\sigma} \\
&  &
~~~~~~~
\left.
{} +
L' M' L M \cdot 2^{I_2(X'X Y'Y:E)_\sigma}
\right) \\
& \leq &
1 + (L')^{-1} \cdot 2^{I_2(X':E)_\sigma} +
(M')^{-1} \cdot 2^{I_2(Y':E)_\sigma}  +
(L'M')^{-1} \cdot 2^{I_2(X'Y':E)_\sigma}  \\
& &
~~
{} +
(L'M'M)^{-1} \cdot 2^{I_2(X'Y'Y:E)_\sigma} +
(L'L M)^{-1} \cdot 2^{I_2(X'XY':E)_\sigma} \\
& &
~~
{} +
(L'L M'M)^{-1} \cdot 2^{I_2(X'XY'Y:E)_\sigma}.
\end{array}
\end{equation}
Putting Equation~\ref{eq:CMGsquare3} back into
Equation~\ref{eq:CMGsquare1}, we get
\begin{eqnarray*}
\CMG^2
& \leq &
(L')^{-1} \cdot 2^{I_2(X':E)_\sigma} +
(M')^{-1} \cdot 2^{I_2(Y':E)_\sigma}  \\
&  &
~~
{} +
(L'L)^{-1} \cdot 2^{I_2(X'X:E)_\sigma} +
(M'M)^{-1} \cdot 2^{I_2(Y'Y:E)_\sigma} +
(L'M')^{-1} \cdot 2^{I_2(X'Y':E)_\sigma}  \\
& &
~~
{} +
(L'M'M)^{-1} \cdot 2^{I_2(X'Y'Y:E)_\sigma} +
(L'L M)^{-1} \cdot 2^{I_2(X'XY':E)_\sigma} \\
& &
~~
{} +
(L'L M'M)^{-1} \cdot 2^{I_2(X'XY'Y:E)_\sigma}.
\end{eqnarray*}
We have thus proved the following non-smooth CMG covering lemma.
\begin{lemma}
\label{lem:nonsmoothCMGcovering}
For the CMG covering problem, let
\begin{eqnarray*}
\log L'
& > &
I_2(X':E)_\sigma + \log \epsilon^{-2}, \\
\log M'
& > &
I_2(Y':E)_\sigma + \log \epsilon^{-2}, \\
\log L' + \log L
& > &
I_2(X'X:E)_\sigma + \log \epsilon^{-2}, \\
\log M' + \log M
& > &
I_2(Y'Y:E)_\sigma + \log \epsilon^{-2}, \\
\log L' + \log M'
& > &
I_2(X'Y':E)_\sigma + \log \epsilon^{-2}, \\
\log L' + \log M' + \log M
& > &
I_2(X'Y'Y:E)_\sigma + \log \epsilon^{-2}, \\
\log L' + \log M' + \log L
& > &
I_2(X'XY':E)_\sigma + \log \epsilon^{-2}, \\
\log L' + \log M' + \log L + \log M
& > &
I_2(X'XY'Y:E)_\sigma + \log \epsilon^{-2}.
\end{eqnarray*}
Then, 
$
\CMG < 3 \epsilon.
$
\end{lemma}

For many applications to information theory, it is desirable to obtain
a {\em fully smooth version} of the above covering results. This is
because a fully smooth one shot inner bound, besides being larger than
the non-smooth inner bound,  easily implies finite block
length as well as asymptotic iid inner bounds.
In a fully smooth statement, the 
$D_2(\cdot \| \cdot)$ quantities on the right hand sides of the
inequalities describing the inner bounds are replaced by their smooth
counterparts $D^\epsilon_2(\cdot \| \cdot)$. Let $0 < \epsilon < 1$.
The $\epsilon$-smooth R\'{e}nyi divergences are defined as follows:
\[
D_2^\epsilon(\alpha \| \beta) 
:=  
\min_{\alpha' \approx_\epsilon \alpha}
D_2(\alpha' \| \beta), ~~
D_\infty^\epsilon(\alpha \| \beta) 
:= 
\min_{\alpha' \approx_\epsilon \alpha}
D_\infty(\alpha' \| \beta), 
\]
where the minimisation is over all subnormalised density matrices 
$\alpha'$ satisfying $\|\alpha' - \alpha\|_1 \leq \epsilon (\Tr\alpha)$. 
The smooth mutual informations are defined accordingly.
\[
I_2^\epsilon(X' : E)_\sigma
:=  
D_2^\epsilon(\sigma^{X'E} \| \sigma^{X'} \otimes \sigma^{E}), \\
I_\infty^\epsilon(X' : E)_\sigma
:= 
D_\infty^\epsilon(\sigma^{X'E} \| \sigma^{X'} \otimes \sigma^{E}).
\]

We remark
that several alternate definitions of non-smooth and smooth R\'{e}nyi
divergences have been provided in the literature. However the work
of \cite{Jain:minimax} shows that all the smooth R\'{e}nyi-$p$ 
divergences for $p > 1$ are essentially equivalent to 
$D^\epsilon_\infty$ up to tweaks in the smoothing parameter $\epsilon$
and dimension independent additive corrections that are polynomial in
$\log \epsilon^{-1}$. So henceforth, many of our fully smooth 
convex split and
decoupling results shall be stated in terms of $D^\epsilon_\infty$ or 
quantities derived from $D^\epsilon_\infty$ like $I^\epsilon_\infty$.

One can now immediately see the problems that arise in making the proof
of Lemma~\ref{lem:nonsmoothCMGcovering} 
fully smooth. Take for example the expressions
$
\E_{x'}[\Tr[(\tsigma^E_{x'})^2]],
$
$
\E_{y'}[\Tr[(\tsigma^E_{y'})^2]]
$
arising in the proof. These
expressions get upper bounded by
$2^{I_2(X':E)_\sigma}$, $2^{I_2(Y':E)_\sigma}$
respectively during the proof.
If we now want to upper bound them by their smooth counterparts
$2^{I_2^\epsilon(X':E)_\sigma}$, $2^{I_2^\epsilon(Y':E)_\sigma}$,
we need to prove the existence of a {\bf single} subnormalised
density matrix 
$\bsigma^{X'XY'YE} \approx_{f(\epsilon)} \sigma^{X'XY'YE}$
for some dimension independent function $f(\epsilon)$ that goes to
zero as $\epsilon \rightarrow 0$ such that
\begin{eqnarray*}
D_2(\bsigma^{X'E} \| \bsigma^{X'} \otimes \bsigma^E) 
& \leq &
D_2^\epsilon(\sigma^{X'E} \| \sigma^{X'} \otimes \sigma^E)
+ g(\epsilon^{-1}), \\
D_2(\bsigma^{Y'E} \| \bsigma^{Y'} \otimes \bsigma^E) 
& \leq &
D_2^\epsilon(\sigma^{X'E} \| \sigma^{Y'} \otimes \sigma^E)
+ g(\epsilon^{-1}),
\end{eqnarray*}
for some `not too large' function $g(\epsilon^{-1})$.
One could then run the proof of  
Lemma~\ref{lem:nonsmoothCMGcovering} after
first changing the control state from 
$\sigma^{X'XY'YE}$ to $\bsigma^{X'XY'YE}$ and paying an additive
price of $f(\epsilon)$ in the trace distance. In the fully classical
setting, it is indeed possible to find a single state
(classical probability distribution)
$\bsigma^{X'X Y'Y E}$ satisfying all the eight smoothing requirements
of a putative fully smooth version of the inner bound of
Lemma~\ref{lem:nonsmoothCMGcovering}. Such a state is obtained by
taking the optimum states in each of the eight smoothings, extending
them if required to the full alphabet $X'X Y'Y E$, and then taking 
the `intersection' of these six probability subdistributions. As 
in \cite{sen:oneshot}, we define the 
intersection of two probability subdistributions $p^X$, $q^X$ on the same
alphabet $X$ as
\begin{equation}
\label{eq:classicalintersection}
(p \cap q)^X(x) := \min\{p(x), q(x)\}.
\end{equation}
It is easy to see that 
\[
\|(p \cap q)^X - p^X\|_1 \leq \|p^X - q^X\|_1, ~~
\|(p \cap q)^X - q^X\|_1 \leq \|p^X - q^X\|_1.
\]
However, there is no suitable notion of intersection of subnormalised
density matrices, fundamentally due to the non-commutative nature of
matrix multiplication. The {\em simultaneous smoothing conjecture} 
asserts that there is indeed a quantum state $\bsigma^{X'X Y'Y E}$
satisfying the eight smoothing requirements for a suitable $f(\epsilon)$,
$g(\epsilon^{-1})$. At this level of generality, it remains open till
this date.

We now explain how the {\em telescoping}, aka 
{\em mean zero decomposition},
idea of Cheng, Gao and Berta \cite{Cheng:convexsplit} can be further
extended in order to prove a fully smooth version of
Lemmas~\ref{lem:nonsmoothCMGcovering}.
The telescoping idea bypasses the intersection or simultaneous smoothing
bottleneck for convex split style lemmas. This is done by a clever
use of the triangle inequality for trace distance, upper bounding
the $\CMG$ quantity in Equation~\ref{eq:CMGdef2} by a sum of 
eight trace distances involving suitable
newly defined density matrices. This sum of trace distances gives
rise to the name telescoping. Each individual trace distance term 
satisfies a crucial property, viz. all but one of its natural expectations
evaluate to zero for that term. This feature gives rise to the name
mean-zero decomposition. For the CMG covering problem, we take
Equations~\ref{eq:CMGdef1} and \ref{eq:CMGdef2} and 
write the following telescoping cum mean zero inequality.
\begin{eqnarray*}
\lefteqn{
\CMG
} \\
& = &
\E_{\vec{x'}, \vec{x}, \vec{y'}, \vec{y}}\left[
\left\|
(L' (L+1) M' (M+1))^{-1}
\sum_{a'=1}^{L'} \sum_{a=1}^L
\sum_{b'=1}^{M'} \sum_{b=1}^M
\sigma^E_{x'_{a'} x_{a'a} y'_{b'} y_{b'b}}
 - \sigma^E
\right\|_1
\right] \\
& \leq &
\E_{\vec{x'}, \vec{x}, \vec{y'}, \vec{y}}\left[
\left\|
(L' (L+1) M' (M+1))^{-1}
\sum_{a'=1}^{L'} \sum_{a=1}^L
\sum_{b'=1}^{M'} \sum_{b=1}^M
\sigma^E_{x'_{a'} x_{a'a} y'_{b'} y_{b'b}} 
\right.
\right. \\
&   &
~~~~~~~~~~~~~~~
{} - 
(L' (L+1) M' (M+1))^{-1}
\sum_{a'=1}^{L'} \sum_{a=1}^L
\sum_{b'=1}^{M'} \sum_{b=1}^M
\sigma^E_{x'_{a'} y'_{b'} y_{b'b}} \\
&  &
~~~~~~~~~~~~~~~
{} -
(L' (L+1) M' (M+1))^{-1}
\sum_{a'=1}^{L'} \sum_{a=1}^L
\sum_{b'=1}^{M'} \sum_{b=1}^M
\sigma^E_{x'_{a'} x_{a'a} y'_{b'}} \\
&  &
~~~~~~~~~~~~~~~
\left.
\left.
{} + 
(L' (L+1) M' (M+1))^{-1}
\sum_{a'=1}^{L'} \sum_{a=1}^L
\sum_{b'=1}^{M'} \sum_{b=1}^M
\sigma^E_{x'_{a'} y'_{b'}} 
\right\|_1
\right] \\
&  &
{} +
\E_{\vec{x'}, \vec{x}, \vec{y'}, \vec{y}}\left[
\left\|
(L' (L+1) M' (M+1))^{-1}
\sum_{a'=1}^{L'} \sum_{a=1}^L
\sum_{b'=1}^{M'} \sum_{b=1}^M
\sigma^E_{x'_{a'} y'_{b'} y_{b'b}} 
\right.
\right. \\
&   &
~~~~~~~~~~~~~~~
{} - 
(L' (L+1) M' (M+1))^{-1}
\sum_{a'=1}^{L'} \sum_{a=1}^L
\sum_{b'=1}^{M'} \sum_{b=1}^M
\sigma^E_{y'_{b'} y_{b'b}} \\ 
&   &
~~~~~~~~~~~~~~~
{} - 
(L' (L+1) M' (M+1))^{-1}
\sum_{a'=1}^{L'} \sum_{a=1}^L
\sum_{b'=1}^{M'} \sum_{b=1}^M
\sigma^E_{x'_{a'} y'_{b'}} \\ 
&   &
~~~~~~~~~~~~~~~
\left.
\left.
{} + 
(L' (L+1) M' (M+1))^{-1}
\sum_{a'=1}^{L'} \sum_{a=1}^L
\sum_{b'=1}^{M'} \sum_{b=1}^M
\sigma^E_{y'_{b'}}  
\right\|_1
\right] \\
&  &
{} +
\E_{\vec{x'}, \vec{x}, \vec{y'}, \vec{y}}\left[
\left\|
(L' (L+1) M' (M+1))^{-1}
\sum_{a'=1}^{L'} \sum_{a=1}^L
\sum_{b'=1}^{M'} \sum_{b=1}^M
\sigma^E_{x'_{a'} x_{a'a} y'_{b'}} 
\right.
\right. \\
&   &
~~~~~~~~~~~~~~~
{} - 
(L' (L+1) M' (M+1))^{-1}
\sum_{a'=1}^{L'} \sum_{a=1}^L
\sum_{b'=1}^{M'} \sum_{b=1}^M
\sigma^E_{x'_{a'} x_{a'a}} \\ 
&   &
~~~~~~~~~~~~~~~
{} - 
(L' (L+1) M' (M+1))^{-1}
\sum_{a'=1}^{L'} \sum_{a=1}^L
\sum_{b'=1}^{M'} \sum_{b=1}^M
\sigma^E_{x'_{a'} y'_{b'}} \\ 
&   &
~~~~~~~~~~~~~~~
\left.
\left.
{} + 
(L' (L+1) M' (M+1))^{-1}
\sum_{a'=1}^{L'} \sum_{a=1}^L
\sum_{b'=1}^{M'} \sum_{b=1}^M
\sigma^E_{x'_{a'}}  
\right\|_1
\right] \\
&  &
{} +
\E_{\vec{x'}, \vec{x}, \vec{y'}, \vec{y}}\left[
\left\|
(L' (L+1) M' (M+1))^{-1}
\sum_{a'=1}^{L'} \sum_{a=1}^L
\sum_{b'=1}^{M'} \sum_{b=1}^M
\sigma^E_{x'_{a'} y'_{b'}} 
\right.
\right. 
\hspace*{4cm}
(\arabic{equation})
\label{eq:CMGsmooth1}
\stepcounter{equation}
\\
&   &
~~~~~~~~~~~~~~~
{} - 
(L' (L+1) M' (M+1))^{-1}
\sum_{a'=1}^{L'} \sum_{a=1}^L
\sum_{b'=1}^{M'} \sum_{b=1}^M
\sigma^E_{x'_{a'}} \\ 
&   &
~~~~~~~~~~~~~~~
{} - 
(L' (L+1) M' (M+1))^{-1}
\sum_{a'=1}^{L'} \sum_{a=1}^L
\sum_{b'=1}^{M'} \sum_{b=1}^M
\sigma^E_{y'_{b'}} \\ 
&   &
~~~~~~~~~~~~~~~
\left.
\left.
{} + 
(L' (L+1) M' (M+1))^{-1}
\sum_{a'=1}^{L'} \sum_{a=1}^L
\sum_{b'=1}^{M'} \sum_{b=1}^M
\sigma^E  
\right\|_1
\right] \\
&  &
{} +
\E_{\vec{x'}, \vec{x}, \vec{y'}, \vec{y}}\left[
\left\|
(L' (L+1) M' (M+1))^{-1}
\sum_{a'=1}^{L'} \sum_{a=1}^L
\sum_{b'=1}^{M'} \sum_{b=1}^M
\sigma^E_{x'_{a'} x_{a'a}} 
\right.
\right. \\
&   &
~~~~~~~~~~~~~~~
\left.
\left.
{} - 
(L' (L+1) M' (M+1))^{-1}
\sum_{a'=1}^{L'} \sum_{a=1}^L
\sum_{b'=1}^{M'} \sum_{b=1}^M
\sigma^E_{x'_{a'}} 
\right\|_1
\right] \\
&  &
{} +
\E_{\vec{x'}, \vec{x}, \vec{y'}, \vec{y}}\left[
\left\|
(L' (L+1) M' (M+1))^{-1}
\sum_{a'=1}^{L'} \sum_{a=1}^L
\sum_{b'=1}^{M'} \sum_{b=1}^M
\sigma^E_{y'_{b'} y_{b'b}} 
\right.
\right. \\
&   &
~~~~~~~~~~~~~~~
\left.
\left.
{} - 
(L' (L+1) M' (M+1))^{-1}
\sum_{a'=1}^{L'} \sum_{a=1}^L
\sum_{b'=1}^{M'} \sum_{b=1}^M
\sigma^E_{y'_{b'}} 
\right\|_1
\right] \\
&  &
{} +
\E_{\vec{x'}, \vec{x}, \vec{y'}, \vec{y}}\left[
\left\|
(L' (L+1) M' (M+1))^{-1}
\sum_{a'=1}^{L'} \sum_{a=1}^L
\sum_{b'=1}^{M'} \sum_{b=1}^M
\sigma^E_{x'_{a'}} 
\right.
\right. \\
&   &
~~~~~~~~~~~~~~~
\left.
\left.
{} - 
(L' (L+1) M' (M+1))^{-1}
\sum_{a'=1}^{L'} \sum_{a=1}^L
\sum_{b'=1}^{M'} \sum_{b=1}^M
\sigma^E
\right\|_1
\right] \\
&  &
{} +
\E_{\vec{x'}, \vec{x}, \vec{y'}, \vec{y}}\left[
\left\|
(L' (L+1) M' (M+1))^{-1}
\sum_{a'=1}^{L'} \sum_{a=1}^L
\sum_{b'=1}^{M'} \sum_{b=1}^M
\sigma^E_{y'_{b'}} 
\right.
\right. \\
&   &
~~~~~~~~~~~~~~~
\left.
\left.
{} - 
(L' (L+1) M' (M+1))^{-1}
\sum_{a'=1}^{L'} \sum_{a=1}^L
\sum_{b'=1}^{M'} \sum_{b=1}^M
\sigma^E
\right\|_1
\right].
\end{eqnarray*}

In order to get a fully smooth version
of Lemma \ref{lem:nonsmoothCMGcovering}, we observe that
each individual trace distance term in the upper bound of
Equation~\ref{eq:CMGsmooth1} can then
be further perturbed differently according to the specific smoothing
requirement for that term at a small additive price in the trace distance,
using monotonicity of smooth R\'{e}nyi divergence. 
This is a new observation of this paper which gives us the 
power to bypass simultaneous smoothing.
Each perturbed trace distance term is then further
upper bounded by the Cauchy Schwarz inequality with weighting matrix
$\sigma^E$.
This gives rise to the eight different inequalities in the statement 
of the
fully smooth version of Lemma \ref{lem:nonsmoothCMGcovering}.

Let $i \in [8]$.
For any fixed choice of $x'$, $x$, $y'$, $y$, let 
$\sigma^E_{x'x y'y}(i)$
denote the subnormalised density matrix achieving the optimum in
the smoothed version of the $i$th inequality in the statement of
Lemma \ref{lem:nonsmoothCMGcovering}. For example, for $i = 3$,
the matrix $\sigma^E_{x'x y'y}(3)$ arises from the optimising 
subnormalised density matrix 
$\sigma^{X'XY'YE}(3) \approx_{\epsilon} \sigma^{X'XY'YE}$ in the 
definition of
\[
I_2^\epsilon(X'X:E)_\sigma :=
D_2^\epsilon(\sigma^{X'XE} \| \sigma^{X'X} \otimes \sigma^E) =
D_2(\sigma^{X'XE}(3) \| \sigma^{X'X} \otimes \sigma^E).
\]
Note that formally the optimising matrix in the above expression
for the smooth mutual information is defined only for classical indices
$x', x$ i.e. it is of the form $\sigma^E(3)_{x'x}$. Nevertheless, 
we can extend the notation to indices $x',x,y',y$ by defining
$\sigma^E_{x'x y'y}(3) := \sigma^E(3)_{x'x}$.

We now perturb the $i$th term in the telescoping sum in
Equation~\ref{eq:CMGsmooth1} above according
to $\sigma^E_{x'x y'y}(i)$, paying a total additive price of
$24\epsilon$ in the trace distance.
\begin{eqnarray*}
\lefteqn{\CMG} \\
& \leq &
24 \epsilon +
(L' (L+1) M' (M+1))^{-1} \cdot {} \\
&   &
~~~~~~
\E_{\vec{x'}, \vec{x}, \vec{y'}, \vec{y}}\left[
\left\|
\sum_{a'=1}^{L'} \sum_{a=1}^L \sum_{b'=1}^{M'} \sum_{b=1}^M
(\sigma^E_{x'_{a'} x_{a'a} y'_{b'} y_{b'b}}(8) - 
\sigma^E_{x'_{a'} y'_{b'} y_{b'b}}(8) -
\sigma^E_{x'_{a'} x_{a'a} y'_{b'}}(8) + 
\sigma^E_{x'_{a'} y'_{b'}}(8))
\right\|_1
\right] \\
&  &
{} +
(L' (L+1) M' (M+1))^{-1} \cdot {} \\
&  &
~~~~~~~~~~~
\E_{\vec{x'}, \vec{x}, \vec{y'}, \vec{y}}\left[
\left\|
\sum_{a'=1}^{L'} \sum_{a=1}^L
\sum_{b'=1}^{M'} \sum_{b=1}^M
(
\sigma^E_{x'_{a'} y'_{b'} y_{b'b}}(6) - 
\sigma^E_{y'_{b'} y_{b'b}}(6) - 
\sigma^E_{x'_{a'} y'_{b'}}(6) + 
\sigma^E_{y'_{b'}}(6))
\right\|_1
\right] \\
&  &
{} +
(L' (L+1) M' (M+1))^{-1} \cdot {} \\
& &
~~~~~~~~~~~
\E_{\vec{x'}, \vec{x}, \vec{y'}, \vec{y}}\left[
\left\|
\sum_{a'=1}^{L'} \sum_{a=1}^L
\sum_{b'=1}^{M'} \sum_{b=1}^M
(\sigma^E_{x'_{a'} x_{a'a} y'_{b'}}(7) - 
\sigma^E_{x'_{a'} x_{a'a}}(7) - 
\sigma^E_{x'_{a'} y'_{b'}}(7) + 
\sigma^E_{x'_{a'}}(7))
\right\|_1
\right] \\
&  &
{} + (L' (L+1) M' (M+1))^{-1} \cdot {} \\
& &
~~~~~~~~~~~
\E_{\vec{x'}, \vec{x}, \vec{y'}, \vec{y}}\left[
\left\|
\sum_{a'=1}^{L'} \sum_{a=1}^L
\sum_{b'=1}^{M'} \sum_{b=1}^M
(\sigma^E_{x'_{a'} y'_{b'}}(5) - 
\sigma^E_{x'_{a'}}(5) - 
\sigma^E_{y'_{b'}}(5)  + 
\sigma^E(5))
\right\|_1
\right] 
\hspace*{2.5cm}
(\arabic{equation})
\label{eq:CMGsmooth2}
\stepcounter{equation}
\\
&  &
{} + 
(L' (L+1) M' (M+1))^{-1} 
 \cdot
\E_{\vec{x'}, \vec{x}, \vec{y'}, \vec{y}}\left[
\left\|
\sum_{a'=1}^{L'} \sum_{a=1}^L
\sum_{b'=1}^{M'} \sum_{b=1}^M
(\sigma^E_{x'_{a'} x_{a'a}}(3) - \sigma^E_{x'_{a'}}(3) 
\right\|_1
\right] \\
&  &
{} +
(L' (L+1) M' (M+1))^{-1} \cdot
\E_{\vec{x'}, \vec{x}, \vec{y'}, \vec{y}}\left[
\left\|
\sum_{a'=1}^{L'} \sum_{a=1}^L
\sum_{b'=1}^{M'} \sum_{b=1}^M
(\sigma^E_{y'_{b'} y_{b'b}}(4) - \sigma^E_{y'_{b'}}(4))
\right\|_1
\right] \\
&  &
{} +
(L' (L+1) M' (M+1))^{-1} \cdot
\E_{\vec{x'}, \vec{x}, \vec{y'}, \vec{y}}\left[
\left\|
\sum_{a'=1}^{L'} \sum_{a=1}^L
\sum_{b'=1}^{M'} \sum_{b=1}^M
(\sigma^E_{x'_{a'}}(1) - \sigma^E(1))
\right\|_1
\right] \\
&  &
{} +
(L' (L+1) M' (M+1))^{-1} \cdot
\E_{\vec{x'}, \vec{x}, \vec{y'}, \vec{y}}\left[
\left\|
\sum_{a'=1}^{L'} \sum_{a=1}^L
\sum_{b'=1}^{M'} \sum_{b=1}^M
(\sigma^E_{y'_{b'}}(2) - \sigma^E(2))
\right\|_1
\right].
\end{eqnarray*}

We now upper bound each term in the telescoping sum above using the
Cauchy-Schwarz inequality Fact~\ref{fact:matrixCauchySchwarz} with 
weighting matrix $\sigma^E$. Define 
$
\tsigma^E_{x'_{a'} x_{a'a} y'_{b'} y_{b'b}}(8) :=
(\sigma^E)^{-1/4} \circ
\sigma^E_{x'_{a'} x_{a'a} y'_{b'} y_{b'b}}(8),
$
and other weighted states similarly.
Thus,
\begin{equation}
\label{eq:CMGsmooth3}
\begin{array}{rcl}
\lefteqn{
(L' (L+1) M' (M+1)) \cdot (\CMG - 24 \epsilon) 
} \\
& \leq &
\E_{\vec{x'}, \vec{x}, \vec{y'}, \vec{y}}\left[
\left\|
\sum_{a'=1}^{L'} \sum_{a=1}^L \sum_{b'=1}^{M'} \sum_{b=1}^M
(\tsigma^E_{x'_{a'} x_{a'a} y'_{b'} y_{b'b}}(8) - 
\tsigma^E_{x'_{a'} y'_{b'} y_{b'b}}(8) 
\right.
\right. \\
&  &
~~~~~~~~~~~~~~~~~~~~~~~~~~~~~~~~~~~~~~~~~~~~~~~~~
\left.
\left.
{} -
\tsigma^E_{x'_{a'} x_{a'a} y'_{b'}}(8) + 
\tsigma^E_{x'_{a'} y'_{b'}}(8))
\right\|_2
\right] \\
&  &
{} +
\E_{\vec{x'}, \vec{x}, \vec{y'}, \vec{y}}\left[
\left\|
\sum_{a'=1}^{L'} \sum_{a=1}^L
\sum_{b'=1}^{M'} \sum_{b=1}^M
(
\tsigma^E_{x'_{a'} y'_{b'} y_{b'b}}(6) - 
\tsigma^E_{y'_{b'} y_{b'b}}(6) 
\right.
\right. \\
&  &
~~~~~~~~~~~~~~~~~~~~~~~~~~~~~~~~~~~~~~~~~~~~~~~~~~~~
\left.
\left.
{} - 
\tsigma^E_{x'_{a'} y'_{b'}}(6) + 
\tsigma^E_{y'_{b'}}(6))
\right\|_2
\right] \\
&  &
{} +
\E_{\vec{x'}, \vec{x}, \vec{y'}, \vec{y}}\left[
\left\|
\sum_{a'=1}^{L'} \sum_{a=1}^L
\sum_{b'=1}^{M'} \sum_{b=1}^M
(\tsigma^E_{x'_{a'} x_{a'a} y'_{b'}}(7) - 
\tsigma^E_{x'_{a'} x_{a'a}}(7) 
\right.
\right. \\
&  &
~~~~~~~~~~~~~~~~~~~~~~~~~~~~~~~~~~~~~~~~~~~~~~~~~~~~
\left.
\left.
{} - 
\tsigma^E_{x'_{a'} y'_{b'}}(7) + 
\tsigma^E_{x'_{a'}}(7))
\right\|_2
\right] \\
& &
{} +
\E_{\vec{x'}, \vec{x}, \vec{y'}, \vec{y}}\left[
\left\|
\sum_{a'=1}^{L'} \sum_{a=1}^L
\sum_{b'=1}^{M'} \sum_{b=1}^M
(\tsigma^E_{x'_{a'} y'_{b'}}(5) - 
\tsigma^E_{x'_{a'}}(5) - 
\tsigma^E_{y'_{b'}}(5)  + 
\tsigma^E(5))
\right\|_2
\right] \\
&  &
{} + 
\E_{\vec{x'}, \vec{x}, \vec{y'}, \vec{y}}\left[
\left\|
\sum_{a'=1}^{L'} \sum_{a=1}^L
\sum_{b'=1}^{M'} \sum_{b=1}^M
(\tsigma^E_{x'_{a'} x_{a'a}}(3) - \tsigma^E_{x'_{a'}}(3) 
\right\|_2
\right] \\
&  &
{} +
\E_{\vec{x'}, \vec{x}, \vec{y'}, \vec{y}}\left[
\left\|
\sum_{a'=1}^{L'} \sum_{a=1}^L
\sum_{b'=1}^{M'} \sum_{b=1}^M
(\tsigma^E_{y'_{b'} y_{b'b}}(4) - \tsigma^E_{y'_{b'}}(4))
\right\|_2
\right] \\
&  &
{} +
\E_{\vec{x'}, \vec{x}, \vec{y'}, \vec{y}}\left[
\left\|
\sum_{a'=1}^{L'} \sum_{a=1}^L
\sum_{b'=1}^{M'} \sum_{b=1}^M
(\tsigma^E_{x'_{a'}}(1) - \tsigma^E(1))
\right\|_2
\right] \\
& &
{} +
\E_{\vec{x'}, \vec{x}, \vec{y'}, \vec{y}}\left[
\left\|
\sum_{a'=1}^{L'} \sum_{a=1}^L
\sum_{b'=1}^{M'} \sum_{b=1}^M
(\tsigma^E_{y'_{b'}}(2) - \tsigma^E(2))
\right\|_2
\right].
\end{array}
\end{equation}

Consider a term like the first term in the above telescoping inequality.
We can open it up as
\begin{eqnarray*}
\lefteqn{
\left(
\E_{\vec{x'}, \vec{x}, \vec{y'}, \vec{y}}\left[
\left\|
\sum_{a'=1}^{L'} \sum_{a=1}^L \sum_{b'=1}^{M'} \sum_{b=1}^M
(\tsigma^E_{x'_{a'} x_{a'a} y'_{b'} y_{b'b}}(8) - 
\tsigma^E_{x'_{a'} y'_{b'} y_{b'b}}(8) -
\tsigma^E_{x'_{a'} x_{a'a} y'_{b'}}(8) + 
\tsigma^E_{x'_{a'} y'_{b'}}(8))
\right\|_2
\right]
\right)^2
} \\
& \leq &
\E_{\vec{x'}, \vec{x}, \vec{y'}, \vec{y}}\left[
\left(
\left\|
\sum_{a'=1}^{L'} \sum_{a=1}^L \sum_{b'=1}^{M'} \sum_{b=1}^M
(\tsigma^E_{x'_{a'} x_{a'a} y'_{b'} y_{b'b}}(8) - 
\tsigma^E_{x'_{a'} y'_{b'} y_{b'b}}(8) -
\tsigma^E_{x'_{a'} x_{a'a} y'_{b'}}(8) + 
\tsigma^E_{x'_{a'} y'_{b'}}(8))
\right\|_2
\right)^2
\right] \\
& \leq &
\sum_{a',\hat{a'}=1}^{L'} \sum_{a,\hat{a}=1}^L 
\sum_{b',\hat{b'}=1}^{M'} \sum_{b,\hat{b}=1}^M 
\hspace*{12.5cm}
(\arabic{equation})
\label{eq:CMGsmooth4} 
\stepcounter{equation} \\
&   &
~~~~~
\E_{\vec{x'}, \vec{x}, \vec{y'}, \vec{y}}\left[
\Tr\left[
(\tsigma^E_{x'_{a'} x_{a'a} y'_{b'} y_{b'b}}(8) - 
\tsigma^E_{x'_{a'} y'_{b'} y_{b'b}}(8) -
\tsigma^E_{x'_{a'} x_{a'a} y'_{b'}}(8) + 
\tsigma^E_{x'_{a'} y'_{b'}}(8)) 
\right.
\right.
\\
&   &
~~~~~~~~~~~~~~~~~~~~~~~~
\left.
\left.
(\tsigma^E_{x'_{\hat{a'}} x_{\hat{a'}\hat{a}} 
	    y'_{\hat{b'}} y_{\hat{b'}\hat{b}}}(8) - 
\tsigma^E_{x'_{\hat{a'}} y'_{\hat{b'}} y_{\hat{b'}\hat{b}}}(8) -
\tsigma^E_{x'_{\hat{a'}} x_{\hat{a'}\hat{a}} y'_{\hat{b'}}}(8) + 
\tsigma^E_{x'_{\hat{a'}} y'_{\hat{b'}}}(8))
\right]
\right].
\end{eqnarray*}
Consider now a term in the above summand for a fixed choice of, say,
$a' = \hat{a'}$, $a = \hat{a}$, $b' \neq \hat{b'}$ and any fixed choice 
of $b$, $\hat{b}$. Its expectation looks like
\begin{eqnarray*}
\lefteqn{
\E_{x', x, y', \hat{y'}, y,\hat{y}}\left[
\Tr\left[
(\tsigma^E_{x' x y' y}(8) - 
\tsigma^E_{x' y' y}(8) -
\tsigma^E_{x' x y'}(8) + 
\tsigma^E_{x' y'}(8)) 
(\tsigma^E_{x' x \hat{y'} \hat{y}}(8) - 
\tsigma^E_{x' \hat{y'} \hat{y}}(8) -
\tsigma^E_{x' x \hat{y'}}(8) + 
\tsigma^E_{x' y'}(8))
\right]
\right]
} \\
& = &
\E_{x', x}\left[
\Tr\left[
\left(
\E_{y', y}[
\tsigma^E_{x' x y' y}(8) - 
\tsigma^E_{x' y' y}(8) -
\tsigma^E_{x' x y'}(8) + 
\tsigma^E_{x' y'}(8)] 
\right) 
\right.
\right.
\\
&  &
~~~~~~~~~~~~~~~
\left.
\left.
\left(
\E_{\hat{y'}, \hat{y}}[
\tsigma^E_{x' x \hat{y'} \hat{y}}(8) - 
\tsigma^E_{x' \hat{y'} \hat{y}}(8) -
\tsigma^E_{x' x \hat{y'}}(8) + 
\tsigma^E_{x' y'}(8)]
\right)
\right]
\right] \\
& = &
\E_{x', x}\left[
\Tr\left[
(\tsigma^E_{x' x}(8) - 
\tsigma^E_{x'}(8) -
\tsigma^E_{x' x}(8) + 
\tsigma^E_{x'}(8)) 
(\tsigma^E_{x' x}(8) - 
\tsigma^E_{x'}(8) -
\tsigma^E_{x' x}(8) + 
\tsigma^E_{x'}(8))
\right]
\right] 
\;=\;
0.
\end{eqnarray*}
A similar situation holds for all cases of fixed choices of
$a'$, $\hat{a'}$, $a$, $\hat{a}$, $b'$, $\hat{b'}$, $b$, $\hat{b}$ 
except the case of $a' = \hat{a'}$, $a = \hat{a}$,
$b' = \hat{b'}$, $b = \hat{b}$. This is the mean zero property originally
described in \cite{Cheng:convexsplit}.
The number of terms in this special case
is $L' L M' M$. For this non-zero special case we have,
\begin{eqnarray*}
\lefteqn{
\E_{x', x, y', y}\left[
\Tr[(\tsigma^E_{x' x y' y}(8) - 
\tsigma^E_{x' y' y}(8) -
\tsigma^E_{x' x y'}(8) + 
\tsigma^E_{x' y'}(8))^2]
\right]
} \\
& = &
\E_{x', x, y', y}\left[
\|\tsigma^E_{x' x y' y}(8) - 
\tsigma^E_{x' y' y}(8) -
\tsigma^E_{x' x y'}(8) + 
\tsigma^E_{x' y'}(8)\|_2^2
\right] \\
& \leq &
\E_{x', x, y', y}\left[
(\|\tsigma^E_{x' x y' y}(8)\|_2 + \|\tsigma^E_{x' y' y}(8)\|_2 +
 \|\tsigma^E_{x' x y'}(8)\|_2 + \|\tsigma^E_{x' y'}(8)\|_2)^2
\right] \\
& \leq &
4 \cdot 
\E_{x', x, y', y}\left[
\|\tsigma^E_{x' x y' y}(8)\|_2^2 + \|\tsigma^E_{x' y' y}(8)\|_2^2 +
\|\tsigma^E_{x' x y'}(8)\|_2^2 + \|\tsigma^E_{x' y'}(8)\|_2^2
\right] \\
& \leq &
16 \cdot 
\E_{x', x, y', y}[\|\tsigma^E_{x' x y' y}(8)\|_2^2]
\;=\;
16 \cdot 
2^{D_2(\bsigma^{X'XY'YE}(8) \| \sigma^{X'XY'Y} \otimes \sigma^E)} \\
& = &
16 \cdot 2^{I_2^\epsilon(X'XY'Y:E)_\sigma},
\end{eqnarray*}
where the first inequality comes from the triangle inequality for
the Schatten-$\ell_2$ norm, the second inequality comes from
Cauchy-Schwarz and the third inequality comes from the monotonicity
of the smooth R\'{e}nyi-2 entropy under trace out. We have thus shown
\begin{eqnarray*}
\lefteqn{
\left(
\E_{\vec{x'}, \vec{x}, \vec{y'}, \vec{y}}\left[
\left\|
\sum_{a'=1}^{L'} \sum_{a=1}^L \sum_{b'=1}^{M'} \sum_{b=1}^M
(\tsigma^E_{x'_{a'} x_{a'a} y'_{b'} y_{b'b}}(8) - 
\tsigma^E_{x'_{a'} y'_{b'} y_{b'b}}(8) -
\tsigma^E_{x'_{a'} x_{a'a} y'_{b'}}(8) + 
\tsigma^E_{x'_{a'} y'_{b'}}(8))
\right\|_2
\right]
\right)^2
} \\
& \leq &
16 \cdot 2^{I_2^\epsilon(X'XY'Y:E)_\sigma}.
~~~~~~~~~~~~~~~~~~~~~~~~~~~~~~~~~~~~~~~~~~~~~~~~~~~~~~~~~~~~~~~~~~~~~
\end{eqnarray*}

In a similar vein, we can bound the other seven terms in the telescoping
sum of Equation~\ref{eq:CMGsmooth3} in terms of their 
respective smooth R\'{e}nyi-2 mutual informations.
We then get
\begin{eqnarray*}
\lefteqn{(\CMG - 24\epsilon)^2} \\
& \leq &
(L'(L+1)M'(M+1))^{-2} \cdot {} \\
&  &
~~~~
(L'LM'M \cdot 16 \cdot 2^{I_2^\epsilon(X'XY'Y:E)_\sigma} +
L'L(L-1) M'M \cdot 16 \cdot 2^{I_2^\epsilon(X'Y'Y:E)_\sigma} \\
&  &
~~~~
{} +
L'LM'M(M-1) \cdot 16 \cdot 2^{I_2^\epsilon(X'XY':E)_\sigma} +
L'L(L-1) M'M(M-1) \cdot 16 \cdot 2^{I_2^\epsilon(X'Y':E)_\sigma} \\
&  &
~~~~
{} +
L'L M'(M'-1)M^2 \cdot 4 \cdot 2^{I_2^\epsilon(X'X:E)_\sigma} +
 (L'-1)L'L^2 M'M \cdot 4 \cdot 2^{I_2^\epsilon(Y'Y:E)_\sigma} \\
&  &
~~~~
{} +
L'L(L-1)(M'-1)M'M^2 \cdot 4 \cdot 2^{I_2^\epsilon(X':E)_\sigma} +
(L'-1)L'L^2 M'M(M-1) \cdot 4 \cdot 2^{I_2^\epsilon(Y':E)_\sigma}) \\
& \leq &
(L'LM'M)^{-1} \cdot 16 \cdot 2^{I_2^\epsilon(X'XY'Y:E)_\sigma} +
(L' M'M)^{-1} \cdot 16 \cdot 2^{I_2^\epsilon(X'Y'Y:E)_\sigma} \\
&  &
~~~~
{} +
 (L'LM')^{-1} \cdot 16 \cdot 2^{I_2^\epsilon(X'XY':E)_\sigma} +
 (L'M')^{-1} \cdot 16 \cdot 2^{I_2^\epsilon(X'Y':E)_\sigma} \\
&  &
~~~~
{} +
 (L'L)^{-1} \cdot 4 \cdot 2^{I_2^\epsilon(X'X:E)_\sigma} +
 (M'M)^{-1} \cdot 4 \cdot 2^{I_2^\epsilon(Y'Y:E)_\sigma} \\
&  &
~~~~
{} +
 (L')^{-1} \cdot 4 \cdot 2^{I_2^\epsilon(X':E)_\sigma} +
 (M')^{-1} \cdot 4 \cdot 2^{I_2^\epsilon(Y':E)_\sigma}.
\end{eqnarray*}
We have thus proved the following fully smooth version of 
Lemma~\ref{lem:nonsmoothCMGcovering}.
\begin{lemma}
\label{lem:smoothCMGcovering}
For the CMG covering problem, let
\begin{eqnarray*}
\log L'
& > &
I_2^\epsilon(X':E)_\sigma + \log \epsilon^{-2}, \\
\log M'
& > &
I_2^\epsilon(Y':E)_\sigma + \log \epsilon^{-2}, \\
\log L' + \log L
& > &
I_2^\epsilon(X'X:E)_\sigma + \log \epsilon^{-2}, \\
\log M' + \log M
& > &
I_2^\epsilon(Y'Y:E)_\sigma + \log \epsilon^{-2}, \\
\log L' + \log M'
& > &
I_2^\epsilon(X'Y':E)_\sigma + \log \epsilon^{-2}, \\
\log L' + \log M' + \log M
& > &
I_2^\epsilon(X'Y'Y:E)_\sigma + \log \epsilon^{-2}, \\
\log L' + \log M' + \log L
& > &
I_2^\epsilon(X'XY':E)_\sigma + \log \epsilon^{-2}, \\
\log L' + \log M' + \log L + \log M
& > &
I_2^\epsilon(X'XY'Y:E)_\sigma + \log \epsilon^{-2}.
\end{eqnarray*}
Then, 
$
\CMG < 31 \epsilon.
$
\end{lemma}

We can now appreciate the minimum requirement of pairwise independence
in the proofs of both the non-ssmooth as well as smooth CMG covering
lemmas. For example, in order to prove
the mean zero property in an expression like the one below,
which arises when $a' = \hat{a'}$, $b' = \hat{b'}$, $b = \hat{b}$ and
$a \neq \hat{a}$ in Equation~\ref{eq:CMGsmooth4}, we need to use
\begin{eqnarray*}
\lefteqn{
\E_{x', x, \hat{x}, y', y}\left[
\Tr\left[
(\tsigma^E_{x' x y' y}(8) - 
\tsigma^E_{x' y' y}(8) -
\tsigma^E_{x' x y'}(8) + 
\tsigma^E_{x' y'}(8)) 
(\tsigma^E_{x' \hat{x} y' y}(8) - 
\tsigma^E_{x' y' y}(8) -
\tsigma^E_{x' \hat{x} y'}(8) + 
\tsigma^E_{x' y'}(8))
\right]
\right]
} \\
& = &
\E_{x', y', y}\left[
\Tr\left[
\left(
\E_{x|x'}[
\tsigma^E_{x' x y' y}(8) - 
\tsigma^E_{x' y' y}(8) -
\tsigma^E_{x' x y'}(8) + 
\tsigma^E_{x' y'}(8)] 
\right) 
\right.
\right.
\\
&  &
~~~~~~~~~~~~~~~
\left.
\left.
\left(
\E_{\hat{x}|x'}[
\tsigma^E_{x' \hat{x} y' y}(8) - 
\tsigma^E_{x' y' y}(8) -
\tsigma^E_{x' \hat{x} y'}(8) + 
\tsigma^E_{x' y'}(8)]
\right)
\right]
\right] \\
& = &
\E_{x', y', y}\left[
\Tr\left[
(\tsigma^E_{x' y' y}(8) - 
\tsigma^E_{x' y' y}(8) -
\tsigma^E_{x' y'}(8) + 
\tsigma^E_{x' y'}(8)) 
(\tsigma^E_{x' y' y}(8) - 
\tsigma^E_{x' y' y}(8) -
\tsigma^E_{x' y'}(8) + 
\tsigma^E_{x' y'}(8))
\right]
\right] \\
& = &
0.
\end{eqnarray*}
This is only possible when in the obfuscating strategy of Alice,
the distribution of $x_{a' a}$ conditioned
on $x'_{a'}$ is pairwise independent from the distribution of
$x_{a' \hat{a}}$ conditioned on $x'_{a'}$ for any pair $a \neq \hat{a}$.
However if this condition is `slightly broken', then the mean zero
property fails. In a companion paper \cite{Sen:flatten}, we will
show via a different sophisticated technique, independent of telescoping,
how to continue to prove fully smooth inner bounds when pairwise
independence is `slightly broken'. That paper will formally define
what `slightly broken' means in this context.

\section{Fully smooth multipartite convex split via telescoping}
\label{sec:convexsplit}
For completeness, we prove the non-smooth as well the fully smooth
multipartite convex split lemma using a conceptually simple matrix
weighted Cauchy-Schwarz inequality. The proof of the smooth version
additionally uses telescoping.

We first reprove the non-smooth multipartite
convex split lemma of \cite{anshu:slepianwolf}. We note that our proof
uses only the matrix weighted Cauchy-Schwarz inequality, which is 
arguably more elementary than the use of tricky identities involving
Shannon, aka Kullback-Leibler, divergences in the original proof of
\cite{anshu:slepianwolf}. The original proof gave inner bounds in 
terms of the
non-smooth  max divergence quantities $D_\infty(\cdot \| \cdot)$. 
In contrast, our elementary approach gives larger inner bounds in terms
of the non-smooth R\'{e}nyi-2 divergence quantities
$D_2(\cdot \| \cdot) \leq D_\infty(\cdot \| \cdot)$. This shows that
convex split is a $D_2$ property and that there is no dependence on
the number of distinct eigenvalues of the quantum states involved,
unlike the results in some earlier works.

Let $A$, $B$ be positive integers and $X$, $Y$, $M$ be Hilbert spaces.
We will work with the $A$-fold tensor power of $X$ denoted by $X^A$.
For any $a \in [A]$, $X_a$ denotes the $a$th Hilbert space isomorphic
to $X$ and $X^{-a}$ denotes the $(A-1)$-fold tensor product of all 
Hilbert spaces isomorphic to $X$ except $X_a$. Similar definitions hold
for $Y^B$, $Y_b$ and $Y^{-b}$.
\begin{lemma}
\label{lem:nonsmoothconvexsplit}
Let $\rho^{XYM}$ be a subnormalised density matrix.
Let $\alpha^X$, $\beta^Y$ be normalised
density matrices such that $\supp(\rho^X) \leq \supp(\alpha^X)$ and
$\supp(\rho^Y) \leq \supp(\beta^Y)$. 
Define the $A$-fold tensor product states
$
\alpha^{X^A} := (\alpha^X)^A, 
\beta^{Y^B} := (\beta^Y)^B, 
$
and the $(A-1)$-fold tensor product states $\alpha^{X^{-a}}$, 
$\beta^{Y^{-b}}$.
Define the {\em convex split state}
\[
\sigma^{X^A Y^B M} :=
(AB)^{-1} 
\sum_{a=1}^A \sum_{b=1}^B
\rho^{X_a Y_b M} \otimes \alpha^{X^{-a}} \otimes \beta^{Y^{-b}},
\]
and the fully decoupled state 
\[
\tau^{X^A Y^B M} := \alpha^{X^A} \otimes \beta^{Y^B} \otimes \rho^M.
\]
Suppose 
\begin{eqnarray*}
\log A 
& > &
D_2(\rho^{XM} \| \alpha^X \otimes \rho^M) + \log \epsilon^{-2} \\
\log B 
& > &
D_2(\rho^{YM} \| \beta^Y \otimes \rho^M) + \log \epsilon^{-2} \\
\log A  + \log B
& > &
D_2(\rho^{XYM} \| \alpha^X \otimes \beta^Y \otimes \rho^M) + 
\log \epsilon^{-2}.
\end{eqnarray*}
Then,
$
\|\sigma^{X^A Y^B M} - \tau^{X^A Y^B M}\|_1 < 2 \epsilon (\Tr \rho).
$
\end{lemma}
\begin{proof}
We follow the lines of the non-smooth inner bound proof for the 
CMG covering problem described in Lemma~\ref{lem:nonsmoothCMGcovering},
using Cauchy-Schwarz inequality with the weighting matrix
$\tau^{X^A Y^B M}$. We can assume that $\Tr \rho = 1$ without loss of
generality, by rescaling the inequality to be proved.

Fact~\ref{fact:matrixCauchySchwarz} gives us
\[
\|\sigma^{X^A Y^B M} - \tau^{X^A Y^B M}\|_1 \leq 
\|((\tau^{X^A Y^B M})^{-1/4} \circ \sigma^{X^A Y^B M}) - 
(\tau^{X^A Y^B M})^{1/2}\|_2.
\]
Then,
\begin{equation}
\label{eq:telescopeconvexsplit1}
\begin{array}{rcl}
\lefteqn{
\|((\tau^{X^A Y^B M})^{-1/4} \circ \sigma^{X^A Y^B M}) - 
(\tau^{X^A Y^B M})^{1/2}\|_2^2
} \\
& = &
\Tr[((\tau^{X^A Y^B M})^{-1/4} \circ \sigma^{X^A Y^B M})^2] -
2 \Tr[((\tau^{X^A Y^B M})^{-1/4} \circ \sigma^{X^A Y^B M})
      (\tau^{X^A Y^B M})^{1/2}] \\
&  &
{} +
\Tr[\tau^{X^A Y^B M}] \\
& = &
\Tr[((\tau^{X^A Y^B M})^{-1/4} \circ \sigma^{X^A Y^B M})^2] -
2 \Tr[\sigma^{X^A Y^B M})] + 1 \\
& = &
\Tr[((\tau^{X^A Y^B M})^{-1/4} \circ \sigma^{X^A Y^B M})^2] - 1 \\
& = &
(AB)^{-2} \cdot 
\Tr\left[
\sum_{a,\hat{a}=1}^A \sum_{b,\hat{b}=1}^B
(\tau^{X^A Y^B M})^{-1/4} 
(\rho^{X_a Y_b M} \otimes \alpha^{X^{-a}} \otimes \beta^{Y^{-b}})
(\tau^{X^A Y^B M})^{-1/2} 
\right. \\
&  &
~~~~~~~~~~~~~~~~~~~~~~~~~~~~~~~~~~~~~~
\left.
(\rho^{X_{\hat{a}} Y_{\hat{b}} M} \otimes 
 \alpha^{X^{-\hat{a}}} \otimes \beta^{Y^{-\hat{b}}})
(\tau^{X^A Y^B M})^{-1/4} 
\right] - 1.
\end{array}
\end{equation}

We analyse the above summation by considering several cases.
Consider the following term for a fixed choice of 
$a \neq \hat{a}$, $b \neq \hat{b}$.
\begin{eqnarray*}
\lefteqn{
\Tr[
(\tau^{X^A Y^B M})^{-1/4} 
(\rho^{X_a Y_b M} \otimes \alpha^{X^{-a}} \otimes \beta^{Y^{-b}})
(\tau^{X^A Y^B M})^{-1/2} 
(\rho^{X_{\hat{a}} Y_{\hat{b}} M} \otimes 
 \alpha^{X^{-\hat{a}}} \otimes \beta^{Y^{-\hat{b}}})
(\tau^{X^A Y^B M})^{-1/4}]
} \\
& = &
\Tr[
(((\alpha^{X_a} \otimes \beta^{Y_b} \otimes \rho^M)^{-1/4} 
  \rho^{X_a Y_b M}
  (\alpha^{X_a} \otimes \beta^{Y_b} \otimes \rho^M)^{-1/4}) 
 \otimes (\alpha^{X^{-a}})^{1/2} 
 \otimes (\beta^{Y^{-b}})^{1/2}) \\
&  &
~~~~~~~~~~~~~
(((\alpha^{X_{\hat{a}}} \otimes \beta^{Y_{\hat{b}}} \otimes\rho^M)^{-1/4} 
  \rho^{X_{\hat{a}} Y_{\hat{b}} M} 
  (\alpha^{X_{\hat{a}}} \otimes \beta^{Y_{\hat{b}}} \otimes\rho^M)^{-1/4})
 \otimes (\alpha^{X^{-\hat{a}}})^{1/2} 
 \otimes (\beta^{Y^{-\hat{b}}})^{1/2})] \\
& = &
\Tr[
(((\one^{X_a} \otimes \one^{Y_b} \otimes \rho^M)^{-1/4} 
  \rho^{X_a Y_b M}
  (\one^{X_a} \otimes \beta^{Y_b} \otimes \rho^M)^{-1/4}) 
 \otimes \one^{X_{\hat{a}}} \otimes \one^{Y_{\hat{b}}}
 \otimes (\alpha^{X^{-a,\hat{a}}})^{1/2} 
 \otimes (\beta^{Y^{-b,\hat{b}}})^{1/2}) \\
&  &
~~~~~~
(((\one^{X_{\hat{a}}} \otimes \one^{Y_{\hat{b}}} \otimes\rho^M)^{-1/4} 
  \rho^{X_{\hat{a}} Y_{\hat{b}} M} 
  (\one^{X_{\hat{a}}} \otimes \one^{Y_{\hat{b}}} \otimes\rho^M)^{-1/4})
 \otimes \one^{X_{a}} \otimes \one^{Y_{b}}
 \otimes (\alpha^{X^{-a,\hat{a}}})^{1/2} 
 \otimes (\beta^{Y^{-b,\hat{b}}})^{1/2})] \\
& = &
\Tr[
((((\one^{X_a} \otimes \one^{Y_b} \otimes \rho^M)^{-1/4} 
   \rho^{X_a Y_b M}
   (\one^{X_a} \otimes \beta^{Y_b} \otimes \rho^M)^{-1/4}) 
  \otimes \one^{X_{\hat{a}}} \otimes \one^{Y_{\hat{b}}}) \\
&  &
~~~~~~
 (((\one^{X_{\hat{a}}} \otimes \one^{Y_{\hat{b}}} \otimes\rho^M)^{-1/4} 
   \rho^{X_{\hat{a}} Y_{\hat{b}} M} 
   (\one^{X_{\hat{a}}} \otimes \one^{Y_{\hat{b}}} \otimes\rho^M)^{-1/4})
  \otimes \one^{X_{a}} \otimes \one^{Y_{b}}))
 \otimes \alpha^{X^{-a,\hat{a}}} \otimes \beta^{Y^{-b,\hat{b}}}] \\
& = &
\Tr[
(((\one^{X_a} \otimes \one^{Y_b} \otimes \rho^M)^{-1/4} 
   \rho^{X_a Y_b M}
   (\one^{X_a} \otimes \one^{Y_b} \otimes \rho^M)^{-1/4}) 
  \otimes \one^{X_{\hat{a}}} \otimes \one^{Y_{\hat{b}}}) \\
&  &
~~~~~~
(((\one^{X_{\hat{a}}} \otimes \one^{Y_{\hat{b}}} \otimes\rho^M)^{-1/4} 
   \rho^{X_{\hat{a}} Y_{\hat{b}} M} 
   (\one^{X_{\hat{a}}} \otimes \one^{Y_{\hat{b}}} \otimes\rho^M)^{-1/4})
  \otimes \one^{X_{a}} \otimes \one^{Y_{b}})] \\
& = &
\Tr[
(((\rho^M)^{-1/4} \rho^{M} (\rho^M)^{-1/4}) 
  \otimes \one^{X_{\hat{a}}} \otimes \one^{Y_{\hat{b}}}) 
((\one^{X_{\hat{a}}} \otimes \one^{Y_{\hat{b}}} \otimes\rho^M)^{-1/4} 
 \rho^{X_{\hat{a}} Y_{\hat{b}} M} 
 (\one^{X_{\hat{a}}} \otimes \one^{Y_{\hat{b}}} \otimes\rho^M)^{-1/4})]\\
& = &
\Tr[
((\rho^M)^{-1/4} \rho^{M} (\rho^M)^{-1/4}) 
((\rho^M)^{-1/4} \rho^{M} (\rho^M)^{-1/4})]
\;=\;
\Tr[\rho^{M}]
\;=\;
1.
\end{eqnarray*}
There are $A(A-1) B(B-1)$ such terms.

Consider the following term for a fixed choice of 
$a = \hat{a}$, $b \neq \hat{b}$.
\begin{eqnarray*}
\lefteqn{
\Tr[
(\tau^{X^A Y^B M})^{-1/4} 
(\rho^{X_a Y_b M} \otimes \alpha^{X^{-a}} \otimes \beta^{Y^{-b}})
(\tau^{X^A Y^B M})^{-1/2} 
(\rho^{X_{a} Y_{\hat{b}} M} \otimes 
\alpha^{X^{-a}} \otimes \beta^{Y^{-\hat{b}}})
(\tau^{X^A Y^B M})^{-1/4}]
} \\
& = &
\Tr[
(((\alpha^{X_a} \otimes \beta^{Y_b} \otimes \rho^M)^{-1/4} 
  \rho^{X_a Y_b M}
  (\alpha^{X_a} \otimes \beta^{Y_b} \otimes \rho^M)^{-1/4}) 
 \otimes (\alpha^{X^{-a}})^{1/2} 
 \otimes (\beta^{Y^{-b}})^{1/2}) \\
&  &
~~~~~~~~~~~~~
(((\alpha^{X_{a}} \otimes \beta^{Y_{\hat{b}}} \otimes\rho^M)^{-1/4} 
  \rho^{X_{a} Y_{\hat{b}} M} 
  (\alpha^{X_{a}} \otimes \beta^{Y_{\hat{b}}} \otimes\rho^M)^{-1/4})
 \otimes (\alpha^{X^{-a}})^{1/2} 
 \otimes (\beta^{Y^{-\hat{b}}})^{1/2})] \\
& = &
\Tr[
(((\alpha^{X_a} \otimes \one^{Y_b} \otimes \rho^M)^{-1/4} 
  \rho^{X_a Y_b M}
  (\alpha^{X_a} \otimes \one^{Y_b} \otimes \rho^M)^{-1/4}) 
 \otimes \one^{Y_{\hat{b}}}
 \otimes (\alpha^{X^{-a}})^{1/2} 
 \otimes (\beta^{Y^{-b,\hat{b}}})^{1/2}) \\
&  &
~~~~~~
(((\alpha^{X_{a}} \otimes \one^{Y_{\hat{b}}} \otimes\rho^M)^{-1/4} 
  \rho^{X_{a} Y_{\hat{b}} M} 
  (\alpha^{X_{a}} \otimes \one^{Y_{\hat{b}}} \otimes\rho^M)^{-1/4})
 \otimes \one^{Y_{b}}
 \otimes (\alpha^{X^{-a}})^{1/2} 
 \otimes (\beta^{Y^{-b,\hat{b}}})^{1/2})] \\
& = &
\Tr[
((((\alpha^{X_a} \otimes \one^{Y_b} \otimes \rho^M)^{-1/4} 
   \rho^{X_a Y_b M}
   (\alpha^{X_a} \otimes \one^{Y_b} \otimes \rho^M)^{-1/4}) 
  \otimes \one^{Y_{\hat{b}}}) \\
&  &
~~~~~~
 (((\alpha^{X_{a}} \otimes \one^{Y_{\hat{b}}} \otimes\rho^M)^{-1/4} 
   \rho^{X_{a} Y_{\hat{b}} M} 
   (\alpha^{X_{a}} \otimes \one^{Y_{\hat{b}}} \otimes\rho^M)^{-1/4})
  \otimes \one^{Y_{b}}))
 \otimes \alpha^{X^{-a}} \otimes \beta^{Y^{-b,\hat{b}}}] \\
& = &
\Tr[
(((\alpha^{X_a} \otimes \one^{Y_b} \otimes \rho^M)^{-1/4} 
   \rho^{X_a Y_b M}
   (\alpha^{X_a} \otimes \beta^{Y_b} \otimes \rho^M)^{-1/4}) 
  \otimes \one^{Y_{\hat{b}}}) \\
&  &
~~~~~~
(((\alpha^{X_{a}} \otimes \one^{Y_{\hat{b}}} \otimes\rho^M)^{-1/4} 
   \rho^{X_{a} Y_{\hat{b}} M} 
   (\alpha^{X_{a}} \otimes \one^{Y_{\hat{b}}} \otimes\rho^M)^{-1/4})
  \otimes \one^{Y_{b}})] \\
& = &
\Tr[
(((\alpha^{X_a} \otimes \rho^M)^{-1/4} 
   \rho^{X_a M}
   (\alpha^{X_a} \otimes \rho^M)^{-1/4}) 
  \otimes \one^{Y_{\hat{b}}}) \\
&  &
~~~~~~
((\alpha^{X_{a}} \otimes \one^{Y_{\hat{b}}} \otimes\rho^M)^{-1/4} 
   \rho^{X_{a} Y_{\hat{b}} M} 
   (\alpha^{X_{a}} \otimes \one^{Y_{\hat{b}}} \otimes\rho^M)^{-1/4})] \\
& = &
\Tr[
((\alpha^{X_a} \otimes \rho^M)^{-1/4} 
 \rho^{X_a M}
 (\alpha^{X_a} \otimes \rho^M)^{-1/4}) 
((\alpha^{X_{a}} \otimes\rho^M)^{-1/4} 
 \rho^{X_{a} M} 
 (\alpha^{X_{a}} \otimes\rho^M)^{-1/4})] \\
& = &
2^{D_2(\rho^{XM} \| \alpha^X \otimes \rho^M)}.
\end{eqnarray*}
There are $A B(B-1)$ such terms.

Consider the following term for a fixed choice of 
$a \neq \hat{a}$, $b = \hat{b}$. We have similarly,
\begin{eqnarray*}
\lefteqn{
\Tr[
(\tau^{X^A Y^B M})^{-1/4} 
(\rho^{X_a Y_b M} \otimes \alpha^{X^{-a}} \otimes \beta^{Y^{-b}})
(\tau^{X^A Y^B M})^{-1/2} 
(\rho^{X_{a} Y_{\hat{b}} M} \otimes 
\alpha^{X^{-a}} \otimes \beta^{Y^{-\hat{b}}})
(\tau^{X^A Y^B M})^{-1/4}]
} \\
& = &
2^{D_2(\rho^{YM} \| \beta^Y \otimes \rho^M)}.
~~~~~~~~~~~~~~~~~~~~~~~~~~~~~~~~~~~~~~~~~~~~~~~~~~~~~~~~~~~~~~~~~~~~~
~~~~~~~~~~~~~~~~~~~~~~~~~~~~~
\end{eqnarray*}
There are $A(A-1) B$ such terms.

Finally, consider the following term for a fixed choice of 
$a = \hat{a}$, $b = \hat{b}$.
\begin{eqnarray*}
\lefteqn{
\Tr[
(\tau^{X^A Y^B M})^{-1/4} 
(\rho^{X_a Y_b M} \otimes \alpha^{X^{-a}} \otimes \beta^{Y^{-b}})
(\tau^{X^A Y^B M})^{-1/2} 
(\rho^{X_{a} Y_{b} M} \otimes 
\alpha^{X^{-a}} \otimes \beta^{Y^{-b}})
(\tau^{X^A Y^B M})^{-1/4}]
} \\
& = &
\Tr[
(((\alpha^{X_a} \otimes \beta^{Y_b} \otimes \rho^M)^{-1/4} 
  \rho^{X_a Y_b M}
  (\alpha^{X_a} \otimes \beta^{Y_b} \otimes \rho^M)^{-1/4}) 
 \otimes (\alpha^{X^{-a}})^{1/2} 
 \otimes (\beta^{Y^{-b}})^{1/2}) \\
&  &
~~~~~~~~~~~~~
(((\alpha^{X_{a}} \otimes \beta^{Y_{b}} \otimes\rho^M)^{-1/4} 
  \rho^{X_{a} Y_{b} M} 
  (\alpha^{X_{a}} \otimes \beta^{Y_{b}} \otimes\rho^M)^{-1/4})
 \otimes (\alpha^{X^{-a}})^{1/2} 
 \otimes (\beta^{Y^{-b}})^{1/2})] \\
& = &
\Tr[
(((\alpha^{X_a} \otimes \beta^{Y_b} \otimes \rho^M)^{-1/4} 
  \rho^{X_a Y_b M}
  (\alpha^{X_a} \otimes \one^{Y_b} \otimes \rho^M)^{-1/4}) \\
&  &
~~~~~~
 ((\alpha^{X_{a}} \otimes \beta^{Y_{b}} \otimes\rho^M)^{-1/4} 
   \rho^{X_{a} Y_{b} M} 
  (\alpha^{X_{a}} \otimes \beta^{Y_{b}} \otimes\rho^M)^{-1/4}))
\otimes \alpha^{X^{-a}} \otimes \beta^{Y^{-b}}] \\
& = &
\Tr[
((\alpha^{X_a} \otimes \beta^{Y_b} \otimes \rho^M)^{-1/4} 
 \rho^{X_a Y_b M}
 (\alpha^{X_a} \otimes \one^{Y_b} \otimes \rho^M)^{-1/4}) \\
&  &
~~~~~~
((\alpha^{X_{a}} \otimes \beta^{Y_{b}} \otimes\rho^M)^{-1/4} 
 \rho^{X_{a} Y_{b} M} 
(\alpha^{X_{a}} \otimes \beta^{Y_{b}} \otimes\rho^M)^{-1/4})]\\
& = &
2^{D_2(\rho^{XYM} \| \alpha^X \otimes \beta^Y \otimes \rho^M)}.
\end{eqnarray*}
There are $A B$ such terms.

Thus we can upper bound the right hand side of 
Equation~\ref{eq:telescopeconvexsplit1} by
\begin{eqnarray*}
\lefteqn{
\Tr\left[
\sum_{a,\hat{a}=1}^A \sum_{b,\hat{b}=1}^B
(\tau^{X^A Y^B M})^{-1/4} 
(\rho^{X_a Y_b M} \otimes \alpha^{X^{-a}} \otimes \beta^{Y^{-b}})
(\tau^{X^A Y^B M})^{-1/2} 
\right.
} \\
&  &
~~~~~~~~~~~~~
\left.
(\rho^{X_{\hat{a}} Y_{\hat{b}} M} \otimes 
 \alpha^{X^{-\hat{a}}} \otimes \beta^{Y^{-\hat{b}}})
(\tau^{X^A Y^B M})^{-1/4} 
\right]  \\
& = &
AB(B-1) \cdot 
2^{D_2(\rho^{XM} \| \alpha^X \otimes \rho^M)} +
A(A-1)B \cdot 
2^{D_2(\rho^{YM} \| \beta^Y \otimes \rho^M)} \\
&  &
~~~~
{} +
AB \cdot 
2^{D_2(\rho^{XYM} \| \alpha^X \otimes \beta^Y \otimes \rho^M)} +
A(A-1)B(B-1) \cdot 1.
\end{eqnarray*}

Plugging it back into Equation~\ref{eq:telescopeconvexsplit1}, we get
\begin{eqnarray*}
\lefteqn{
\|((\tau^{X^A Y^B M})^{-1/4} \circ \sigma^{X^A Y^B M}) - 
(\tau^{X^A Y^B M})^{1/2}\|_2^2
} \\
& = &
(AB)^{-2} \cdot 
\Tr\left[
\sum_{a,\hat{a}=1}^A \sum_{b,\hat{b}=1}^B
(\tau^{X^A Y^B M})^{-1/4} 
(\rho^{X_a Y_b M} \otimes \alpha^{X^{-a}} \otimes \beta^{Y^{-b}})
(\tau^{X^A Y^B M})^{-1/2} 
\right. \\
&  &
~~~~~~~~~~~~~~~~~~~~~~~~~~~~~~~~~~~~~~
\left.
(\rho^{X_{\hat{a}} Y_{\hat{b}} M} \otimes 
 \alpha^{X^{-\hat{a}}} \otimes \beta^{Y^{-\hat{b}}})
(\tau^{X^A Y^B M})^{-1/4} 
\right] - 1 \\
& = &
(AB)^{-2} \cdot 
(AB(B-1) \cdot 
2^{D_2(\rho^{XM} \| \alpha^X \otimes \rho^M)} +
A(A-1)B \cdot 
2^{D_2(\rho^{YM} \| \beta^Y \otimes \rho^M)}  \\
&  &
~~~~~~~~~~~~~~~
{} +
AB \cdot 
2^{D_2(\rho^{XYM} \| \alpha^X \otimes \beta^Y \otimes \rho^M)} +
A(A-1)B(B-1) \cdot 1) - 1 \\
& \leq &
A^{-1} \cdot 
2^{D_2(\rho^{XM} \| \alpha^X \otimes \rho^M)} +
B^{-1} \cdot 
2^{D_2(\rho^{YM} \| \beta^Y \otimes \rho^M)}  +
(AB)^{-1} \cdot 
2^{D_2(\rho^{XYM} \| \alpha^X \otimes \beta^Y \otimes \rho^M)} +
1 - 1 \\
&  =   &
A^{-1} \cdot 
2^{D_2(\rho^{XM} \| \alpha^X \otimes \rho^M)} +
B^{-1} \cdot 
2^{D_2(\rho^{YM} \| \beta^Y \otimes \rho^M)}  +
(AB)^{-1} \cdot 
2^{D_2(\rho^{XYM} \| \alpha^X \otimes \beta^Y \otimes \rho^M)}.
\end{eqnarray*}

This finishes the proof of the lemma.
\end{proof}

Adding telescoping to the above proof, we obtain for the first time 
the following fully smooth 
version of the multipartite convex split lemma. We believe our proof below
is substantially simpler than the treatment given in 
\cite{Cheng:convexsplit}, while being more general in the sense that
our proof gives a fully smooth one shot result.
\begin{proposition}
\label{prop:smoothconvexsplit}
Let $\rho^{XYM}$ be a subnormalised density matrix.
Let $\alpha^X$, $\beta^Y$ be normalised
density matrices such that $\supp(\rho^X) \leq \supp(\alpha^X)$ and 
$\supp(\rho^Y) \leq \supp(\beta^Y)$. 
Define the $A$-fold tensor product states
$
\alpha^{X^A} := (\alpha^X)^{\otimes A}, 
\beta^{Y^B} := (\beta^Y)^{\otimes B}, 
$
and the $(A-1)$-fold tensor product states $\alpha^{X^{-a}}$, 
$\beta^{Y^{-b}}$ for any $a \in [A]$, $b \in [B]$.
Define the {\em convex split state}
\[
\sigma^{X^A Y^B M} :=
(AB)^{-1} 
\sum_{a=1}^A \sum_{b=1}^B
\rho^{X_a Y_b M} \otimes \alpha^{X^{-a}} \otimes \beta^{Y^{-b}},
\]
and the fully decoupled state 
\[
\tau^{X^A Y^B M} := \alpha^{X^A} \otimes \beta^{Y^B} \otimes \rho^M.
\]
Suppose 
\begin{eqnarray*}
\log A 
& > &
D_2^\epsilon(\rho^{XM} \| \alpha^X \otimes \rho^M) + \log \epsilon^{-2} \\
\log B 
& > &
D_2^\epsilon(\rho^{YM} \| \beta^Y \otimes \rho^M) + \log \epsilon^{-2} \\
\log A  + \log B
& > &
D_2^\epsilon(\rho^{XYM} \| \alpha^X \otimes \beta^Y \otimes \rho^M) + 
\log \epsilon^{-2}.
\end{eqnarray*}
Then,
$
\|\sigma^{X^A Y^B M} - \tau^{X^A Y^B M}\|_1 < 16 \epsilon (\Tr \rho).
$
\end{proposition}
\begin{proof}
We work along the lines of the proof of Lemma
\ref{lem:smoothCMGcovering}, using the weighting matrix 
$\tau^{X^A Y^B M}$. There are fewer terms in the telescoping sum now:
only analogues of terms numbered (5), (1) and (2) in the proof of
Lemma \ref{lem:smoothCMGcovering} are needed. 

Let $\rho^{XYM}(3)$ denote the optimising state in the definition
of $D^\epsilon_2(\cdot \| \cdot)$ in the third inequality of the
desired inner bound,
$\rho^{XM}(1)$ the optimising state in the definition
of $D^\epsilon_2(\cdot \| \cdot)$ in the first inequality of the
desired inner bound,
$\rho^{YM}(2)$ the optimising state in the definition
of $D^\epsilon_2(\cdot \| \cdot)$ in the second inequality of the
desired inner bound. Define
\[
\begin{array}{c}
\trho^{X_a Y_b M}(3) :=
(\alpha^{X_a} \otimes \beta^{Y_b} \otimes \rho^M)^{-1/4} \circ
\rho^{X_a Y_b M}(3), \\
\trho^{X_a M}(1) :=
(\alpha^{X_a} \otimes \rho^M)^{-1/4} \circ
\rho^{X_a M}(1), ~~
\trho^{Y_b M}(2) :=
(\beta^{Y_b} \otimes \rho^M)^{-1/4} \circ
\rho^{Y_b M}(2).
\end{array}
\]
These states are close to their corresponding marginals of $\rho^{XYM}$.
\[
\|\trho^{X M}(1) - \rho^{XM}\|_1 \leq \epsilon, ~
\|\trho^{Y M}(2) - \rho^{YM}\|_1 \leq \epsilon, ~
\|\trho^{XYM}(3) - \rho^{XYM}\|_1 \leq \epsilon.
\]

We set up the telescoping sum and apply the Cauchy-Schwarz
inequality Fact~\ref{fact:matrixCauchySchwarz} as follows.
\begin{eqnarray*}
\lefteqn{
\|\sigma^{X^A Y^B M} - \tau^{X^A Y^B M}\|_1
} \\
& = &
(AB)^{-1} \cdot 
\left\|
\sum_{a=1}^A \sum_{b=1}^B
((\rho^{X_a Y_b M} - \alpha^{X_a} \otimes \beta^{Y_b} \otimes \rho^M)
 \otimes \alpha^{X^{-a}} \otimes \beta^{Y^{-b}})
\right\|_1 \\
& \leq &
(AB)^{-1} \cdot 
\left\|
\sum_{a=1}^A \sum_{b=1}^B
((\rho^{X_a Y_b M} - 
  \alpha^{X_a} \otimes \rho^{Y_b M} -
  \beta^{Y_b} \otimes \rho^{X_a M} +
  \alpha^{X_a} \otimes \beta^{Y_b} \otimes \rho^M) 
\right. \\
&  &
~~~~~~~~~~~~~~~~~~~~~~~~~~~~~~~~~~~
\left.
{} \otimes \alpha^{X^{-a}} \otimes \beta^{Y^{-b}})
\right\|_1 \\
& &
~~~
{} +
(AB)^{-1} \cdot 
\left\|
\sum_{a=1}^A \sum_{b=1}^B
((\alpha^{X_a} \otimes \rho^{Y_b M} -
  \alpha^{X_a} \otimes \beta^{Y_b} \otimes \rho^M)
 \otimes \alpha^{X^{-a}} \otimes \beta^{Y^{-b}})
\right\|_1 \\
& &
~~~
{} +
(AB)^{-1} \cdot 
\left\|
\sum_{a=1}^A \sum_{b=1}^B
((\beta^{Y_b} \otimes \rho^{X_a M} -
  \alpha^{X_a} \otimes \beta^{Y_b} \otimes \rho^M)
 \otimes \alpha^{X^{-a}} \otimes \beta^{Y^{-b}})
\right\|_1 \\
& \leq &
8\epsilon + {} 
\hspace*{12.5cm}
(\arabic{equation})
\label{eq:telescopeconvexsplit2}
\stepcounter{equation} \\
&  &
(AB)^{-1} \cdot 
\left\|
\sum_{a=1}^A \sum_{b=1}^B
((\rho^{X_a Y_b M}(3) - 
  \alpha^{X_a} \otimes \rho^{Y_b M}(3) -
  \beta^{Y_b} \otimes \rho^{X_a M}(3) 
\right. \\
&  &
~~~~~~~~~~~~~~~~~~~~~~~~~~~~~~~~~~~~~~
\left.
{} +
  \alpha^{X_a} \otimes \beta^{Y_b} \otimes \rho^M(3))
 \otimes \alpha^{X^{-a}} \otimes \beta^{Y^{-b}})
\right\|_1 \\
& &
{} +
(AB)^{-1} \cdot 
\left\|
\sum_{a=1}^A \sum_{b=1}^B
((\alpha^{X_a} \otimes \rho^{Y_b M}(2) -
  \alpha^{X_a} \otimes \beta^{Y_b} \otimes \rho^M(2))
 \otimes \alpha^{X^{-a}} \otimes \beta^{Y^{-b}})
\right\|_1 \\
& &
{} +
(AB)^{-1} \cdot 
\left\|
\sum_{a=1}^A \sum_{b=1}^B
((\beta^{Y_b} \otimes \rho^{X_a M}(1) -
  \alpha^{X_a} \otimes \beta^{Y_b} \otimes \rho^M(1))
 \otimes \alpha^{X^{-a}} \otimes \beta^{Y^{-b}})
\right\|_1 \\
& \leq &
8 \epsilon + {} \\
&  &
(AB)^{-1} \cdot 
\left\|
\sum_{a=1}^A \sum_{b=1}^B
((\trho^{X_a Y_b M}(3) - 
  (\alpha^{X_a})^{1/2} \otimes \trho^{Y_b M}(3) -
  (\beta^{Y_b})^{1/2} \otimes \trho^{X_a M}(3) 
\right. \\
&  &
~~~~~~~~~~~~~~~~~~~~~~~~~~~~~~
\left.
{} +
  (\alpha^{X_a})^{1/2} \otimes (\beta^{Y_b})^{1/2} \otimes \trho^M(3)) 
 \otimes (\alpha^{X^{-a}})^{1/2} \otimes (\beta^{Y^{-b}})^{1/2})
\right\|_2 \\
& &
~~~
{} +
(AB)^{-1} \cdot 
\left\|
\sum_{a=1}^A \sum_{b=1}^B
(((\alpha^{X_a})^{1/2} \otimes \trho^{Y_b M}(2) -
  (\alpha^{X_a})^{1/2} \otimes (\beta^{Y_b})^{1/2} \otimes \trho^M(2)) 
\right. \\
&  &
~~~~~~~~~~~~~~~~~~~~~~~~~~~~~~~~~~~~~~~~~~~
\left.
{} \otimes (\alpha^{X^{-a}})^{1/2} \otimes (\beta^{Y^{-b}})^{1/2})
\right\|_2 \\
& &
~~~
{} +
(AB)^{-1} \cdot 
\left\|
\sum_{a=1}^A \sum_{b=1}^B
(((\beta^{Y_b})^{1/2} \otimes \trho^{X_a M}(1) -
  (\alpha^{X_a})^{1/2} \otimes (\beta^{Y_b})^{1/2} \otimes \trho^M(1))
\right. \\
&  &
~~~~~~~~~~~~~~~~~~~~~~~~~~~~~~~~~~~~~~~~~~~
\left.
{} \otimes (\alpha^{X^{-a}})^{1/2} \otimes (\beta^{Y^{-b}})^{1/2})
\right\|_2.
\end{eqnarray*}

We now analyse the first term in the Schatten-$\ell_2$ telescoping upper 
bound above.
\begin{eqnarray*}
\lefteqn{
\left\|
\sum_{a=1}^A \sum_{b=1}^B
((\trho^{X_a Y_b M}(3) - 
  (\alpha^{X_a})^{1/2} \otimes \trho^{Y_b M}(3) -
  (\beta^{Y_b})^{1/2} \otimes \trho^{X_a M}(3) 
\right.
} \\
&  &
~~~~~~~~~~~~~~~~~~
\left.
{} +
  (\alpha^{X_a})^{1/2} \otimes (\beta^{Y_b})^{1/2} \otimes \trho^M(3))
\otimes (\alpha^{X^{-a}})^{1/2} \otimes (\beta^{Y^{-b}})^{1/2})
\right\|_2^2 \\
& = &
\sum_{a,\hat{a}=1}^A \sum_{b,\hat{b}=1}^B
\Tr\left[
((\trho^{X_a Y_b M}(3) - 
  (\alpha^{X_a})^{1/2} \otimes \trho^{Y_b M}(3) -
  (\beta^{Y_b})^{1/2} \otimes \trho^{X_a M}(3)
\right. 
\hspace*{2.5cm}
(\arabic{equation})
\label{eq:telescopeconvexsplit3}
\stepcounter{equation} \\
& &
~~~~~~~~~~~~~~~~~~~~~~~~~~
{} +
  (\alpha^{X_a})^{1/2} \otimes (\beta^{Y_b})^{1/2} \otimes \trho^M(3))
 \otimes (\alpha^{X^{-a}})^{1/2} \otimes (\beta^{Y^{-b}})^{1/2}) \\
& &
~~~~~~~~~~~~~~~~~~~~~~~~~
((\trho^{X_{\hat{a}} Y_{\hat{b}} M}(3) - 
 (\alpha^{X_{\hat{a}}})^{1/2} \otimes \trho^{Y_{\hat{b}} M}(3) -
 (\beta^{Y_{\hat{b}}})^{1/2} \otimes \trho^{X_{\hat{a}} M}(3) \\
& &
~~~~~~~~~~~~~~~~~~~~~~~~~~~
\left.
{} +
(\alpha^{X_{\hat{a}}})^{1/2} \otimes (\beta^{Y_{\hat{b}}})^{1/2} 
 \otimes \trho^M(3))
\otimes (\alpha^{X^{-\hat{a}}})^{1/2} 
\otimes (\beta^{Y^{-\hat{b}}})^{1/2})
\right] \\
\end{eqnarray*}

Consider an expression corresponding to a fixed choice of
$a \neq \hat{a}$, $b = \hat{b}$. We get,
\begin{eqnarray*}
\lefteqn{
\Tr\left[
((\trho^{X_a Y_b M}(3) - 
  (\alpha^{X_a})^{1/2} \otimes \trho^{Y_b M}(3) -
  (\beta^{Y_b})^{1/2} \otimes \trho^{X_a M}(3)
\right.
} \\
& &
~~~~~~~
{} +
  (\alpha^{X_a})^{1/2} \otimes (\beta^{Y_b})^{1/2} \otimes \trho^M(3))
 \otimes (\alpha^{X^{-a}})^{1/2} \otimes (\beta^{Y^{-b}})^{1/2}) \\
& &
~~~~~
((\trho^{X_{\hat{a}} Y_{b} M}(3) - 
 (\alpha^{X_{\hat{a}}})^{1/2} \otimes \trho^{Y_{b} M}(3) -
 (\beta^{Y_{b}})^{1/2} \otimes \trho^{X_{\hat{a}} M}(3) \\
& &
~~~~~~~
\left.
{} +
(\alpha^{X_{\hat{a}}})^{1/2} \otimes (\beta^{Y_{b}})^{1/2} 
 \otimes \trho^M(3))
\otimes (\alpha^{X^{-\hat{a}}})^{1/2} 
\otimes (\beta^{Y^{-b}})^{1/2})
\right] \\
& =  &
\Tr\left[
((((\alpha^{X_a})^{1/4} \trho^{X_a Y_b M}(3) (\alpha^{X_a})^{1/4} - 
   \alpha^{X_a} \otimes \trho^{Y_b M}(3) 
\right. \\
&  &
~~~~~~~~~~~~~~~
{} -
   (\beta^{Y_b})^{1/2} \otimes 
   ((\alpha^{X_a})^{1/4} \trho^{X_a M}(3) (\alpha^{X_a})^1/4) - \\
& &
~~~~~~~~~~~~~~~
{} +
  \alpha^{X_a} \otimes (\beta^{Y_b})^{1/2} \otimes \trho^M(3))
 \otimes (\alpha^{X^{-a}})^{1/2} \otimes (\beta^{Y^{-b}})^{1/2}) \\
& &
~~~~~
((\trho^{X_{\hat{a}} Y_{b} M}(3) - 
 (\alpha^{X_{\hat{a}}})^{1/2} \otimes \trho^{Y_{b} M}(3) -
 (\beta^{Y_{b}})^{1/2} \otimes \trho^{X_{\hat{a}} M}(3) \\
& &
~~~~~~~
\left.
{} +
(\alpha^{X_{\hat{a}}})^{1/2} \otimes (\beta^{Y_{b}})^{1/2} 
 \otimes \trho^M(3))
\otimes \one^{X_a}
\otimes (\alpha^{X^{-a,\hat{a}}})^{1/2} 
\otimes (\beta^{Y^{-b}})^{1/2})
\right] \\
& =  &
\Tr\left[
((((\one^{X_a} \otimes (\beta^Y)^{-1/4} \otimes (\rho^M)^{-1/4}) 
   \circ \rho^{X_a Y_b M}(3)) -
  \alpha^{X_a} \otimes \trho^{Y_b M}(3) 
\right. \\
&  &
~~~~~~~~~~~~~~~
{} -
   (\beta^{Y_b})^{1/2} \otimes 
   ((\one^{X_a} \otimes (\rho^M)^{-1/4}) \circ \rho^{X_a M}(3)) - \\
& &
~~~~~~~~~~~~~~~
{} +
  \alpha^{X_a} \otimes (\beta^{Y_b})^{1/2} \otimes \trho^M(3))
 \otimes (\alpha^{X^{-a}})^{1/2} \otimes (\beta^{Y^{-b}})^{1/2}) \\
& &
~~~~~
((\trho^{X_{\hat{a}} Y_{b} M}(3) - 
 (\alpha^{X_{\hat{a}}})^{1/2} \otimes \trho^{Y_{b} M}(3) -
 (\beta^{Y_{b}})^{1/2} \otimes \trho^{X_{\hat{a}} M}(3) \\
& &
~~~~~~~
\left.
{} +
(\alpha^{X_{\hat{a}}})^{1/2} \otimes (\beta^{Y_{b}})^{1/2} 
 \otimes \trho^M(3))
\otimes \one^{X_a}
\otimes (\alpha^{X^{-a,\hat{a}}})^{1/2} 
\otimes (\beta^{Y^{-b}})^{1/2})
\right] \\
& =  &
\Tr\left[
(((((\beta^Y)^{-1/4} \otimes (\rho^M)^{-1/4}) 
   \circ \rho^{Y_b M}(3)) - \trho^{Y_b M}(3) 
\right. \\
&  &
~~~~~~~~~~~~~~~
{} -
   (\beta^{Y_b})^{1/2} \otimes 
   ((\rho^M)^{-1/4} \circ \rho^{M}(3)) - \\
& &
~~~~~~~~~~~~~~~
{} +
  (\beta^{Y_b})^{1/2} \otimes \trho^M(3))
 \otimes (\alpha^{X^{-a}})^{1/2} \otimes (\beta^{Y^{-b}})^{1/2}) \\
& &
~~~~~
((\trho^{X_{\hat{a}} Y_{b} M}(3) - 
 (\alpha^{X_{\hat{a}}})^{1/2} \otimes \trho^{Y_{b} M}(3) -
 (\beta^{Y_{b}})^{1/2} \otimes \trho^{X_{\hat{a}} M}(3) \\
& &
~~~~~~~
\left.
{} +
(\alpha^{X_{\hat{a}}})^{1/2} \otimes (\beta^{Y_{b}})^{1/2} 
 \otimes \trho^M(3))
\otimes (\alpha^{X^{-a,\hat{a}}})^{1/2} 
\otimes (\beta^{Y^{-b}})^{1/2})
\right] \\
& =  &
\Tr\left[
((\trho^{Y_b M}(3) - \trho^{Y_b M}(3) -
   (\beta^{Y_b})^{1/2} \otimes \trho^{M}(3) +
  (\beta^{Y_b})^{1/2} \otimes \trho^M(3))
 \otimes (\alpha^{X^{-a}})^{1/2} \otimes (\beta^{Y^{-b}})^{1/2}) 
\right.\\
& &
~~~~~~~~
((\trho^{X_{\hat{a}} Y_{b} M}(3) - 
 (\alpha^{X_{\hat{a}}})^{1/2} \otimes \trho^{Y_{b} M}(3) -
 (\beta^{Y_{b}})^{1/2} \otimes \trho^{X_{\hat{a}} M}(3) \\
& &
~~~~~~~~~~~
\left.
{} +
(\alpha^{X_{\hat{a}}})^{1/2} \otimes (\beta^{Y_{b}})^{1/2} 
 \otimes \trho^M(3))
\otimes (\alpha^{X^{-a,\hat{a}}})^{1/2} 
\otimes (\beta^{Y^{-b}})^{1/2})
\right] 
\;=\;
0.
\end{eqnarray*}
This is an example of the mean zero property of the telescoping 
formulation. Similarly, 
any expression corresponding to a fixed choice of
$a = \hat{a}$, $b \neq \hat{b}$, or $a \neq \hat{a}$, $b \neq \hat{b}$
evaluates to zero. The only terms that survive corresponding to
fixed choices $a = \hat{a}$, $b = \hat{b}$, and there are $AB$ of them.
Such an expression evaluates to
\begin{eqnarray*}
\lefteqn{
\Tr\left[
((\trho^{X_a Y_b M}(3) - 
  (\alpha^{X_a})^{1/2} \otimes \trho^{Y_b M}(3) -
  (\beta^{Y_b})^{1/2} \otimes \trho^{X_a M}(3)
\right.
} \\
& &
~~~~~~~
{} +
  (\alpha^{X_a})^{1/2} \otimes (\beta^{Y_b})^{1/2} \otimes \trho^M(3))
 \otimes (\alpha^{X^{-a}})^{1/2} \otimes (\beta^{Y^{-b}})^{1/2}) \\
& &
~~~~~
((\trho^{X_{a} Y_{b} M}(3) - 
 (\alpha^{X_{a}})^{1/2} \otimes \trho^{Y_{b} M}(3) -
 (\beta^{Y_{b}})^{1/2} \otimes \trho^{X_{a} M}(3) \\
& &
~~~~~~~
\left.
{} +
(\alpha^{X_{a}})^{1/2} \otimes (\beta^{Y_{b}})^{1/2} 
 \otimes \trho^M(3))
\otimes (\alpha^{X^{-a}})^{1/2} 
\otimes (\beta^{Y^{-b}})^{1/2})
\right] \\
& = &
\Tr\left[
(\trho^{X_a Y_b M}(3) - 
  (\alpha^{X_a})^{1/2} \otimes \trho^{Y_b M}(3) -
  (\beta^{Y_b})^{1/2} \otimes \trho^{X_a M}(3)
\right. \\
& &
~~~~~~~
\left.
{} +
  (\alpha^{X_a})^{1/2} \otimes (\beta^{Y_b})^{1/2} \otimes \trho^M(3))^2
 \otimes \alpha^{X^{-a}} \otimes \beta^{Y^{-b}} 
\right]\\
& = &
\Tr\left[
(\trho^{X_a Y_b M}(3) - 
  (\alpha^{X_a})^{1/2} \otimes \trho^{Y_b M}(3) -
  (\beta^{Y_b})^{1/2} \otimes \trho^{X_a M}(3)
\right. \\
& &
~~~~~~~
\left.
{} +
  (\alpha^{X_a})^{1/2} \otimes (\beta^{Y_b})^{1/2} \otimes \trho^M(3))^2
\right]\\
& = &
\left\|
(\trho^{X_a Y_b M}(3) - 
  (\alpha^{X_a})^{1/2} \otimes \trho^{Y_b M}(3) -
  (\beta^{Y_b})^{1/2} \otimes \trho^{X_a M}(3) +
  (\alpha^{X_a})^{1/2} \otimes (\beta^{Y_b})^{1/2} \otimes \trho^M(3))^2
\right\|_2^2 \\
& \leq &
\left(
\|\trho^{X_a Y_b M}(3)\|_2 + 
\|(\alpha^{X_a})^{1/2} \otimes \trho^{Y_b M}(3)\|_2 +
\|(\beta^{Y_b})^{1/2} \otimes \trho^{X_a M}(3)\|_2 +
\|(\alpha^{X_a})^{1/2} \otimes (\beta^{Y_b})^{1/2}\otimes \trho^M(3)\|_2
\right)^2 \\
& = &
\left(
\|\trho^{X_a Y_b M}(3)\|_2 + 
\|\trho^{Y_b M}(3)\|_2 +
\|\trho^{X_a M}(3)\|_2 +
\|\trho^M(3)\|_2
\right)^2 \\
& \leq &
16 \cdot \|\trho^{X_a Y_b M}(3)\|_2^2 
\;=\;
16 \cdot 
2^{D_2(\rho^{XYM}(3) \| 
       \alpha^X \otimes \beta^Y \otimes \rho^M)}
\;=\;
16 \cdot 
2^{D_2^\epsilon(\rho^{XYM} \| 
                \alpha^X \otimes \beta^Y \otimes \rho^M)}.
\end{eqnarray*}
The first inequality follows from the triangle inequality for the
Schatten-$\ell_2$ norm, while the second inequality follows from
the monotonicity of $D_2(\cdot \| \cdot)$ under trace out.

We have thus shown that the first term in the telescoping 
Schatten-$\ell_2$ inequality is upper bounded by
\begin{equation}
\label{eq:telescopeconvexsplit4}
\begin{array}{rcl}
\lefteqn{
\left\|
\sum_{a=1}^A \sum_{b=1}^B
((\trho^{X_a Y_b M}(3) - 
  (\alpha^{X_a})^{1/2} \otimes \trho^{Y_b M}(3) -
  (\beta^{Y_b})^{1/2} \otimes \trho^{X_a M}(3) 
\right.
} \\
&  &
~~~~~~~~~~~~~~~~~~
\left.
{} +
  (\alpha^{X_a})^{1/2} \otimes (\beta^{Y_b})^{1/2} \otimes \trho^M(3))
\otimes (\alpha^{X^{-a}})^{1/2} \otimes (\beta^{Y^{-b}})^{1/2})
\right\|_2^2 \\
& \leq &
16 A B \cdot
2^{D_2^\epsilon(\rho^{XYM} \| 
                \alpha^X \otimes \beta^Y \otimes \rho^M)}.
\end{array}
\end{equation}
The next two terms are similarly upper bounded by
\begin{equation}
\label{eq:telescopeconvexsplit5}
\begin{array}{rcl}
\lefteqn{
\left\|
\sum_{a=1}^A \sum_{b=1}^B
(((\alpha^{X_a})^{1/2} \otimes \trho^{Y_b M}(2) -
  (\alpha^{X_a})^{1/2} \otimes (\beta^{Y_b})^{1/2} \otimes \trho^M(2))
\otimes (\alpha^{X^{-a}})^{1/2} \otimes (\beta^{Y^{-b}})^{1/2})
\right\|_2^2
} \\
& \leq &
4 A B (B-1) \cdot
2^{D_2^\epsilon(\rho^{XM} \| \alpha^X \otimes \rho^M)},
~~~~~~~~~~~~~~~~~~~~~~~~~~~~~~~~~~~~~~~~~~~~~~~~~~~~~~~~~~~~~~~~~~~
\end{array}
\end{equation}
and
\begin{equation}
\label{eq:telescopeconvexsplit6}
\begin{array}{rcl}
\lefteqn{
\left\|
\sum_{a=1}^A \sum_{b=1}^B
(((\beta^{Y_b})^{1/2} \otimes \trho^{X_a M}(1) -
  (\alpha^{X_a})^{1/2} \otimes (\beta^{Y_b})^{1/2} \otimes \trho^M(1))
\otimes (\alpha^{X^{-a}})^{1/2} \otimes (\beta^{Y^{-b}})^{1/2})
\right\|_2^2
} \\
& \leq &
4 A(A-1) B \cdot
2^{D_2^\epsilon(\rho^{YM} \| \beta^Y \otimes \rho^M)}.
~~~~~~~~~~~~~~~~~~~~~~~~~~~~~~~~~~~~~~~~~~~~~~~~~~~~~~~~~~~~~~~~~~~
\end{array}
\end{equation}

Putting Equations~\ref{eq:telescopeconvexsplit4},
\ref{eq:telescopeconvexsplit5}, \ref{eq:telescopeconvexsplit6} back into
Equation~\ref{eq:telescopeconvexsplit2}, we get
\begin{eqnarray*}
\lefteqn{
\|\sigma^{X^A Y^B M} - \tau^{X^A Y^B M}\|_1
} \\
& \leq &
8\epsilon +
(AB)^{-1} 
(
\sqrt{16 A B \cdot
2^{D_2^\epsilon(\rho^{XYM} \| 
                \alpha^X \otimes \beta^Y \otimes \rho^M)}
} +
\sqrt{4 A B(B-1) \cdot
2^{D_2^\epsilon(\rho^{XM} \| \alpha^X \otimes \rho^M)}
} \\
&  &
~~~~~~~~~~~~~~~~~~~~~~~~
{} +
\sqrt{4 A(A-1) B \cdot
2^{D_2^\epsilon(\rho^{YM} \| \beta^Y \otimes \rho^M)}
}
) \\
& \leq &
8\epsilon +
\sqrt{16 (A B)^{-1} \cdot
2^{D_2^\epsilon(\rho^{XYM} \| 
                \alpha^X \otimes \beta^Y \otimes \rho^M)}
} +
\sqrt{4 A^{-1}  \cdot
2^{D_2^\epsilon(\rho^{XM} \| \alpha^X \otimes \rho^M)}
} \\
&  &
~~~~~~~~~~~~~~~~
{} +
\sqrt{4 B^{-1} \cdot
2^{D_2^\epsilon(\rho^{YM} \| \beta^Y \otimes \rho^M)}
}
).
\end{eqnarray*}

Putting in the lower bounds on $\log A$, $\log B$ and $\log AB$
into the above inequality finishes the proof of the theorem.
\end{proof}

A similar theorem can be
obtained for any constant number of  parties $X, Y, Z, \ldots,$, with a
completely analogous proof. We state the result below for completeness.
\begin{theorem}
\label{thm:smoothconvexsplitmanyparties}
Let $k$ be a positive integer. Let $X_1, \ldots, X_k$ be $k$ Hilbert
spaces. For any subset
$S \subseteq [k]$, let $X_S := \otimes_{s \in S} X_s$.
Let $\rho^{X_{[k]}M}$ be a subnormalised density matrix.
Let $\alpha^{X_1}, \ldots, \alpha^{X_k}$ be normalised
density matrices such that 
$\supp(\rho^{X_i}) \leq \supp(\alpha^{X_i})$.
For any subset $S \subseteq [k]$, let 
$\alpha^{X_S} := \otimes_{s \in S} \alpha^{X_s}$.
Let $A_1, \ldots, A_k$ be positive integers.
Define for each $i \in [k]$, $X_i^{A_i} := X_i^{\otimes A_i}$.
The $j$th tensor multiplicand in the previous expression will be denoted
by $A_{ij}$ i.e. $X_i^{A_i} = \otimes_{j=1}^{A_i} X_{ij}$.
Define the $|A_i|$-fold tensor product state
$
\alpha^{X_i^{A_i}} := \otimes_{j=1}^{A_i} \alpha^{X_{ij}},
$
and the $|A_i|$ many $(A_i-1)$-fold tensor product states 
$
\alpha^{X_i^{-a_i}} := \otimes_{j=1,j\neq a_i}^{A_i} \alpha^{X_{ij}}
$
for any $a_i \in [A_i]$.
Define the {\em convex split state}
\[
\sigma^{X_1^{A_1} \cdots X_k^{A_k} M} :=
(A_1 \cdots A_k)^{-1} 
\sum_{a_1=1}^{A_1} \cdots \sum_{a_k=1}^{A_k}
\rho^{X_{a_1} \cdots X_{a_k} M} \otimes 
\alpha^{X_1^{-a_1}} \otimes \cdots \otimes \alpha^{X_k^{-a_k}},
\]
and the fully decoupled state 
\[
\tau^{X_1^{A_1} \cdots X_k^{A_k} M} := 
\alpha^{X_1^{A_1}} \otimes \cdots \otimes \alpha^{X_k^{A_k}} 
\otimes \rho^M.
\]
Suppose for each non-empty subset $\{\} \neq S \subseteq [k]$,
\begin{eqnarray*}
\sum_{s \in S} \log A_s 
& > &
D_2^\epsilon(\rho^{X_S M} \| \alpha^{X_S} \otimes \rho^M) + 
\log \epsilon^{-2}.
\end{eqnarray*}
Then,
\[
\|\sigma^{X_1^{A_1} \cdots X_k^{A_k} M} - 
  \tau^{X_1^{A_1} \cdots X_k^{A_k} M}\|_1 < 
2(3^k - 1) \epsilon (\Tr \rho).
\]
\end{theorem}

Observe that Theorem~\ref{thm:smoothconvexsplitmanyparties}
applied to classical quantum states can alternatively be viewed
as a {\em soft classical quantum smooth multipartite covering lemma
in expectation}.
\begin{proposition}
\label{prop:smoothcovering}
Let $k$ be a positive integer. Let $X_1, \ldots, X_k$ be $k$ classical
alphabets. For any subset
$S \subseteq [k]$, let $X_S := (X_s)_{s \in S}$.
Let $p^{X_{[k]}}$ be a normalised probability distribution on $X_{[k]}$.
The notation $p^{X_S}$ denotes the marginal distribution on $X_S$. 
Let $q^{X_1}, \ldots, q^{X_k}$ be normalised
probability distributions on the respective alphabets.
For each $(x_1, \ldots, x_k) \in X_{[k]}$, 
let $\rho^M_{x_1, \ldots, x_k}$
be a subnormalised density matrix on $M$. 
The classical quantum {\em control state} is now defined as
\[
\rho^{X_{[k]} M} :=
\sum_{(x_1, \ldots, x_k) \in X_{[k]}}
p^{X_{[k]}}(x_1, \ldots, x_k) \ketbra{x_1, \ldots, x_k}^{X_{[k]}}
\otimes \rho^M_{x_1, \ldots, x_k}.
\]
Suppose $\supp(p^{X_i}) \leq \supp(q^{X_i})$.
For any subset $S \subseteq [k]$, let 
$q^{X_S} := \times_{s \in S} q^{X_s}$.
Let $A_1, \ldots, A_k$ be positive integers.
For each $i \in [k]$, let
$x_i^{(A_i)} := (x_i(1), \ldots, x_i(A_i))$ denote a $|A_i|$-tuple
of elements from  $X_i$.
Denote the $A_i$-fold product alphabet 
$X_i^{A_i} := X_i^{\times A_i}$, and the product probability distribution
$
q^{X_i^{A_i}} :=
(q^{X_i})^{\times A_i}.
$
For any collection of tuples
$x_i^{(A_i)} \in X_i^{A_i}$, $i \in [k]$, we
define the {\em sample average covering state}
\[
\sigma^M_{x_1^{(A_1}, \ldots, x_k^{(A_k)}} :=
(A_1 \cdots A_k)^{-1}
\sum_{a_1 = 1}^{A_1} \cdots \sum_{a_k = 1}^{A_k}
\frac{p^{X_{[k]}}(x_1(a_1), \ldots, x_k(a_k))}
     {q^{X_1}(x_1(a_1)) \cdots q^{X_k}(x_k(a_k))}
\rho^M_{x_1(a_1), \ldots, x_k(a_k)},
\]
where the fraction term above represents the `change of measure' from
the product probability distribution $q^{X_{[k]}}$ to the
joint probability distribution $p^{X_{[k]}}$.
Suppose for each non-empty subset $\{\} \neq S \subseteq [k]$,
\begin{eqnarray*}
\sum_{s \in S} \log A_s 
& > &
D_2^\epsilon(\rho^{X_S M} \| q^{X_S} \otimes \rho^M) + 
\log \epsilon^{-2}.
\end{eqnarray*}
Then,
\[
\E_{x_1^{(A_1)}, \ldots, x_k^{(A_k)}}[
\|\sigma^M_{x_1^{(A_1)}, \ldots, x_k^{(A_k)}} - \rho^M\|_1
] < 
2(3^k - 1) \epsilon (\Tr \rho).
\]
where the expectation is taken over independent choices of tuples
$x_i^{(A_i)}$ from the distributions $q^{X_i^{A_i}}$, $i \in [k]$.
\end{proposition}

The term `in expectation' in the expression `covering lemma in expectation'
means that the sample average state 
$\sigma^M_{x_1^{(A_1)}, \ldots, x_k^{(A_k)}}$ in 
Proposition~\ref{prop:smoothcovering} is close to 
the marginal state $\rho^M$ in expectation
over the independent choices of tuples
$x_i^{(A_i)}$ from the iid distributions $q^{X_i^{A_i}}$, $i \in [k]$.
As remarked after the proof of Lemma~\ref{lem:smoothCMGcovering} above,
the telescoping proof also tells us that 
exactly the same result of Proposition~\ref{prop:smoothcovering} holds 
if the expectation is taken over
the independent choices, over $i \in [k]$, of tuples $x_i^{(A_i)}$ 
taken from a pairwise
independent distribution on the alphabet $X_i^{A_i}$, if the marginal
distribution on any single copy of $X_i$ equals $q^{X_i}$. 
For independent choices of tuples
$x_i^{(A_i)}$ from the iid distributions $q^{X_i^{A_i}}$, one can hence
hope for some better viz. if there is a
{\em covering lemma in concentration}. 
By this, we ask if the sample average 
state is close to the marginal state in concentration i.e. whether
for an overwhelming probability of tuples
$x_i^{(A_i)}$ chosen independently from the iid distributions 
$q^{X_i^{A_i}}$, $i \in [k]$, the sample average state
$\sigma^M_{x_1^{(A_1)}, \ldots, x_k^{(A_k)}}$ is close to 
the marginal state $\rho^M$. The paper \cite{Sen:MatrixChernoff} proves
just such a result.

Anshu and Jain \cite{decoupling_convexsplit} give 
`more economical' unipartite 
non-smooth convex split lemmas that require smaller amounts of catalytic 
ancillas
and  are more efficient / closer to classical than the standard
unipartite convex split lemma of \cite{Jain:convexsplit}. These
lemmas require only pairwise independence amongst certain
(different) marginals of
quantum states involved in their statements. Being unipartite, they
can be easily smoothed. Unfortunately, Anshu and Jain do not give any
multipartite versions for their economical convex split lemmas. We
remark that one can easily prove fully smooth natural multipartite
versions of those economical convex split lemmas via telescoping.
Our versions of those lemmas are stated in terms of 
$D_2^\epsilon$ instead of $D_\infty$ as in the original paper.

We note that the notion of pairwise independence referred to
in Section~\ref{sec:CMGcovering} in the context of the CMG covering
problem, or in the context of the soft covering lemma
Proposition~\ref{prop:smoothcovering}
is different from the notion of pairwise independence in
Anshu and Jain's work \cite{decoupling_convexsplit}. 
All the smooth and multipartite generalisations of Anshu and
Jain's results that we prove by telescoping 
continue to have the pairwise independence
in their sense \cite{decoupling_convexsplit}. 
\begin{theorem}
\label{thm:smoothefficientconvexsplit1}
Let $\rho^{XYM}$ be a subnormalised density matrix.
Let $\mu$ denote the completely mixed state over a Hilbert space. 
Let $A$, $B$ be positive integers.
Let $\{V_a^{X X^{'4}}\}_{a \in [A]}$ be an appropriate unitary
family as defined in \cite{decoupling_convexsplit} (which is constructed
from a family of pairwise independent hash functions and a 1-design
of unitaries), where $X'$ is another Hilbert space of the same dimension
as $X$ and $X^{'4} := (X')^{\otimes 4}$. Similarly, let
$\{V_b^{Y Y^{'4}}\}_{b \in [B]}$ be an appropriate unitary family.
Define the {\em convex split state}
\[
\sigma^{X X^{'4} Y Y^{'4} M} :=
(AB)^{-1} 
\sum_{a=1}^A \sum_{b=1}^B
(V_a \otimes V_b \otimes \one^M) \circ
(\rho^{X Y M} \otimes \mu^{X^{'4} Y^{'4}}),
\]
and the fully decoupled state 
\[
\tau^{X X^{'4} Y Y^{'4} M} :=
\mu^X \otimes \mu^Y \otimes \rho^M \otimes \mu^{X^{'4} Y^{'4}}.
\]
Suppose 
\begin{eqnarray*}
\log A 
& > &
D_2^\epsilon(\rho^{XM} \| \mu^X \otimes \rho^M) + \log \epsilon^{-2} \\
\log B 
& > &
D_2^\epsilon(\rho^{YM} \| \mu^Y \otimes \rho^M) + \log \epsilon^{-2} \\
\log A  + \log B
& > &
D_2^\epsilon(\rho^{XYM} \| \mu^X \otimes \mu^Y \otimes \rho^M) + 
\log \epsilon^{-2}.
\end{eqnarray*}
Then,
$
\|\sigma^{X X^{'4} Y Y^{'4} M} - 
  \tau^{X X^{'4} Y Y^{'4} M}\|_1 < 16 \epsilon (\Tr \rho).
$
\end{theorem}
\begin{theorem}
\label{thm:smoothefficientconvexsplit2}
Let $\rho^{XYM}$ be a subnormalised density matrix.
Let $\mu$ denote the completely mixed state over a Hilbert space. 
Let $A$, $B$ be positive integers and $Q_X$, $Q_Y$ be two Hilbert spaces 
of dimension 2 each.
Let $\{U_a^{X X^{'3} Q_X^2}\}_{a \in [A]}$ be an appropriate unitary
family as defined in \cite{decoupling_convexsplit} (which is a 
family of so-called `classical' unitaries derived from operations
over a certain cyclic group),
where $X'$ is another Hilbert space of the same dimension
as $X$ and $X^{'3} := (X')^{\otimes 3}$. Similarly, let
$\{V_b^{Y Y^{'3} Q_Y^2}\}_{b \in [B]}$ be an appropriate unitary family.
Define the {\em convex split state}
\[
\sigma^{X X^{'3} Q_X^2 Y Y^{'3} Q_Y^2 M} :=
(AB)^{-1} 
\sum_{a=1}^A \sum_{b=1}^B
(V_a \otimes V_b \otimes \one^M) \circ
(\rho^{X Y M} \otimes \mu^{X^{'3} Q_X Y^{'3} Q_Y} 
 \otimes \ketbra{0}^{Q_X Q_Y}),
\]
and the fully decoupled state 
\[
\tau^{X X^{'3} Q_X^2 Y Y^{'3} Q_Y^2 M} :=
\mu^X \otimes \mu^Y \otimes \rho^M \otimes 
\mu^{X^{'3} Q_X^2 Y^{'3} Q_Y^2}.
\]
Suppose 
\begin{eqnarray*}
\log A 
& > &
D_2^\epsilon(\rho^{XM} \| \mu^X \otimes \rho^M) + \log \epsilon^{-2} \\
\log B 
& > &
D_2^\epsilon(\rho^{YM} \| \mu^Y \otimes \rho^M) + \log \epsilon^{-2} \\
\log A  + \log B
& > &
D_2^\epsilon(\rho^{XYM} \| \mu^X \otimes \mu^Y \otimes \rho^M) + 
\log \epsilon^{-2}.
\end{eqnarray*}
Then,
$
\|\sigma^{X X^{'3} Q_X^2 Y Y^{'3} Q_Y^2 M} - 
  \tau^{X X^{'3} Q_X^2 Y Y^{'3} Q_Y^2 M}\|_1 < 16 \epsilon (\Tr \rho).
$
\end{theorem}

Our fully smooth convex split result in 
Proposition~\ref{prop:smoothconvexsplit} easily allows us to prove
a fully smooth one shot inner bound for the generalised Slepian
Wolf problem of Anshu, Jain and Warsi \cite{anshu:slepianwolf}. Such
a result was open till now. Anshu, Jain and Warsi stated their one
shot inner bound in terms of non-smooth $D_\infty$ and smooth
$D^\epsilon_H$. 
Our fully smooth result is stated in terms of 
smooth $D_2^\epsilon$ and smooth $D^\epsilon_H$. 
\begin{theorem}
\label{thm:slepianwolf}
Let $0 < \epsilon < 1$. Let Alice hold registers $A M$, Bob 
$B N$ and Charlie $C$. The three parties share a quantum state 
which is purified by a reference system $R$. Let 
$\ket{\Psi}^{R (AM) (BN) C}$ denote that state. Alice and Charlie,
as well as Bob and Charlie share independent unlimited prior entanglement.
Alice and Bob perform local quantum computations on their respective
registers. Alice sends a classical message of $R_1$ bits to
Charlie. Bob send a classical message of $R_2$ bits to Charlie. Charlie
does a local quantum operation on his registers conditioned on the
two received messages. At the end of this protocol, Alice, Bob,
Charlie and Reference share a (not necessarily pure) quantum state
$\Phi^{R A B (CMN)}$ where Charlie now possesses the registers
$C$, $M$, $N$. Let $\sigma^M$, $\omega^N$ be density matrices.

Suppose the rates satisfy the following inequalities:
\begin{eqnarray*}
R_1 
& > &
D^\epsilon_2(\Psi^{R AM B} \| \Psi^{R A B} \otimes \sigma^M) -
D^{\epsilon^2}_H(\Psi^{AM} \| \Psi^{A} \otimes \sigma^M) + 
\log \epsilon^{-2}, \\
R_2 
& > &
D^\epsilon_2(\Psi^{R A BN} \| \Psi^{R A B} \otimes \omega^N) -
D^{\epsilon^2}_H(\Psi^{BN} \| \Psi^{B} \otimes \omega^N) + 
\log \epsilon^{-2}, \\
R_1 + R_2 
& > &
D^\epsilon_2(\Psi^{R AM BN} \| \Psi^{R A B} \otimes \sigma^M \otimes
	                       \omega^N) \\
& &
~~~~~~
{} -
D^{\epsilon^2}_H(\Psi^{AM} \| \Psi^{A} \otimes \sigma^M) - 
D^{\epsilon^2}_H(\Psi^{BN} \| \Psi^{B} \otimes \omega^N) + 
\log \epsilon^{-2}. 
\end{eqnarray*}
Then,
$
\|\Phi^{R A B (CMN)} - \Psi^{R (AM) (BN) C} \|_1 < 21 \epsilon.
$
\end{theorem}
The above fully smooth one shot result easily implies the following
result in the asymptotic iid limit, stated in terms of Shannon
mutual information $I(\cdot : \cdot)$ only. Anshu et al. 
\cite{anshu:slepianwolf}
had proved the asymptotic iid result below using ad hoc techniques
assuming that $C$ was trivial i.e. Charlie did not have any prior 
information. They had left the general asymptotic iid case open. 
The improved convex split lemma of 
Cheng et al. \cite{Cheng:convexsplit} does allow one to prove the
asymptotic iid result below without any restriction on $C$, but 
still fails to prove
the smooth one shot result in Theorem~\ref{thm:slepianwolf} above.
\begin{corollary}
\label{cor:slepianwolf}
In the asymptotic iid limit, the communcation rates per channel
use satisfy the following inner bound:
\begin{eqnarray*}
R_1 
& > &
I(R A B : M)_\Psi - I(A : M)_\Psi, \\
R_2 
& > &
I(R A B : N)_\Psi - I(B : N)_\Psi, \\
R_1 + R_2 
& > &
I(R A B : M N)_\Psi + I(M : N)_\Psi
- I(A : M)_\Psi - I(B : N)_\Psi.
\end{eqnarray*}
\end{corollary}

The smooth multipartite convex split lemma also allows us to easily
prove the natural pentagonal one shot inner bound for sending 
private classical information over a wiretap QMAC. This problem was
studied earlier in \cite{Chakraborty:wiretapQMAC}, but the authors could 
only obtain a subset of the pentagonal rate region due to the lack
of a smooth multipartite convex split lemma. Nevertheless, their
one shot region was good enough to lead to the pentagonal inner bound
in the asymptotic iid limit. Of course, obtaining the pentagonal
region in the asymptotic iid limit is trivial given our pentagonal
one shot inner bound.

Let $\cN^{A B \rightarrow C E}$ denote a wiretap QMAC from two
senders Alice, Bob to a single legitimate receiver Charlie and
an eavesdropper Eve. Alice, Bob would like to send classical
messages $m \in 2^{R_1}$, $n \in 2^{R_2}$ respectively to Charlie
by using the channel $\cN$ in such a way that Eve gets almost
no information about $(m,n)$. For this, they adopt the following
strategy: 
Alice fixes a classical alphabet $X$ and a classical to
quantum encoding $x \mapsto \alpha^A_x$ from classical symbols in
$X$ to quantum states in the input Hilbert space $A$.
Similarly, Bob fixes a classical alphabet $Y$ and a classical to
quantum encoding $y \mapsto \beta^B_x$ from classical symbols in
$Y$ to quantum states in the input Hilbert space $B$. 
In order to send her message $m$, Alice encodes $m$ into a certain
mixture of states from the ensemble $\{\alpha^A_x\}_x$. Bob adopts
a similar strategy; he encodes $n$  into a certain mixture of states
from the ensemble $\{\beta^B_y\}_y$.

Good encoding strategies for Alice and Bob were described in the
work of \cite{Chakraborty:wiretapQMAC}. For doing so, the authors 
defined a cq `control state' as follows: Define a new
`timesharing' alphabet $Q$. Put a joint probability distribution
on $Q \times X \times Y$ of the form $p(q) p(x|q) p(y|q)$ i.e.
the distributions on $X$ and $Y$ are independent conditioned on 
any $q \in Q$.
Consider the cq `control' state
\[
\sigma^{Q X Y C E} :=
\sum_{q \in Q} \sum_{x \in X} \sum_{y \in Y}
p(q) p(x|q) p(y|q) \ketbra{q,x,y}^{QXY} \otimes 
\cN^{A B\rightarrow CE} (\alpha^A_x \otimes \beta^B_y).
\]
The {\em smooth conditional hypothesis testing mutual information} can now
be defined in the standard fashion e.g.
\begin{eqnarray*}
I^\epsilon_H(XY : C | Q)_\sigma 
& := &
D^\epsilon_H(\sigma^{QXYC} \| \bar{\sigma}^{XY:C|Q}), \\
\bar{\sigma}^{XY:C|Q} 
& := &
\sum_{qxy} p(q)p(x|q) p(y|q) \ketbra{q,x,y}^{QXY} \otimes
\sigma^C_q, \\
\sigma^C_q 
& := &
\sum_{x,y} p(x|q) p(y|q) \sigma^C_{xy}, ~~
\sigma^C_{xy} 
 := 
\Tr_E[\cN^{A B\rightarrow CE} (\alpha^A_x \otimes \beta^B_y)].
\end{eqnarray*}
Then \cite{Chakraborty:wiretapQMAC} showed that a random codebook 
construction together with standard obfuscation
techniques using the control state gives rise to a private coding
scheme for $\cN$. Analysing their same scheme using the fully smooth
multipartite convex split lemma gives the following theorem.
\begin{theorem}
\label{thm:wiretapQMAC}
Let the rates $R_1$, $R_2$ satisfy the following inequalities.
\begin{eqnarray*}
R_1
& < & 
I^\epsilon_H(X : YC | Q)_\sigma - 
I^\epsilon_\infty(X : E | Q)_\sigma + \log \epsilon^{-1}, \\
R_2
& < & 
I^\epsilon_H(Y : XC | Q)_\sigma - 
I^\epsilon_\infty(Y : E | Q)_\sigma + \log \epsilon^{-1}, \\
R_1 + R_2
& < & 
I^\epsilon_H(XY : C | Q)_\sigma - 
I^\epsilon_\infty(XY : E | Q)_\sigma + \log \epsilon^{-1}.
\end{eqnarray*}
Then,
\begin{eqnarray*}
\E_{m,n}[\mbox{probability Charlie decodes $(m,n)$ incorrectly}] 
& < &
50 \sqrt{\epsilon}
~~~~~~~ \mbox{accurate transmission}, \\
\E_{m,n}[\|\sigma^E_{m,n} - \sigma^E\|_1] 
& < & 
16 \sqrt{\epsilon}
~~~~~~~ \mbox{high privacy},
\end{eqnarray*}
where $\E_{m,n}[\cdot]$ denotes the expectation over a uniform
choice of message pair $(m,n) \in [2^{R_1}] \times [2^{R_2}]$,
$\sigma^E_{m,n}$ denotes Eve's state when the message pair 
$(m,n)$ is sent, and $\sigma^E$ denotes the marginal of the control
state on $E$.
\end{theorem}
The asymptotic iid limit is now immediate.
\begin{corollary}
\label{cor:wiretapQMAC}
In the asymtotic iid limit of a wiretap QMAC, the rate pairs 
per channel use satisfying the following inequalities are achievable.
\begin{eqnarray*}
R_1
& < & 
I(X : YC | Q)_\sigma - 
I(X : E | Q)_\sigma, \\
R_2
& < & 
I(Y : XC | Q)_\sigma - 
I(Y : E | Q)_\sigma, \\
R_1 + R_2
& < & 
I(XY : C | Q)_\sigma - 
IXY : E | Q)_\sigma.
\end{eqnarray*}
\end{corollary}

\section{Fully smooth multipartite decoupling via telescoping}
\label{sec:decoupling}
In this section, we will show how to extend the telescoping arguments
of Colomer and Winter \cite{Colomer:decoupling} in order to obtain
a fully smooth multipartite decoupling theorem. Below, we let
$\I$ denote the identity superoperator on a system. The EPR state is
defined as
\[
\Phi^{A'A} := 
|A|^{-1} \sum_{a_1,a_2 = 1}^{|A|} \ket{a_1, a_1}^{A'A}\bra{a_2, a_2},
\]
where $A'$ is another Hilbert space of the same dimension as $A$.
We remark that though Proposition~\ref{prop:decoupling} is stated in terms
of expectation over the Haar measure on unitaries as is the convention,
it continues to hold without change for expectation over a perfect
2-design of unitaries, and with minor additive terms for expectation
over an approximate 2-design of unitaries. We believe that our 
proof of Proposition~\ref{prop:decoupling} is simpler than the treatment
in \cite{Colomer:decoupling}, while being more general in the sense
that it leads to a fully smooth result.

Below, the smooth conditional R\'{e}nyi entropy is defined by
\[
H_2^\epsilon(A|R)_\rho :=
-D_2^\epsilon(\rho^{AR} \| \one^A \otimes \rho^R), ~~
H_\infty^\epsilon(A|R)_\rho :=
-D_\infty^\epsilon(\rho^{AR} \| \one^A \otimes \rho^R).
\]
Various alternate definitions of the smooth conditional 
R\'{e}nyi-2 entropy
have been given in the literature, but there all equivalent up to 
minor tweaks in the smoothing parameter $\epsilon$ and small dimension
independent additive terms polynomial in $\log \epsilon^{-1}$.
Also, $H_2^\epsilon(A|R)_\rho \approx H_\infty^\epsilon(A|R)_\rho$
\cite{Jain:minimax}.
\begin{proposition}
\label{prop:decoupling}
Let $\cT^{A_1 A_2 \rightarrow B}$ a completely positive trace 
non increasing superoperator with Choi state defined by
$
\tau^{A'_1 A'_2 B} :=
(\cT^{A_1 A_2 \rightarrow B} \otimes \I^{A'_1 A'_2})
(\Phi^{A'_1 A_1} \otimes \Phi^{A'_2 A_2}).
$.
Let $U^{A_1}$, $U^{A_2}$ be two unitary matrices on their respective
Hilbert spaces.
Let $\rho^{A_1 A_2 R}$ be a subnormalised density matrix. Define
the {\em unitarily perturbed state} 
\[
\sigma^{B R}(U^{A_1}, U^{A_2}) :=
(\cT \otimes \I^R)(
(U^{A_1} \otimes U^{A_2} \otimes \one^R) \circ \rho^{A_1 A_2 R}).
\]
Define the positive quantities
\[
D_1 := (1 - |A_1|^{-2})^{-1/2},
D_2 := (1 - |A_2|^{-2})^{-1/2},
D_{1,2} := 2 (1 - |A_1|^{-2})^{-1/2} (1 - |A_2|^{-2})^{-1/2}
\]
Let $0 < \epsilon < 1$.
Then,
\begin{eqnarray*}
\lefteqn{
\E_{U^{A_1} U^{A_2}}[
\|\sigma^{B R}(U^{A_1}, U^{A_2}) - \tau^B \otimes \rho^R\|_1]
} \\
& \leq &
16 \epsilon + 
2 D_1 \cdot 
2^{-\frac{1}{2} H_2^\epsilon(A_1 | R)_\rho
   -\frac{1}{2} H_2^\epsilon(A'_1 | B)_\tau} \\
& &
~~~~~
{} +
2 D_2 \cdot 
2^{-\frac{1}{2} H_2^\epsilon(A_2 | R)_\rho
   -\frac{1}{2} H_2^\epsilon(A'_2 | B)_\tau} +
4 D_{1,2} \cdot 
2^{-\frac{1}{2} H_2^\epsilon(A_1 A_2 | R)_\rho
   -\frac{1}{2} H_2^\epsilon(A'_1 A'_2 | B)_\tau},
\end{eqnarray*}
where the expectation on the left is taken over independent choices
of unitaries $U^{A_1}$, $U^{A_2}$ from their respective Haar measures.
\end{proposition}
\begin{proof}
We follow the lines of the proof of the half smooth decoupling
theorem in \cite{Colomer:decoupling}. Define the matrices
\[
\sigma^{B R}(U^{A_2}) :=
\E_{U^{A_1}}[\sigma^{B R}(U^{A_1} U^{A_2})] =
(\cT \otimes \I^R)(
(\one^{A_1} \otimes U^{A_2} \otimes \one^R) \circ 
(\frac{\one^{A_1}}{|A_1|}) \otimes \rho^{A_2 R}).
\]
The matrices $\sigma^{B R}(U^{A_1})$, $\sigma^{BR}$ are defined similarly.
Note that 
$
\sigma^{BR} = \tau^B \otimes \rho^R.
$

We first set up the telescoping sum
\begin{equation}
\label{eq:decoupling1}
\begin{array}{rcl}
\lefteqn{
\E_{U^{A_1} U^{A_2}}[
\|\sigma^{B R}(U^{A_1}, U^{A_2}) - \tau^B \otimes \rho^R\|_1]
} \\
& \leq &
\E_{U^{A_1} U^{A_2}}[
\|\sigma^{B R}(U^{A_1}, U^{A_2}) - 
  \sigma^{BR}(U^{A_1}) -
  \sigma^{BR}(U^{A_2}) +
  \sigma^{BR}\|_1] \\
&  &
~~~
{} +
\E_{U^{A_1} U^{A_2}}[\|\sigma^{B R}(U^{A_1}) - \sigma^{BR}\|_1] +
\E_{U^{A_1} U^{A_2}}[\|\sigma^{B R}(U^{A_2}) - \sigma^{BR}\|_1].
\end{array}
\end{equation}

Let $\rho^{A_1 A_2 R}(3)$ be the subnormalised state
achieving the optimum in the definition of the third smooth R\'{e}nyi-2
conditional information term involving $\rho$ in the upper
bound of Proposition~\ref{prop:decoupling}.
Similarly, let $\tau^{A'_1 A'_2 B}(3)$ be the subnormalised state
achieving the optimum in the definition of the third smooth R\'{e}nyi-2
conditional information term involving $\tau$ in the upper
bound of Proposition~\ref{prop:decoupling}. We similarly define
$\rho^{A_1 R}(1)$, $\tau^{A'_1 B}(1)$, 
$\rho^{A_2 R}(2)$, $\tau^{A'_2 B}(2)$.
These states are close to the respective marginals of $\rho^{A_1 A_2 R}$
and $\tau^{A'_1 A'_2 B}$.
\[
\begin{array}{c}
\|\rho^{A_1 A_2 R}(3) - \rho^{A_1 A_2 R}\|_1 \leq \epsilon, ~
\|\rho^{A_1 R}(1) - \rho^{A_1 R}\|_1 \leq \epsilon, ~
\|\rho^{A_2 R}(2) - \rho^{A_2 R}\|_1 \leq \epsilon, \\
\|\tau^{A'_1 A'_2 B}(3) - \tau^{A'_1 A'_2 B}\|_1 \leq \epsilon, ~
\|\tau^{A'_1 B}(1) - \tau^{A'_1 B}\|_1 \leq \epsilon, ~
\|\tau^{A'_2 B}(2) - \tau^{A'_2 B}\|_1 \leq \epsilon.
\end{array}
\]
Let $T^{A_1 A_2 \rightarrow B}(3)$ be the completely positive superoperator
obtained by taking the Choi preimage of $\tau^{A'_1 A'_2 B}(3)$. Note
that no general statement can be made regarding how 
$T^{A_1 A_2 \rightarrow B}(3)$ 
affects the trace of positive semidefinite matrices.
Define $T^{A_1 \rightarrow B}(1)$,
$T^{A_2 \rightarrow B}(2)$ similarly.
Define 
\begin{equation}
\label{eq:decoupling2}
\begin{array}{c}
\sigma^{B R}(U^{A_1}, U^{A_2})(3) :=
(\cT(3) \otimes \I^R)(
(U^{A_1} \otimes U^{A_2} \otimes \one^R) \circ \rho^{A_1 A_2 R}(3)), \\
\sigma^{B R}(U^{A_1})(3) :=
(\cT(3) \otimes \I^R)(
((U^{A_1} \otimes \one^R) \circ \rho^{A_1 R}(3)) \otimes 
(\frac{\one^{A_2}}{|A_2|})), \\
\sigma^{B R}(U^{A_2})(3) :=
(\cT(3) \otimes \I^R)(
((U^{A_2} \otimes \one^R) \circ \rho^{A_2 R}(3)) \otimes 
(\frac{\one^{A_1}}{|A_1|})), \\
\sigma^{B R}(3) :=
(\cT(3) \otimes \I^R)(
(\frac{\one^{A_1}}{|A_1|}) \otimes (\frac{\one^{A_2}}{|A_2|})
\otimes \rho^R(3)), \\
\sigma^{B R}(U^{A_1})(1) :=
(\cT(1) \otimes \I^R)(
(U^{A_1} \otimes \one^R) \circ \rho^{A_1 R}(1)), \\ 
\sigma^{B R}(1) :=
(\cT(1) \otimes \I^R)(
(\frac{\one^{A_1}}{|A_1|}) \otimes \rho^R(1)), \\
\sigma^{B R}(U^{A_2})(2) :=
(\cT(2) \otimes \I^R)(
(U^{A_2} \otimes \one^R) \circ \rho^{A_2 R}(2)), \\
\sigma^{B R}(2) :=
(\cT(2) \otimes \I^R)(
(\frac{\one^{A_2}}{|A_2|}) \otimes \rho^R(2)).
\end{array}
\end{equation}

We now follow a proof argument originally
given in \cite{Smooth_decoupling}. Let
\[
\tau^{A'_1 A'_2 B}(3) - \tau^{A'_1 A'_2 B}
= \omega^{A'_1 A'_2 B}_+ - \omega^{A'_1 A'_2 B}_-,
\]
where $\omega^{A'_1 A'_2 B}_+$, $\omega^{A'_1 A'_2 B}_-$ are positive
semidefinite matrices with orthogonal support satisfying
\[
\Tr[\omega^{A'_1 A'_2 B}_+] + \Tr[\omega^{A'_1 A'_2 B}_-] = 
\|\tau^{A'_1 A'_2 B}(3) - \tau^{A'_1 A'_2 B}\|_1 \leq \epsilon.
\]
Let $\cO^{A'_1 A'_2 \rightarrow B}_+$, $\cO^{A'_1 A'_2 \rightarrow B}_-$
be completely positive superoperators corresponding to the Choi preimagess
of $\omega^{A'_1 A'_2 B}_+$, $\omega^{A'_1 A'_2 B}_-$. Then,
\[
\cT^{A'_1 A'_2 \rightarrow B}(3) - \cT^{A'_1 A'_2 \rightarrow B}
= \cO^{A'_1 A'_2 \rightarrow B}_+ - \cO^{A'_1 A'_2 \rightarrow B}_-.
\]
Then,
\begin{eqnarray*}
\lefteqn{
\E_{U^{A_1} U^{A_2}}[
\|(\cT(3) \otimes \I^R)(
    (U^{A_1} \otimes U^{A_2} \otimes \one^R) \circ \rho^{A_1 A_2 R}(3)) -
  (\cT \otimes \I^R)(
    (U^{A_1} \otimes U^{A_2} \otimes \one^R) \circ \rho^{A_1 A_2 R}(3))
\|_1]
} \\
& = &
\E_{U^{A_1} U^{A_2}}[
\|((\cT(3) - \cT) \otimes \I^R)(
    (U^{A_1} \otimes U^{A_2} \otimes \one^R) \circ \rho^{A_1 A_2 R}(3))
\|_1] \\
& = &
\E_{U^{A_1} U^{A_2}}[
\|((\cO_+ - \cO_-) \otimes \I^R)(
    (U^{A_1} \otimes U^{A_2} \otimes \one^R) \circ \rho^{A_1 A_2 R}(3))
\|_1] \\
& \leq &
\E_{U^{A_1} U^{A_2}}[
\|(\cO_+ \otimes \I^R)(
    (U^{A_1} \otimes U^{A_2} \otimes \one^R) \circ \rho^{A_1 A_2 R}(3))
\|_1 \\
&  &
~~~~~~~~~~~~~~~~~
{} +
\|(\cO_- \otimes \I^R)(
    (U^{A_1} \otimes U^{A_2} \otimes \one^R) \circ \rho^{A_1 A_2 R}(3))
\|_1] \\
& =    &
\E_{U^{A_1} U^{A_2}}[
\Tr[(\cO_+ \otimes \I^R)(
(U^{A_1} \otimes U^{A_2} \otimes \one^R) \circ \rho^{A_1 A_2 R}(3))]] \\
&  &
~~~
{} +
\E_{U^{A_1} U^{A_2}}[
\Tr[(\cO_- \otimes \I^R)(
(U^{A_1} \otimes U^{A_2} \otimes \one^R) \circ \rho^{A_1 A_2 R}(3))]] \\
& =    &
\Tr[(\cO_+ \otimes \I^R)(
(\frac{\one^{A_1 A_2}}{|A_1| |A_2|}) \otimes \rho^{R}(3))] +
\Tr[(\cO_- \otimes \I^R)(
(\frac{\one^{A_1 A_2}}{|A_1| |A_2|}) \otimes \rho^{R}(3))] \\
& =    &
\Tr[\omega^B_+ \otimes \rho^R(3)] +
\Tr[\omega^B_- \otimes \rho^R(3)] 
\;\leq\;
\Tr[\omega^B_+] + \Tr[\omega^B_-] 
\;\leq\;
\epsilon.
\end{eqnarray*}
By the triangle inequality for the trace distance, we get
\begin{equation}
\label{eq:decoupling3}
\begin{array}{rcl}
\lefteqn{
\E_{U^{A_1} U^{A_2}}[
\|\sigma^{BR}(U^{A_1}, U^{A_2})(3) - \sigma^{BR}(U^{A_1}, U^{A_2})\|_1]
} \\
& \leq &
\E_{U^{A_1} U^{A_2}}[
\|(\cT(3) \otimes \I^R)(
    (U^{A_1} \otimes U^{A_2} \otimes \one^R) \circ \rho^{A_1 A_2 R}(3)) \\ 
& &
~~~~~~~~~~~~~~~~~~
{} -
\|(\cT \otimes \I^R)(
    (U^{A_1} \otimes U^{A_2} \otimes \one^R) \circ \rho^{A_1 A_2 R}(3))
\|_1] \\
& &
~~~
{} +
\E_{U^{A_1} U^{A_2}}[
\|(\cT \otimes \I^R)(
    (U^{A_1} \otimes U^{A_2} \otimes \one^R) \circ \rho^{A_1 A_2 R}(3)) \\
& &
~~~~~~~~~~~~~~~~~~~~~~~~
{} -
\|(\cT \otimes \I^R)(
    (U^{A_1} \otimes U^{A_2} \otimes \one^R) \circ \rho^{A_1 A_2 R})
\|_1] \\
& \leq &
\epsilon +
\E_{U^{A_1} U^{A_2}}[
\|(U^{A_1} \otimes U^{A_2} \otimes \one^R)
  \circ (\rho^{A_1 A_2 R}(3) - \rho^{A_1 A_2 R})\|_1]  \\
&  =   &
\epsilon +
\|\rho^{A_1 A_2 R}(3) - \rho^{A_1 A_2 R}\|_1 
\;\leq\;
2\epsilon,
\end{array}
\end{equation}
where we used the fact that $\cT^{A_1 A_2 \rightarrow B}$ is a trace
non-increasing completely positive superoperator in the second
inequality. By triangle inequality of trace distance, we
also similarly get
\begin{equation}
\label{eq:decoupling4}
\begin{array}{c}
\E_{U^{A_1} U^{A_2}}[
\|\sigma^{BR}(U^{A_1})(3) - \sigma^{BR}(U^{A_1})\|_1] \leq 2\epsilon, \\
\E_{U^{A_1} U^{A_2}}[
\|\sigma^{BR}(U^{A_2})(3) - \sigma^{BR}(U^{A_2})\|_1] \leq 2\epsilon, \\
\E_{U^{A_1} U^{A_2}}[
\|\sigma^{BR}(3) - \sigma^{BR}\|_1] \leq 2\epsilon, \\
\E_{U^{A_1} U^{A_2}}[
\|\sigma^{BR}(U^{A_1})(1) - \sigma^{BR}(U^{A_1})\|_1] \leq 2\epsilon, \\
\E_{U^{A_1} U^{A_2}}[
\|\sigma^{BR}(U^{A_2})(2) - \sigma^{BR}(U^{A_2})\|_1] \leq 2\epsilon, \\
\E_{U^{A_1} U^{A_2}}[
\|\sigma^{BR}(1) - \sigma^{BR}\|_1] \leq 2\epsilon, ~~
\|\sigma^{BR}(2) - \sigma^{BR}\|_1] \leq 2\epsilon.
\end{array}
\end{equation}

Putting Equations~\ref{eq:decoupling3}, \ref{eq:decoupling4} back into
Equation~\ref{eq:decoupling1}, and using the triangle inequality,
we get
\begin{equation}
\label{eq:decoupling5}
\begin{array}{rcl}
\lefteqn{
\E_{U^{A_1} U^{A_2}}[
\|\sigma^{B R}(U^{A_1}, U^{A_2}) - \tau^B \otimes \rho^R\|_1]
} \\
& \leq &
16 \epsilon +
\E_{U^{A_1} U^{A_2}}[
\|\sigma^{B R}(U^{A_1}, U^{A_2})(3) - 
  \sigma^{BR}(U^{A_1})(3) -
  \sigma^{BR}(U^{A_2})(3) +
  \sigma^{BR}(3)\|_1] \\
&  &
~~~
{} +
\E_{U^{A_1} U^{A_2}}[\|\sigma^{B R}(U^{A_1})(1) - \sigma^{BR}(1)\|_1] +
\E_{U^{A_1} U^{A_2}}[\|\sigma^{B R}(U^{A_2})(2) - \sigma^{BR}(2)\|_1].
\end{array}
\end{equation}

We now use Cauchy-Schwarz inequality Fact~\ref{fact:matrixCauchySchwarz} 
with weighting matrix
$\tau^B \otimes \rho^R$ to upper bound each term in 
telescoping sum of Equation~\ref{eq:decoupling5}
in terms of the Schatten-$\ell_2$ norm. Define
\[
\tsigma^{B R}(U^{A_1}, U^{A_2})(3) := 
(\tau^B \otimes \rho^R)^{-1/4} \circ
\sigma^{B R}(U^{A_1}, U^{A_2})(3).
\]
Other weighted terms like e.g.
$\tsigma^{B R}(U^{A_1})(1)$ are defined similarly. We get,
\begin{equation}
\label{eq:decoupling6}
\begin{array}{rcl}
\lefteqn{
\E_{U^{A_1} U^{A_2}}[
\|\sigma^{B R}(U^{A_1}, U^{A_2}) - \tau^B \otimes \rho^R\|_1]
} \\
& \leq &
16 \epsilon +
\E_{U^{A_1} U^{A_2}}[
\|\tsigma^{B R}(U^{A_1}, U^{A_2})(3) - 
  \tsigma^{BR}(U^{A_1})(3) -
  \tsigma^{BR}(U^{A_2})(3) +
  \tsigma^{BR}(3)\|_2] \\
&  &
~~~
{} +
\E_{U^{A_1} U^{A_2}}[\|\tsigma^{B R}(U^{A_1})(1) - \tsigma^{BR}(1)\|_2] +
\E_{U^{A_1} U^{A_2}}[\|\tsigma^{B R}(U^{A_2})(2) - \tsigma^{BR}(2)\|_2].
\end{array}
\end{equation}

We now upper bound a Schatten-$\ell_2$ term as follows. For example,
\begin{eqnarray*}
\lefteqn{
\left(
\E_{U^{A_1} U^{A_2}}[
\|\tsigma^{B R}(U^{A_1}, U^{A_2})(3) - 
  \tsigma^{BR}(U^{A_1})(3) -
  \tsigma^{BR}(U^{A_2})(3) +
  \tsigma^{BR}(3)\|_2] 
\right)^2
} \\
& \leq  &
\E_{U^{A_1} U^{A_2}}[
\|\tsigma^{B R}(U^{A_1}, U^{A_2})(3) - 
  \tsigma^{BR}(U^{A_1})(3) -
  \tsigma^{BR}(U^{A_2})(3) +
  \tsigma^{BR}(3)\|_2^2].
\end{eqnarray*}
Let $\hA_1$, $\hA_2$, $\hA'_1$, $\hA'_2$, $\hB$, $\hR$ 
be new Hilbert spaces of the same dimensions as
$A_1$, $A_2$, $A'_1$, $A'_2$, $B$, $R$. 
Define 
\[
\tcT(3)^{A_1 A_2 \rightarrow B} := 
(\tau^B)^{-1/4} \circ
\cT(3)^{A_1 A_2 \rightarrow B},
\trho(3)^{A_1 A_2 R} :=
(\rho^B)^{-1/4} \circ
\rho(3)^{A_1 A_2 R}.
\]
The corresponding Choi state will be denote by
$\ttau^{A'_1 A'_2 B}(3)$.
Let $\htcT(3)^{\hA_1 \hA_2 \rightarrow \hB}$ be the same
superoperator as $\tcT(3)$ but on the new Hilbert spaces. The corresponding
Choi state will be denoted by $\httau^{\hA'_1 \hA'_2 \hB}(3)$.
Similarly, let $\htrho^{\hA_1 \hA_2 \hR}(3)$ be the same quantum
state as $\trho(3)$ but on the new Hilbert spaces. We define the 
Hermitian matrix
\[
\btrho^{A_1 A_2 R}(3) :=
\trho^{A_1 A_2 R}(3) -
(\frac{\one^{A_1}}{|A_1|}) \otimes \trho^{A_2 R}(3) -
(\frac{\one^{A_2}}{|A_2|}) \otimes \trho^{A_1 R}(3) +
(\frac{\one^{A_1}}{|A_1|}) \otimes (\frac{\one^{A_2}}{|A_2|}) \otimes
\trho^{R}(3).
\]
The Hermitian matrix $\hbtrho^{\hA_1 \hA_2 \hR}(3)$ is defined similarly.
Let the so-called {\em swap operator} be defined by $F^{R \hR}$;
swap operators on other pairs of Hilbert spaces are defined accordingly.
We handle the $\|\cdot\|_2^2$ in the above expression 
using the so-called {\em tensorisation cum swap trick} going back
to Dupuis \cite{decoupling}:
\begin{eqnarray*}
\lefteqn{
\E_{U^{A_1} U^{A_2}}[
  \|\tsigma^{B R}(U^{A_1}, U^{A_2})(3) - \tsigma^{BR}(U^{A_1})(3) -
    \tsigma^{BR}(U^{A_2})(3) + \tsigma^{BR}(3)\|_2^2]
} \\
& = &
\E_{U^{A_1} U^{A_2}}[
      \Tr[((\tcT(3) \otimes \I^R)(
              (U^{A_1} \otimes U^{A_2} \otimes \one^R) \circ 
	      \btrho^{A_1 A_2 R}(3)))^2]] \\
& = &
\E_{U^{A_1} U^{A_2}}[
      \Tr[(((\tcT(3) \otimes \I^R)(
               (U^{A_1} \otimes U^{A_2} \otimes \one^R) \circ 
	       \btrho^{A_1 A_2 R}(3))) \\
&  &
~~~~~~~~~~~~~~~~~~~
           {} \otimes
          ((\htcT(3) \otimes \I^{\hR})(
	      (U^{\hA_1} \otimes U^{\hA_2} \otimes \one^{\hR}) \circ 
	      \hbtrho^{\hA_1 \hA_2 \hR}(3)))) 
	        (F^{B \hB} \otimes F^{R \hR})]]\\
& = &
\E_{U^{A_1} U^{A_2}}[
      \Tr[((((U^{A_1} \otimes U^{A_2})^\dag \otimes 
	     (U^{\hA_1} \otimes U^{\hA_2})^\dag) \circ {} \\
&  &
~~~~~~~~~~~~~~~~~
	   ((\tcT(3)^\dag \otimes \htcT(3)^\dag)(F^{B \hB}))) 
	      \otimes F^{R \hR})
	   (\btrho^{A_1 A_2 R}(3) \otimes \hbtrho^{\hA_1 \hA_2 \hR}(3))]]\\
& = &
\Tr[((\E_{U^{A_1} U^{A_2}}[
          (((U^{A_1} \otimes U^{A_2})^\dag \otimes 
	    (U^{\hA_1} \otimes U^{\hA_2})^\dag) \circ {} \\
&  &
~~~~~~~~~~~~~~~~~
	    ((\tcT(3)^\dag \otimes \htcT(3)^\dag)(F^{B \hB})))]) 
           \otimes F^{R \hR})
	(\btrho^{A_1 A_2 R}(3) \otimes \hbtrho^{\hA_1 \hA_2 \hR}(3))].
\end{eqnarray*}

As argued in \cite{Chakraborty:simultaneous}, the expression
\begin{eqnarray*}
\lefteqn{
\E_{U^{A_1} U^{A_2}}[
    ((U^{A_1} \otimes U^{A_2})^\dag \otimes 
     (U^{\hA_1} \otimes U^{\hA_2})^\dag) M^{A_1 A_2 \hA_1 \hA_2}]
} \\
& = &
\alpha_{00} \one^{A_1 \hA_1} \otimes \one^{A_2 \hA_2} +
\alpha_{01} \one^{A_1 \hA_1} \otimes F^{A_2 \hA_2} +
\alpha_{10} F^{A_1 \hA_1} \otimes \one^{A_2 \hA_2} +
\alpha_{11} F^{A_1 \hA_1} \otimes F^{A_2 \hA_2}
\end{eqnarray*}
for any $M^{A_1 A_2 \hA_1 \hA_2}$, where the complex numbers
$\alpha_{00}$, $\alpha_{01}$, $\alpha_{10}$, $\alpha_{11}$ are functions
of $M$. Now observe that
\begin{eqnarray*}
\lefteqn{
\Tr[((\one^{A_1 \hA_1} \otimes F^{A_2 \hA_2}) \otimes F^{R \hR})
    (\btrho^{A_1 A_2 R}(3) \otimes \hbtrho^{\hA_1 \hA_2 \hR}(3))]
} \\
& = &
\Tr[(F^{A_2 \hA_2} \otimes F^{R \hR})
    (\btrho^{A_2 R}(3) \otimes \hbtrho^{\hA_2 \hR}(3))] 
\;=\;
0,
\end{eqnarray*}
since
$
\btrho^{A_2 R}(3) :=
\Tr_{A_1}[\btrho^{A_1 A_2 R}(3)] = 0.
$
This is an example of the {\em mean zero property} of the Colomer
and Winter telescoping decomposition \cite{Colomer:decoupling}.
Similarly arguing, we see that the only possibly non-zero term 
that can survive in the above expression is
\begin{eqnarray*}
\lefteqn{
\Tr[((\E_{U^{A_1} U^{A_2}}[
          (((U^{A_1} \otimes U^{A_2})^\dag \otimes 
	    (U^{\hA_1} \otimes U^{\hA_2})^\dag) \circ {} 
} \\
&  &
~~~~~~~~~~~~~~~~~
	    ((\tcT(3)^\dag \otimes \htcT(3)^\dag)(F^{B \hB})))]) 
           \otimes F^{R \hR})
	(\btrho^{A_1 A_2 R}(3) \otimes \hbtrho^{\hA_1 \hA_2 \hR}(3))] \\
& = &
\alpha_{11} 
\Tr[((F^{A_1 \hA_1} \otimes F^{A_2 \hA_2}) \otimes F^{R \hR})
    (\btrho^{A_1 A_2 R}(3) \otimes \hbtrho^{\hA_1 \hA_2 \hR}(3))] 
\;=\;
\alpha_{11} \|\btrho^{A_1 A_2 R}(3)\|_2^2 \\
& \leq &
4 \alpha_{11} 
(\|\trho^{A_1 A_2 R}(3)\|_2^2 +
 |A_1|^{-1} \|\trho^{A_2 R}(3)\|_2^2 +
 |A_2|^{-1} \|\trho^{A_1}(3)\|_2^2 +
 |A_1|^{-1} |A_2|^{-1} \|\trho^{R}(3)\|_2^2) \\
& \leq &
16 \alpha_{11} (\|\trho^{A_1 A_2 R}(3)\|_2^2,
\end{eqnarray*}
where $\alpha_{11}$ is the complex number arising from the matrix
\[
M^{A_1 A_2 \hA_1 \hA_2} =
((\tcT(3)^\dag)^{B \rightarrow A_1 A_2} \otimes 
 (\htcT(3)^\dag)^{\hB \rightarrow \hA_1 \hA_2})(F^{B \hB}),
\]
the first inequality follows from Cauchy Schwarz, and the second
inequality follows from the fact that
\[
|X|^{-1} \|\rho^{Y}\|_2^2 
\leq \|\rho^{XY}\|_2^2 \leq
|X| \|\rho^{Y}\|_2^2
\]
for any positive semidefinite matrix $\rho^{XY}$ \cite{decoupling}.
If the mean zero property were not used above as in decoupling papers
like \cite{Chakraborty:simultaneous} 
prior to \cite{Colomer:decoupling}, we would not get a clean
upper bound in terms of only $\|\trho^{A_1 A_2 R}(3)\|_2^2$.
There would have been extra terms like $\|\trho^{A_1 R}(3)\|_2^2$,
$\|\trho^{A_2 R}(3)\|_2^2$ etc., on which there seems to be
no handle without using simultaneous smoothing. 

Following \cite{Chakraborty:simultaneous},
we can easily show that
\begin{eqnarray*}
\alpha_{11} 
& = &
\frac{|A_1||A_2|}{(|A_1|^2 - 1)(|A_2|^2 - 1)}
(\|\ttau^{B}(3)\|_2^2 - |A_1| \|\ttau^{A'_1 B}(3)\|_2^2 -
 |A_2| \|\ttau^{A'_2 B}(3)\|_2^2 + 
 |A_1| |A_2| \|\ttau^{A'_1 A'_2 B}(3)\|_2^2) \\
& \leq &
\frac{|A_1||A_2|}{(|A_1|^2 - 1)(|A_2|^2 - 1)}
(|A_1||A_2| \|\ttau^{A'_1 A'_2 B}(3)\|_2^2 + 
 |A_1| |A_2| \|\ttau^{A'_1 A'_2 B}(3)\|_2^2) 
\;=\;
D_{1,2}^2 \|\ttau^{A'_1 A'_2 B}(3)\|_2^2.
\end{eqnarray*}
We have thus shown that
\begin{equation}
\label{eq:decoupling7}
\begin{array}{rcl}
\lefteqn{
\E_{U^{A_1} U^{A_2}}[
\|\tsigma^{B R}(U^{A_1}, U^{A_2})(3) - 
  \tsigma^{BR}(U^{A_1})(3) -
  \tsigma^{BR}(U^{A_2})(3) +
  \tsigma^{BR}(3)\|_2^2]
} \\
& \leq &
16 D_{1,2}^2 \cdot
\|\tilde{\tau}^{A'_1 A'_2 B}(3)\|_2^2 \cdot
\|\tilde{\rho}^{A_1 A_2 R}(3)\|_2^2 
\;=\;
16 D_{1,2}^2 \cdot
2^{-H_2^\epsilon(A'_1 A'_2 | B)_\tau} \cdot
2^{-H_2^\epsilon(A_1 A_2 | R)_\rho}.
\end{array}
\end{equation}

At this point we remark that the reason we get a fully smooth
decoupling lemma is because we used arguments from
Szehr et al. \cite{Smooth_decoupling} to smooth the Choi state
$\tau^{A'_1 A'_2 B}$ to $\tau^{A'_1 A'_2 B}(3)$. Colomer and Winter
did not use that argument, which left them with the non-smooth
Choi state $\tau^{A'_1 A'_2 B}$. They did use $\rho^{A_1 A_2 R}(3)$
though, which ultimately lead them to their 
half smooth decoupling lemma.

Similarly, we can show
\begin{equation}
\label{eq:decoupling8}
\begin{array}{c}
\E_{U^{A_1} U^{A_2}}[
\|\tsigma^{B R}(U^{A_1})(1) - 
  \tsigma^{BR}(1)\|_2^2] \leq
4 D_1^2 \cdot 2^{-H_2^\epsilon(A'_1 | B)_\tau} \cdot
2^{-H_2^\epsilon(A_1 | R)_\rho}, \\
\E_{U^{A_1} U^{A_2}}[
\|\tsigma^{B R}(U^{A_2})(2) - 
  \tsigma^{BR}(2)\|_2^2] \leq
4 D_2^2 \cdot 2^{-H_2^\epsilon(A'_2 | B)_\tau} \cdot
2^{-H_2^\epsilon(A_2 | R)_\rho}.
\end{array}
\end{equation}
Putting Equations~\ref{eq:decoupling7}, \ref{eq:decoupling8}
back into Equation~\ref{eq:decoupling6}, we get
\begin{eqnarray*}
\lefteqn{
\E_{U^{A_1} U^{A_2}}[
\|\sigma^{B R}(U^{A_1}, U^{A_2}) - \tau^B \otimes \rho^R\|_1]
} \\
& \leq &
16 \epsilon +
4 D_{1,2} \cdot
2^{-\frac{1}{2} H_2^\epsilon(A'_1 A'_2 | B)_\tau
   -\frac{1}{2} H_2^\epsilon(A_1 A_2 | R)_\rho} \\
&  &
~~~
{} +
2 D_1 \cdot 
2^{-\frac{1}{2} H_2^\epsilon(A'_1 | B)_\tau
   -\frac{1}{2} H_2^\epsilon(A_1 | R)_\rho}  +
2 D_2 \cdot 
2^{-\frac{1}{2} H_2^\epsilon(A'_2 | B)_\tau
   -\frac{1}{2} H_2^\epsilon(A_2 | R)_\rho},
\end{eqnarray*}
finishing the proof of the theorem.
\end{proof}

A similar decoupling result can be stated for any constant number 
of parties, with a completely analogous proof. We state the result
below for completeness.
\begin{theorem}
\label{thm:decouplingmanyparties}
Let $k$ be a positive integer. 
Let $\cT^{A_1 \cdots A_k \rightarrow B}$ a completely positive trace 
non increasing superoperator with Choi state defined by
\[
\tau^{A'_1 \cdots A'_k B} :=
(\cT^{A_1 \cdots A_k \rightarrow B} \otimes \I^{A'_1 \cdots A'_k})
(\Phi^{A'_1 A_1} \otimes \cdots \otimes \Phi^{A'_k A_k}).
\]
Let $U^{A_1}, \ldots, U^{A_k}$ be $k$ unitary matrices on their respective
Hilbert spaces.
Let $\rho^{A_1 \cdots A_k R}$ be a subnormalised density matrix. Define
the {\em unitarily perturbed state} 
\[
\sigma^{B R}(U^{A_1}, \ldots, U^{A_k}) :=
(\cT \otimes \I^R)(
(U^{A_1} \otimes \cdots \otimes U^{A_k} \otimes \one^R) \circ 
\rho^{A_1 \cdots A_k R}).
\]
For any subset $S \subseteq [k]$, let 
$A_S := \otimes_{s \in S} A_s$, 
$A'_S := \otimes_{s \in S} A'_s$, and
$
D_{S} := 2^{|S|-1} \prod_{s\in S} (1 - |A_s|^{-2})^{-1/2}.
$
Let $0 < \epsilon < 1$.
Then,
\begin{eqnarray*}
\lefteqn{
\E_{U^{A_1} U^{A_2}}[
\|\sigma^{B R}(U^{A_1}, \ldots, U^{A_k}) - \tau^B \otimes \rho^R\|_1]
} \\
& \leq &
2(3^k - 1) \epsilon + 
\sum_{\{\} \neq S \subseteq [k]}
2^{|S|} D_S \cdot 
2^{-\frac{1}{2} H_2^\epsilon(A_S | R)_\rho
   -\frac{1}{2} H_2^\epsilon(A'_S | B)_\tau},
\end{eqnarray*}
where the expectation on the left is taken over independent choices
of unitaries $U^{A_1}, \ldots, U^{A_k}$ from their respective Haar 
measures.
\end{theorem}

It is now easy to apply Proposition~\ref{prop:decoupling} in 
order to prove the 
achievability of the natural one shot pentagonal inner bound for sending
quantum information over a QMAC under limited prior entanglement.
This was first conjectured in \cite{Chakraborty:simultaneous} but
remained open till now.
Let $\cN^{A' B' \rightarrow C}$ be a quantum multiple access channel
from two senders Alice, Bob to a single receiver Charlie.
Alice and Charlie are allowed to share $E_A$ EPR pairs
as prior entanglement. Similarly, Bob and Charlie are allowed to
independently share $E_B$ EPR pairs as prior entanglement. Alice
is provided with half of $Q_A$ EPR pairs as an independent quantum 
message to be sent using the channel. Similarly, Bob
is provided with half of $Q_B$ EPR pairs as an independent quantum 
message to be sent using the channel. Alice and Bob apply suitable
independent encoding isometries on their quantum messages and 
prior entanglement halves and send the outputs of the encoders as inputs
to the channel. Charlie applies a decoding quantum operation on the
channel output and his halves of the prior entanglements in order to
get decoded registers of $Q_A$ and $Q_B$ bits respectively. 
The rate quadruple 
$(Q_A,E_A,Q_B,E_B)$ is $\epsilon$-achievable for quantum message 
transmission with limited entanglement assistance through $\cN$ if,
after decoding, the joint state of Charlie and
the two reference systems of the two quantum messages is 
$\epsilon$-close to
$Q_A + Q_B$ EPR pairs in trace distance. Following the lines of
\cite{Chakraborty:simultaneous},  we can prove the following theorem.
\begin{theorem}
\label{thm:QMAC}
Let $0 < \epsilon < 1$.
Let $\cN^{A' B' \rightarrow C}$ be a QMAC.
Fix some pure `control' states $\Omega^{A A'}$ and $\Delta^{B B'}$,
where $A$, $B$ are fresh Hilbert spaces 
of the same dimensions as $A'$, $B'$. Let 
$U_\cN^{A' B' \rightarrow C E}$ denote
the isometric Stinespring dilation of $\cN$. Define the `control' state
\[
\omega^{A B C E} := 
(U_\cN^{A'B' \rightarrow CE} \otimes \one^{A B}) \circ 
(\Omega^{AA'} \otimes \Delta^{BB'}).
\]
Suppose the rates satisfy the following inequalities:
\begin{eqnarray*}
Q_A-E_A+Q_B-E_B 
& < & 
H_2^\epsilon(AB|E)_{\omega}+ \log \epsilon^{-4}, \\
Q_A-E_A 
& < &
H_2^\epsilon(A|E)_{\omega}+\log \epsilon^{-4}, \\
Q_B-E_B 
& < &
H_2^\epsilon(B|E)_{\omega}+\log \epsilon^{-4}, \\
Q_A+E_A 
& < &
H_2^\epsilon(A)_{\omega} + \log \epsilon^{-4}, \\
Q_B+E_B 
& < &
H_2^\epsilon(B)_{\omega} + \log \epsilon^{-4}.
\end{eqnarray*}
Then, the rate quadruple $(Q_A,E_A,Q_B,E_B)$ is 
$(29 \epsilon)$-achievable.
\end{theorem}

A similar theorem can be stated for sending quantum information over
a QMAC with many senders and limited entanglement assistance.
\begin{theorem}
\label{thm:QMACmanyparties}
Let $0 < \epsilon < 1$. Let $k$ be a positive integer.
Let $\cN^{A'_1 \cdots A'_k \rightarrow C}$ be a QMAC. For $i \in [k]$,
fix a pure `control' state $\Omega_i^{A_i A'_i}$. Let
$U_\cN^{A'_1 \cdots A'_k \rightarrow C E}$ denote
the isometric Stinespring dilation of $\cN$. Define the `control' state
\[
\omega^{A_1 \cdots A_k  C E} := 
(U_\cN^{A'_1 \cdots A'_k \rightarrow CE} \otimes \one^{A_1 \cdots A_k B}) 
\circ 
\bigotimes_{i=1}^k \Omega_i^{A_i A'_i}.
\]
For any non-empty subset $\{\} \neq S \subseteq [k]$, let the rates 
satisfy the following inequalities:
\begin{eqnarray*}
\sum_{s \in S} (Q_{A_s} - E_{A_s})
& < & 
H_2^\epsilon(A_S|E)_{\omega}+ \log \epsilon^{-4}.
\end{eqnarray*}
Also, suppose for all $i \in [k]$,
\[
Q_{A_i} + E_{A_i} <
H_2^\epsilon(A_i)_{\omega} + \log \epsilon^{-4}.
\]
Then, the rate tuple $(Q_{A_i}, E_{A_i})_{i \in [k]}$ is 
$((3(3^k - 1) + (2^k - 1) + k) \epsilon)$-achievable.
\end{theorem}

As an easy corollary of Theorem~\ref{thm:QMACmanyparties}, we prove 
the achievability
of the natural polyhedral rate region in the asymptotic iid limit. 
The same result in the iid limit for two senders was first shown in 
\cite{Chakraborty:ratesplitting} with a much more complicated argument
involving successive cancellation and rate splitting. However the
techniques in that paper could only prove the achievability of a 
subset of the natural pentagonal region in the one shot setting for the
two sender case.
Below, the notation 
\[
I(A > B)_\rho :=
H(B)_\rho - H(AB)_\rho
\]
denotes the {\em coherent information} from $A$ to $B$ under the
joint state $\rho$.
\begin{corollary}
\label{cor:QMACmanyparties}
Under the settings of Theorem~\ref{thm:QMACmanyparties} we get 
the following
achievability result per channel use in the asymptotic iid limit.
For any non-empty subset $\{\} \neq S \subseteq [k]$,
\begin{eqnarray*}
\sum_{s \in S} (Q_{A_s} - E_{A_s})
& < & 
I(A_S > C A_{\bar{S}})_\omega,
\end{eqnarray*}
where $\bar{S} := [k] - S$ is the complement of $S$.
Also for any $i \in [k]$,
$Q_{A_i}+E_{A_i} < H(A_i)_\omega.$
\end{corollary}

\section{Conclusion}
\label{sec:conclusion}
We have seen how the telescoping cum mean zero techniques of 
\cite{Cheng:convexsplit} and \cite{Colomer:decoupling}
can be simplified and extended in order to obtain
fully smooth one shot multipartite convex split and decoupling theorems
for the first timee.
We have noted that multipartite convex split becomes the multipartite
soft covering problem when restricted to classical quantum states.
For such states, we have observed how the telescoping technique
continues to hold when full independence amongst the classical
probability distributions is replaced by pairwise independence. 
We have also seen some economical
variants of multipartite convex split, as well as applications of smooth
multipartite convex split and decoupling to important source and 
channel coding problems in quantum Shannon theory.

Our work leads to natural avenues for further research. Finding more
applications of the telescoping and mean zero techniques, as well as
more limitations, is an important direction to pursue. The limitation
of pairwise independence is successfully addressed in the 
companion paper \cite{Sen:flatten} by using a sophisticated machinery
independent of telescoping. It will be interesting to 
see if the telescoping paradigm can be extended to slightly relax
the mean zero requirement. If such a thing is possible, it will be
an alternate way to address the pairwise independent limitation.

\section*{Acknowledgements}
I thank Rahul Jain for hosting me on a sabbatical visit
at the Centre for Quantum Technologies, National University
of Singapore, where most of this work was done. 
I also thank the staff at the Centre for their efficient hospitality.
This research is 
supported in part by the National Research Foundation, Singapore and 
A*STAR under the CQT Bridging Grant and the Quantum Engineering Programme 
Award number NRF2021-QEP2-02-P05.
I also acknowledge support of the Department of Atomic Energy, Government
of India, under project no. RTI4001.

\bibliography{telescoping}

\end{document}